\documentclass{llncs}
\pdfoutput=1

\usepackage{graphicx}

\usepackage{times}
\usepackage{amsmath}

\usepackage{verbatim}
\usepackage{epsfig}
\usepackage[iso]{umlaute}
\usepackage{amssymb}
\usepackage{graphicx}
\usepackage{bbm}
\usepackage{subfigure}
\usepackage{ifthen}
\usepackage{url}
\usepackage{algorithm}

\usepackage{algorithm}
\usepackage{algorithmic}

\newcommand{\DB}{\ensuremath{\mathcal{D}}}
\newcommand{\RR}{\ensuremath{\mathbb{R}}}

\newcommand{\eps}{\ensuremath{\varepsilon}}

\newcommand{\qmbr}{\ensuremath{Q^{\Box}}}

\begin{document}
\title{Inverse Queries For Multidimensional Spaces}
%\numberofauthors{2}

\author{
Thomas Bernecker\inst{1} \and Tobias Emrich\inst{1} \and Hans-Peter
Kriegel\inst{1} \and Nikos Mamoulis\inst{2} \and Matthias Renz\inst{1} \and
Shiming Zhang\inst{2} \and Andreas Züfle\inst{1}}
\institute{
Institute for Informatics,
    Ludwig-Maximilians-Universität~München\\ Oettingenstr. 67, D-80538 München,
    Germany\\
    \{bernecker,emrich,kriegel,renz,zuefle\}\\@dbs.ifi.lmu.de
\and Department
    of Computer Science, University of Hong Kong\\ Pokfulam Road, Hong Kong \\
    \{nikos,smzhang\}@cs.hku.hk
}
\maketitle
\begin{abstract}
Traditional spatial queries return, for a given
query object $q$, all database objects that satisfy a given
predicate, such as epsilon range and $k$-nearest neighbors.
This paper defines and studies {\em inverse} spatial queries,
which, given
a subset of database objects $Q$ and a query predicate, return all objects
which, if used as query objects with the predicate, contain $Q$ in
their result. We first show a straightforward solution for
answering inverse spatial queries for any query predicate. Then,
we propose a filter-and-refinement framework that can be used to
improve efficiency. We show how to apply this framework on a
variety of inverse queries, using appropriate space pruning
strategies. In particular, we propose solutions for  inverse
epsilon range queries, inverse $k$-nearest neighbor queries, and
inverse skyline queries. Our experiments show that our framework
is significantly more efficient than naive approaches.
% and
%can be computed in reasonable time and that our
%filter-and-refinement framework significantly boosts
%performance.
%
\end{abstract}
\section{Introduction}

Recently, a lot of interest has grown for {\em reverse} queries,
which take as input an object $o$ and find the queries which have
$o$ in their result set. A characteristic example is the reverse
$k$NN query \cite{KorMut00,TaoPapLia04}, whose objective is to
find the query objects (from a given dataset) that have a given
input object in their $k$NN set. In such an operation the roles of
the query and data objects are reversed; while the $k$NN query
finds the {\em data} objects which are the nearest neighbors of a
given {\em query} object, the reverse query finds the objects
which, if used as queries, return a given data object in their
result. Besides $k$NN search, reverse queries have also been
studied for other spatial and multidimensional search problems,
such as top-$k$ search \cite{VlaDouKotNor10} and dynamic skyline
\cite{LiaChe08}. Reverse queries mainly find application in data
analysis tasks; e.g., given a product find the customer searches
that have this product in their result. \cite{KorMut00} outlines a
wide range of such applications (including business impact
analysis, referral and recommendation systems, maintenance of
document repositories).

In this paper, we generalize the concept of reverse queries. We
note that the current definitions take as input a {\em single}
object. However, similarity queries such as $k$NN queries and
$\eps$-range queries may in general return more than one result.
Data analysts are often interested in the queries that include two
or more given objects in their result. Such information can be
meaningful in applications where only the result of a query can be
(partially) observed, but the actual query object is not known.
For example consider an online shop selling a variety of different
products stored in a database \DB. The online shop may be
interested in offering a {\em package} of products $Q\subseteq
\DB$ for a special price. The problem at hand is to identify
customers which are interested in all items of the package, in
order to direct an advertisement to them. We assume that the
preferences of registered customers are known. First, we need to
define a predicate indicating whether a user is interested in a
product. A customer may be interested in a product if
\begin{itemize}
\item the distance between the product's features and the
customer's preference is less than a threshold $\eps$.

\item the product is contained in the set of his $k$ favorite
items, i.e., the $k$-set of product features closest to the user's
preferences.

\item the product is contained in the customer's dynamic skyline,
i.e., there is no other product that better fits the customer's
preferences in every possible way.
\end{itemize}
%Thus, we require to find all customer preferences to which all
%products $q\in Q$ are similar.
Therefore, we want to identify customers $r$, such that the query
on \DB\ with query object $r$, using one of the query predicates
above, contains $Q$ in the result set. More specifically, consider
a set $\DB\in\RR^d$ as a database of $n$ objects and let
$d(\cdot)$ denote the Euclidean distance in $\RR^d$. Let
${\mathcal P}(q)$ be a query on $\DB$ with predicate ${\mathcal
P}$ and query object $q$.
\begin{definition}\label{def:inv} %[Inverse ${\mathcal P}$ Query ($I{\mathcal P}$Q)]
An inverse ${\mathcal P}$ query (I${\mathcal P}$Q) computes for a
given set of query objects $Q\subseteq\DB$ the set of points
$r\in\RR^d$ for which $Q$ is in the ${\mathcal P}$ query result;
formally:
$$
I{\mathcal P}Q=\{r\in\RR^d:Q\subseteq{\mathcal P}(r))\}
$$
\end{definition}
Simply speaking, the result of the {\em general} inverse query is
the subset of the space defined by all objects $r$ for which all
$Q$-objects are in  ${\mathcal P}(r)$. Special cases of the query
are
\begin{itemize}
\item The mono-chromatic inverse ${\mathcal P}$ Query, for which
the result set is a subset of $\DB$. \item The bi-chromatic
inverse ${\mathcal P}$ Query, for which the result set is a subset
of a given database $\DB'\subseteq\RR^d$.
\end{itemize}

In this paper, we study the inverse versions of three common query
types in spatial and multimedia databases as follows.

%\begin{figure}[t]
%  \centering
%  \includegraphics[width=0.2\textwidth]{../figures/intro_fig.pdf}
%\vspace{-3mm}
%  \caption{Examples of inverse queries.}
%  \label{fig:introex}
%\end{figure}

\begin{figure}[h]
    \centering
    \subfigure[$I\eps \mbox{-}RQ$.]{
        \label{fig:introex1}
        \includegraphics[width=0.2\textwidth]{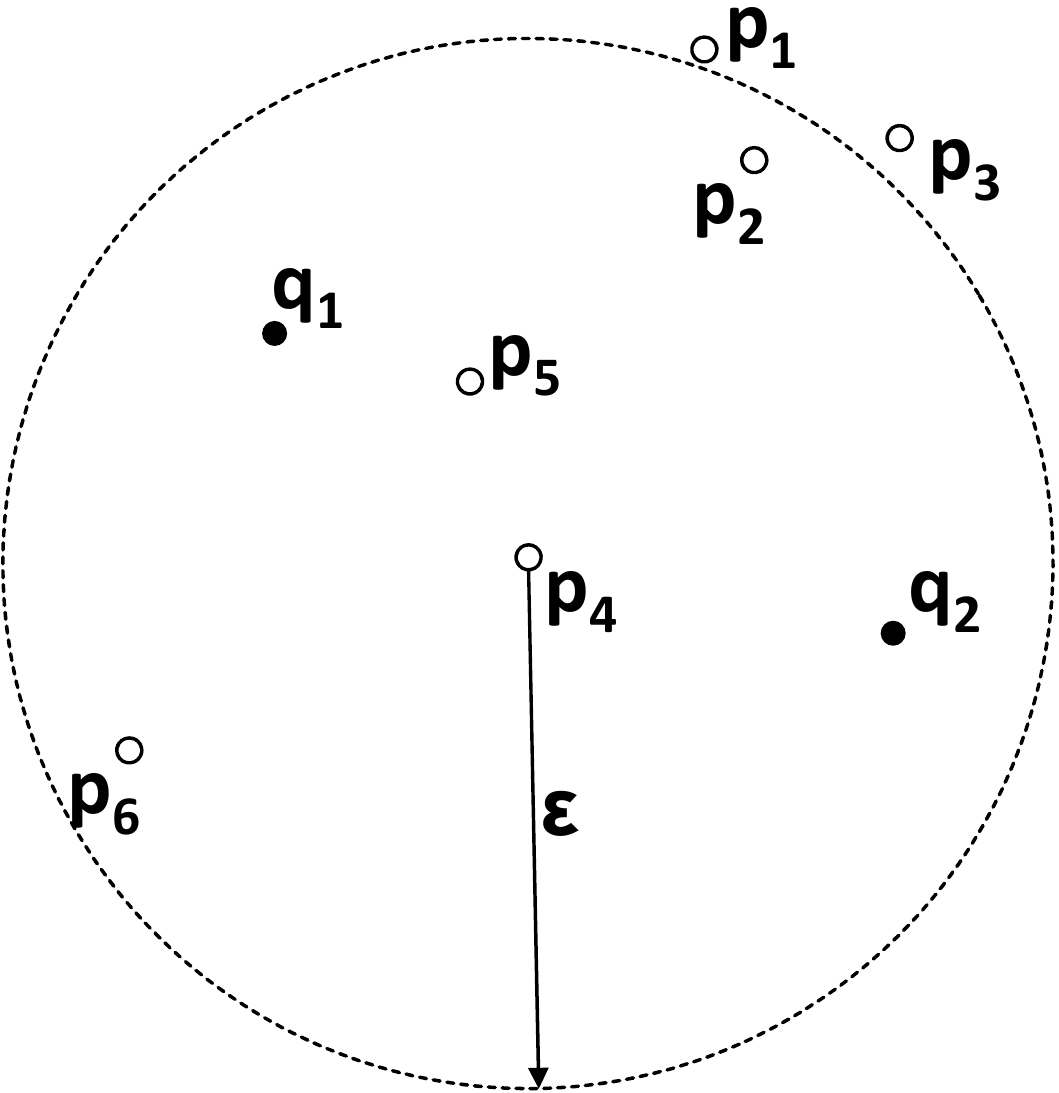}
    }\hfill
    \subfigure[$Ik\mbox{-}NNQ$, $k=3$.]{
        \label{fig:introex2}
        \includegraphics[width=0.2\textwidth]{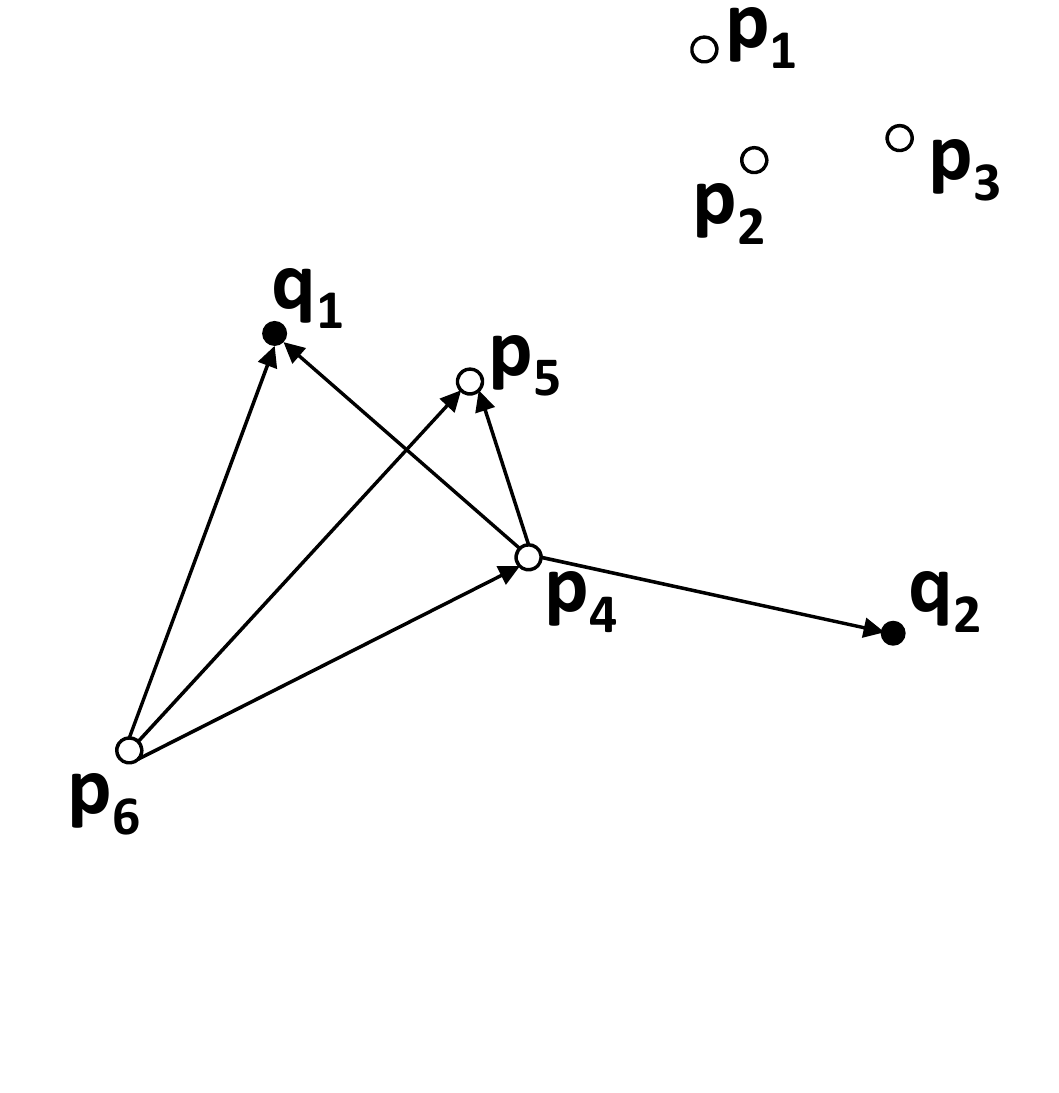}
    }
\vspace{-3mm}
    \caption{Examples of inverse queries.}
    \label{fig:introex}
\end{figure}

{\bf Inverse $\eps$-Range Query ($I\eps \mbox{-}RQ$).}
%The inverse
%$\eps$-range query computes for a given set of query objects
%$Q\subseteq\DB$ all objects $o\in\DB$ for which $Q$ is among the
%$\eps$-range query result, formally:
%$$
%I\eps \mbox{-}RQ(Q) = \{o\in\DB: Q\subseteq \eps \mbox{-}RQ(o)\},
%$$
%where $\eps\mbox{-}RQ(o)$ denotes the maximum subset of $\DB$ for
%which $\forall p\in \DB:d(o,p)\le \eps$ (i.e., a range query in
%$\DB$ around $o$ with distance $\eps$).
The inverse $\eps$-Range query returns all objects which have a
sufficiently low distance to all query objects. For a {\em
bi-chromatic} sample application of this type of query, consider a
movie database containing a large number of movie records. Each
movie record contains features such as humor, suspense, romance,
etc. Users of the database are represented by the same attributes,
describing their preferences. We want to create a recommendation
system that recommends to users movies that are sufficiently
similar to their preferences (i.e., distance less than $\eps$).
Now, assume that a group of users, such as a family,  want to
watch a movie together;
 a bi-chromatic $I\eps \mbox{-}RQ$ query will recommend movies
which are similar to {\em all} members of the family. For a
mono-chromatic case example, consider the set $Q=\{q_1,q_2\}$ of
query objects of Figure \ref{fig:introex1} and the set of database
points $\DB=\{p_1,p_2,\cdots,p_{6}\}$. If the range $\eps$ is as
illustrated in the figure, the result of the  $I\eps\mbox{-}RQ(Q)$
is $\{p_2,p_4,p_5\}$ (e.g., $p_1$ is dropped because
$d(p_1,q_2)>\eps$). %A sample application of this type of query is
%a pizza company that wants to send special advertisements to
%residents in order to win them over another set of companies.
%for which it would pay off to send special
%advertisement for their restaurant in order to win them over. This
%would especially hold for residents that have this restaurant and
%at the same time one or more other restaurants of concurring pizza
%companies within their neighborhood.
%An inverse $\eps$-range query
%on the company and its competitors
%would retrieve the candidates to which the advertisement should be
%forwarded.

{\bf Inverse $k$-NN Query ($Ik\mbox{-}NNQ$).}
%The inverse $k$-nearest neighbor
%query ($Ik\mbox{-}NNQ$)  computes for a given set of query
%objects $Q\subseteq\DB$ the set of objects in $\DB$ having $Q$ in their
%$k$-nearest
%neighbor set. Formally:\\
%$$Ik\mbox{-}NNQ(Q) = \{o\in\DB : Q\subseteq k\mbox{-}NN(o)\},$$
%where the $k$-nearest
%neighbor set $k\mbox{-}NN(o)$ for $o$ is defined as the greatest subset of $\DB\setminus o$
%containing at least $k$ objects, such that:
%\begin{small}
%$$
%\forall p\in
%k\mbox{-}NN(o),\forall\hat{p}\in\DB-\{o\}-k\mbox{-}NN(o):d(o,p)<d(p,\hat{p}).
%$$
%\end{small}
The inverse $k$NN query returns the objects which have all query
points in their $k$NN set. For example, {\em mono-chromatic}
inverse $k$NN queries can be used to aid crime detection. Assume
that a set of households have been robbed in short succession and
the robber must be found. Assume that the robber will only rob
houses which are in his close vicinity, e.g. within the closest
hundred households. Under this assumption, performing an inverse
$100$NN query, using the set of robbed households as $Q$, returns
the set of possible suspects. A mono-chromatic inverse 3NN query
for $Q=\{q_1,q_2\}$ in Figure \ref{fig:introex2} returns
$\{p_4\}$. $p_6$, for example, is dropped, as $q_2$ is not
contained in the list of its 3 nearest neighbors.
%This, however,
%holds only for the mono-chromatic case, where $q_2$ is pruned by
%$q_1$.
%In the bi-chromatic case, $p_6$ is contained in the result
%set, as its 3 nearest neighbors in $\DB$ are $p_4$, $p_5$ and
%$p_2$. $q_1$ and $q_2$ belong to the results in $\DB'$.
%
%A sample application of the $Ik\mbox{-}NNQ$ query is robust
%routing of information in (wireless) sensor networks. Information
%that cannot be directly sent from one node, e.g. $q_1$ in our
%example, to another node $q_2$ (e.g. they are not directly
%reachable due to large distance or obstacles), has to be sent via
%one or more auxiliary nodes. A node that has both $q_1$ and $q_2$
%in its $k$-NN set (e.g., $p_4$ in our example) is a good candidate
%to serve as a mediator in this purpose.
%(reachable) neighbors of a node are
%promising candidates used to transmit information to non-directly
%reachable nodes. An inverse $k$-NN query taking the nodes that
%want to communicate with each other as query objects would report
%auxiliary nodes that qualify for a stable communication, like node
%$p_4$ in our example.

%Note that the inverse $k$-nearest neighbor query w.r.t. a single
%query object (i.e., $|Q|=1$) is equivalent to the \emph{reverse $k$-nearest
%neighbor query} already introduced in \cite{KorMut00}.

{\bf Inverse Dynamic Skyline Query ($I\mbox{-}DSQ$).}
%To define
%the dynamic skyline of a query object $q\in\DB$ we first need to
%introduce the concept of dynamic skyline domination. For a given
%query point $q$ and two database objects $o_1$ and $o_2$, we say
%that $o_1$ dominates $o_2$ w.r.t. $q$, denoted by $o_1\prec_q o_2$
%if and only if for each dimension $1\leq i\leq d$ it holds that
%$|q^i-o_1^i|\leq |q^i-o_2^i|$ and for at least one dimension
%$1\leq i\leq d$ it holds that $|q^i-o_1^i|<|q^i-o_2^i|$, where
%$q^i$ ($o_i^i$) denotes the projection of $q$ ($o_i$) to the
%$i$-th dimension. The dynamic skyline $DS(q)$ of a $q\in\DB$ is
%the largest subset $S \subseteq \DB - \{q\}$ for which it holds
%that $ \forall s\in S,\forall o\in\DB-\{q\}-S: s\prec_q o. $ In
%turn, the inverse dynamic skyline query computes for a given set
%of query objects $Q\subseteq\DB$ the set of objects having $Q$ in
%their dynamic skyline set. Formally:
%$$I\mbox{-}DSQ(Q) = \{o\in\DB : Q\subseteq DS(o)\}.$$
An inverse dynamic skyline query returns the objects, which have
all query objects in their dynamic skyline.  A sample application
for the {\em general} inverse dynamic skyline query is a product
recommendation problem: Assume there is a company, e.g. a photo
camera company, that provides its products via an internet portal.
The company wants to recommend to their customers products by
analyzing the web pages visited by them. The score function used
by the customer to rate the attributes of products is unknown.
However, the set of products that the customer has clicked on can
be seen as samples of products that he or she is interested in,
and thus, must be in the customers dynamic skyline. The inverse
dynamic skyline query can be used to narrow the space which the
customers preferences are located in. Objects which have all
clicked products in their dynamic skyline are likely to be
interesting to the customer.
%Note, that it is possible that the result space is empty. This
%means that there cannot be any further products that may interest
%the buyer, since each product is dominated by at at least on
%product that the customer has already clicked at, regardless of
%the exact customer's preferences.
In Figure \ref{fig:introex}, assuming that  $Q=\{q_1,q_2\}$ are
clicked products,  $I\mbox{-}DSQ(Q)$ includes $p_6$, since both
$q_1$ and $q_2$ are included in the dynamic skyline of $p_6$.

% infor $Q=\{q_1,q_2\}$ returns a partition
%of the space containing $\{p_6\}$.

%In Appendix \ref{app:applications} we propose further applications
%for inverse queries.

For simplicity, we focus on the mono-chromatic cases of the
respective query types (i.e., query points and objects are taken
from the same data set); however, the proposed techniques can also
be applied for the bi-chromatic (cf. Appendix \ref{app:bi}) and
the general case.
%A brief
%discussion of extending the proposed techniques to the
%bi-chromatic case can be found in Appendix \ref{?}:
%able to retrieve such products a costumer could be
%interested in which is very important for individual
%advertisement.

%Note that the inverse dynamic skyline query w.r.t. to a single
%query object is equivalent to the problem introduced in
%\cite{LiaChe08}.

%\subsection{Motivation}
\textbf{Motivation.} A naive way to process any inverse spatial
query  is to compute the corresponding reverse query for each
$q_i\in Q$ and then intersect these results. The problem of this
method is that running a reverse query for each $q_i$ multiplies
the complexity of the reverse query by $|Q|$ both in terms of
computational and I/O cost. Objects that are not shared in two or
more reverse queries in $Q$ are unnecessarily retrieved, while
objects that are shared by two or more queries are redundantly
accessed multiple times.
%Another possible solution is to compute the corresponding reverse
%query for a single $q_i\in Q$ and then run a query for each
%candidate in the reverse query result.
%As a consequence, the search space is large compared to
%the space covering the overall result and is highly redundant.
We propose a
%specialized methods for the three inverse queries under
%study, which avoid the redundant work of the naive approach as
%much as possible. Our solutions are based on a common
filter-refinement framework for inverse queries, which first
applies a  number of filters using the set of query objects $Q$ to
prune effectively objects which may not participate in the result.
Afterwards candidates are pruned by considering other database
objects. Finally, during a {\em refinement} step, the remaining
candidates are verified against the inverse query and the results
are output. Details of our framework are shown in Section
\ref{sec:InverseQueryFramework}. When applying our framework to
the three inverse queries under study, filtering and refinement
are sometimes integrated in the same algorithm, which performs
these steps in an iterative manner. Although for $I\eps
\mbox{-}RQ$ queries the application of our framework is
straightforward, for $Ik\mbox{-}NNQ$ and $I\mbox{-}DSQ$, we define
and exploit special pruning techniques that are novel compared to
the approaches used for solving the corresponding reverse queries.

{\bf Outline.} The rest of the paper is organized as follows. In
the next section we give an overview of the previous work which is
related to inverse query processing. Section
\ref{sec:InverseQueryFramework} describes our framework. In
Sections \ref{sec:ieps}-\ref{sec:idsq} we implement it on the
three  inverse spatial query types; we first briefly introduce the
pruning strategies for the single-query-object case and then show
how to apply the framework in order to handle the
multi-query-object case in an efficient way.
%In Section \ref{}, we provide extensions of our approaches for the
%bichromatic version of the queries.
Section \ref{sec:experiments} is an experimental evaluation and
Section \ref{sec:conclusion} concludes the paper.

%\begin{figure}[t]
%  \centering
%  \includegraphics[width=0.2\textwidth]{test.pdf}
%\vspace{-3mm}
%  \caption{TPL pruning ($k=1$).}
%  \label{fig:tpl-pruning}
%\end{figure}

\vspace{-2mm}
\section{Related Work}\label{sec:related}
The problem of supporting reverse queries efficiently, i.e. the
case where $Q$ only contains a single database object, has been
studied extensively. However, none of the proposed approaches is
directly extendable for the efficient support of inverse queries
when $|Q|>1$. First, there exists no related work on reverse
queries for the $\epsilon$-range query predicate. This is not
surprising since the the reverse $\epsilon$-range query is equal
to a (normal) $\epsilon$-range query. However, there exists a
large body of work for reverse $k$-nearest neighbor (R$k$NN)
queries.
%Existing
%approaches for R$k$NN search in Euclidean spaces can be classified
%as self-pruning approaches or mutual-pruning approaches.
Self-pruning approaches like the RNN-Tree \cite{KorMut00} and the
RdNN-tree \cite{YanLin01} operate on top of a spatial index, like
the R-tree. Their objective is to estimate the $k$NN distance of
each index entry $e$. If the $k$NN distance of $e$ is smaller than
the distance of $e$ to the query $q$, then $e$ can be pruned.
%Thereby, self-pruning approaches do
%not usually consider other entries (database points or index
%nodes) in order to estimate the $k$NN distance of an entry $e$,
%but simply precompute $k$NN distances of database points and
%propagate these distances to higher level index nodes.
These methods suffer from the high materialization and maintenance
cost of the $k$NN distances.
%Although
%such an approach could be extended to process inverse queries
%where $|Q|>1$, by simply pruning an entry if {\em any} query
%object has a distance greater than the $k$NN distance, the problem
%is that the $k$NN distances need to be materialized. Thus both
%approaches are limited to a fixed value of $k$ and cannot be
%generalized to answer queries with arbitrary  $k$.
%In addition the heavy precomputation of the $k$NN distances makes
%the self pruning approaches inapplicable for dynamic databases
%where the database is updated frequently.

Mutual-pruning approaches such as
\cite{StaAgrAbb00,SinFerTos03,TaoPapLia04} use other points to
prune a given index entry $e$. TPL \cite{TaoPapLia04} is the most
general and efficient approach. It uses
%any hierarchical tree-based index structure such as
an R-tree to compute a nearest neighbor ranking of the query point
$q$. The key idea is to iteratively construct Voronoi hyper-planes
around $q$ using the retrieved neighbors.
%Points and index entries
%that are beyond $k$ Voronoi hyper-planes w.r.t. $q$ can be pruned and
%need not to be considered for Voronoi construction anymore.
%The idea of this pruning is illustrated in Figure
%\ref{fig:tpl-pruning} for $k=1$. Entry $e$ can be pruned, because
%it is beyond the Voronoi hyper-plane between $q$ and candidate
%$o$, denoted by $\perp\!\!(q,o)$.
%For the general case where
%$k\geq 1$, $e$ can be pruned if $e$ is beyond $k$ hyper-planes
%w.r.t.\ all current candidates. If $e$ cannot be pruned, the entries
%in the node pointed by it are added to the $k$NN heap,
%or, if $e$ is already a database object, $e$ is a new
%candidate and the hyperplane $\perp\!\!(q,e)$ will be considered
%for pruning in the following. If the ranking queue is empty, the
%remaining candidate points must be refined, i.e.\ for each of
%these candidates, a $k$NN query must be launched.
TPL can be used for inverse $k$NN queries where $|Q|>1$, by simply
performing a reverse $k$NN query for each query point and then
intersecting the results (i.e., the brute-force approach).
%. This adaption of the TPL is one of our
%baselines that we are able to improve in terms of efficiency using
%the techniques of this paper.

For reverse dynamic skyline queries, \cite{DelSee07} proposed an
efficient solution, which first performs a filter-step, pruning
database objects that are globally dominated by some point in the
database. For the remaining points, a window query is performed in
a refinement step. In addition, \cite{LiaChe08} gave a solution
for reverse dynamic skyline computation on uncertain data. None of
these methods considers the case of $|Q|>1$, which is the focus of
our work.

In \cite{VlaDouKotNor10} the problem of reverse top-$k$ queries is
studied. A reverse top-$k$ query returns for a point $q$ and a
positive integer $k$, the set of linear preference functions
%query objects $Q$ so that for
for which $q$ is contained in their top-$k$ result. The authors
provide an efficient solution for the 2D case and discuss its
generalization to the multidimensional case, but do not consider
the case where $|Q|>1$. Although we do not study inverse top-$k$
queries in this paper, we note that it is an interesting subject
for future work.

%Another problem that may sound related at the first glance is the
%problem of inverse ranking of uncertain objects as defined in
%\cite{LianChen09}, which computes for a given query database object
%$q$ and a score function, all possible probabilistic ranks that
%$q$ may have. The inverse ranking definition, however, is irrelevant
%to the inverse queries that we study in this paper.

%\input{sec_naive_solution}
%\input{sec_framework}
%\input{sec_ISQP}
\section{Inverse Query (IQ) Framework}
\label{sec:InverseQueryFramework}

Our solutions for the three inverse queries under study are based
on a common framework consisting of the following
filter-refinement pipeline:

\subsubsection*{Filter 1: Fast Query Based Validation:} The first
component of the framework, called \emph{fast query based
validation}, uses the set of query objects $Q$ only to perform a
quick check on whether it is possible to have any result at all.
In particular, this filter verifies simple constraints that are
necessary conditions for a non-empty result. For example, for the
I$k$NN case, the result is empty if $|Q|>k$.
%Similar constraints exists for the other queries further details
%will be shown in Section \ref{}.
%If the query objects does not
%fulfill such constraints the result must be empty and consequently
%the query processing terminates by reporting a ``no result'' message
%to the user.

\subsubsection*{Filter 2: Query Based Pruning:}
\emph{Query based pruning} again uses the query objects only to
prune objects in $\DB$ which may not participate in the
 result. Unlike the simple first filter, here we
employ the topology of the query objects.

Filters 1 and 2 can be performed very fast because they do not
involve any database object except the query objects.
%Now we will
%add database objects for further filtering.

\subsubsection*{Filter 3: Object Based Pruning:} This filter, called
\emph{object based pruning}, is more advanced because it involves
database objects additional to the query objects. The strategy is
to access database objects in ascending order of their maximum
distance to any query point; formally:
$$
MaxDist(o,Q)=\max_{q\in Q}(d(e,q)).
$$
The rationale for this access order is that, given any query
object $q$, objects that are close to $q$ have more pruning power,
i.e., they are more likely to prune other objects w.r.t. $q$ than
objects that are more distant to $q$. To maximize the pruning
power, we prefer to examine objects that are close to all query
points first.

Note that the applicability of the  filters depends on the query.
\emph{Query based pruning} is applicable if the query objects
suffice to restrict the search space which holds for the inverse
$\eps$-range query and the inverse skyline query but not directly
for the inverse $k$NN query. In contrast,   the \emph{object based
pruning} filter is applicable for queries where database objects
can be used to prune other objects which for example holds for the
inverse $k$NN query and the inverse skyline query but not for the
inverse $\eps$-range query.

\subsubsection*{Refinement:} In the final {\em refinement} step, the
remaining candidates are verified and the \emph{true hits} are
reported as results.

%\section{Implementations for the Inverse Spatial Query Processing Framework}

\section{Inverse $\eps$-Range Query}
\label{sec:ieps} We will start with the simpler query, the inverse
$\eps$-range query. First, consider the case of a query object $q$
(i.e., $|Q|=1$). In this case, the inverse $\eps$-range query
computes all objects, that have $q$ within their $\eps$-range
sphere. Due to the symmetry of the $\eps$-range query predicate,
all objects satisfying the inverse $\eps$-range query predicate
are within the $\eps$-range sphere of $q$ as illustrated in Figure
\ref{subfig:pruning_space_eps_single}. In the following, we
consider the general  case, where $|Q|>1$ and show how our
framework can be applied.

\begin{figure}
    \centering
    \subfigure[Single query case.\label{subfig:pruning_space_eps_single}]{\includegraphics[width=0.45\columnwidth]{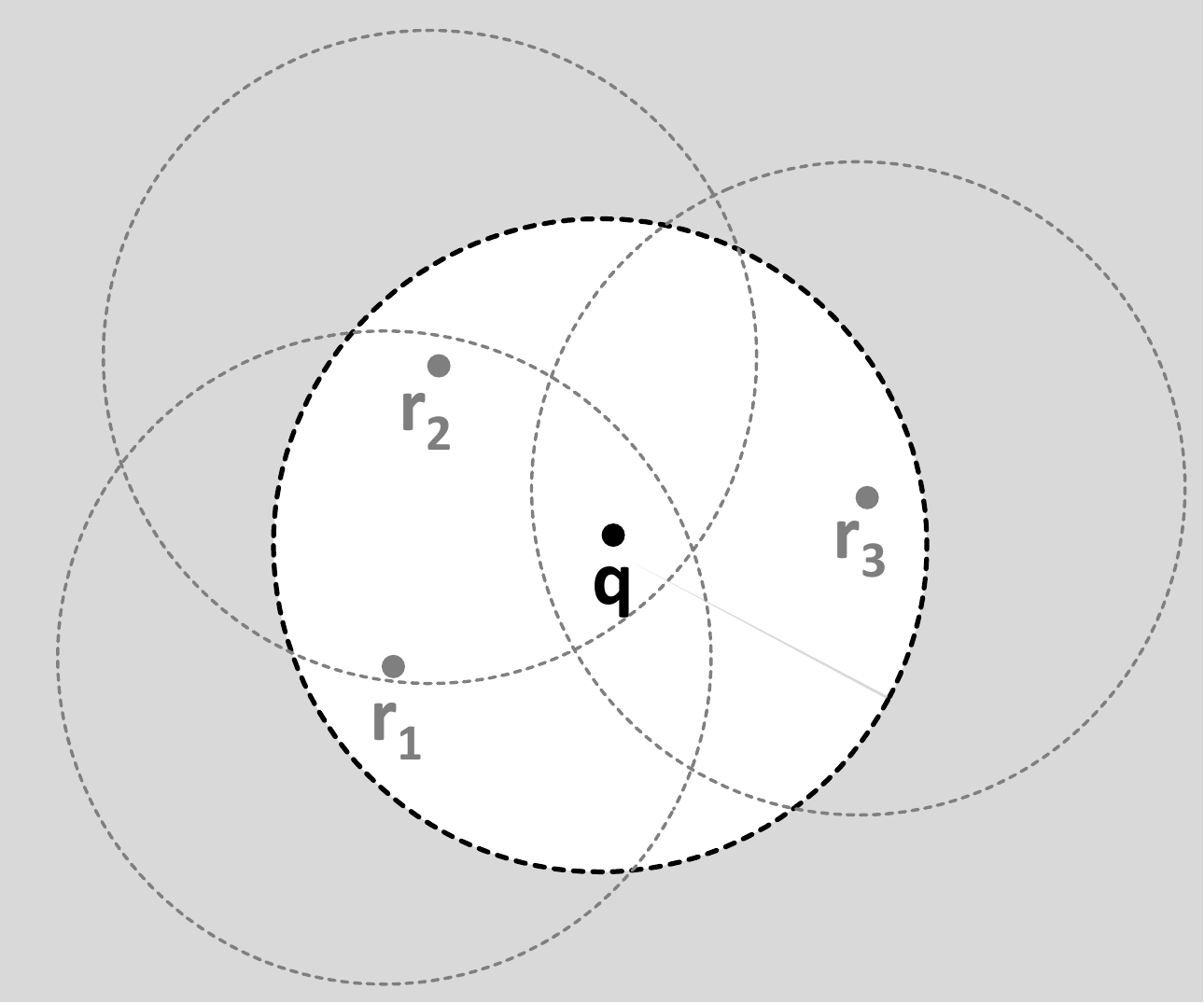}}
    \subfigure[Multiple query case.\label{subfig:pruning_space_eps_multiple}]{\includegraphics[width=0.45\columnwidth]{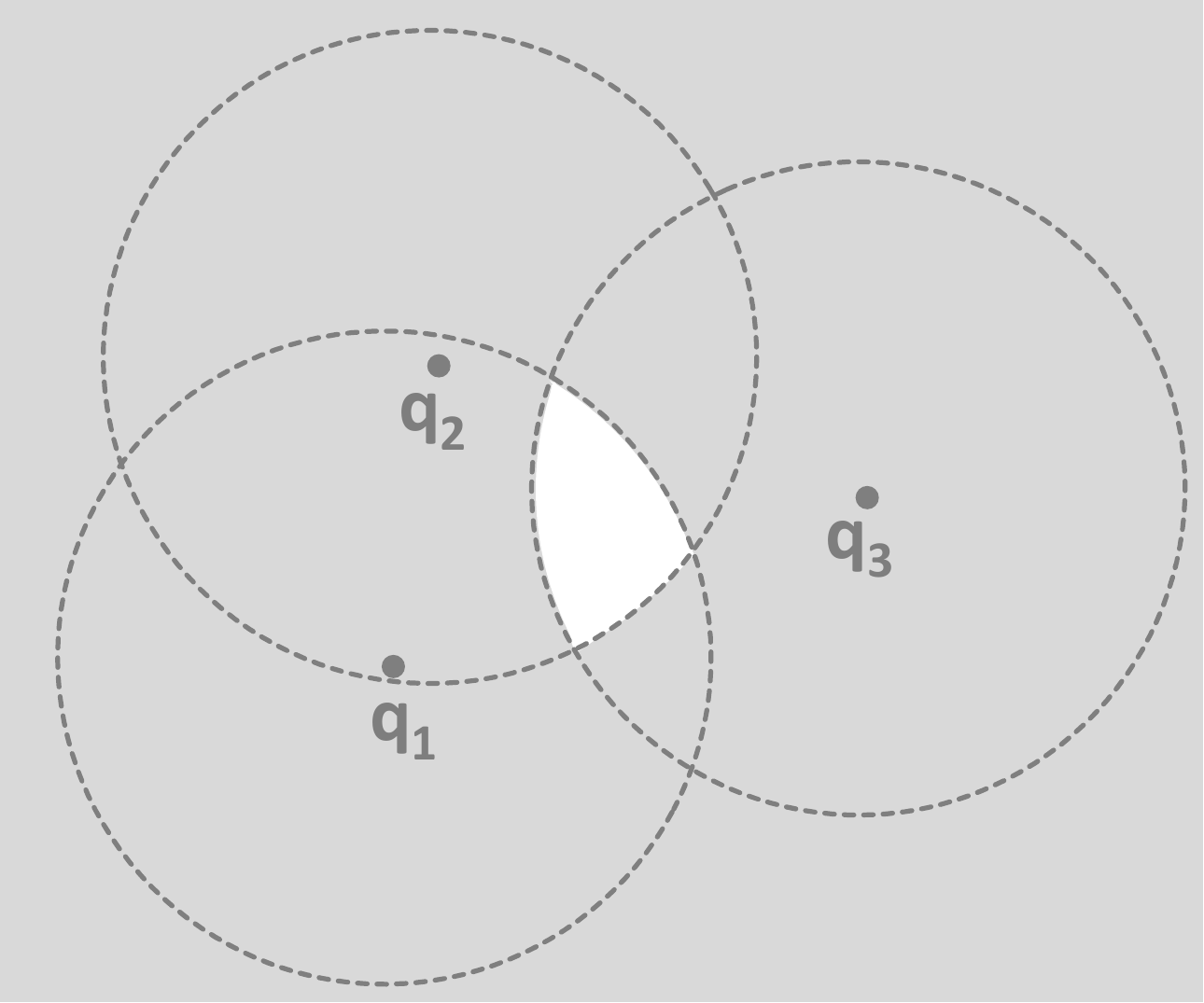}}
\vspace{-3mm}
    \caption{Pruning space for $I\eps\mbox{-}RQ$.}
    \label{fig:pruning_space_eps}
\end{figure}

\subsection{Framework Implementation}

\subsubsection*{Fast Query Based Validation}
%As mentioned in Section
%\ref{sec:InverseQueryFramework}, this filter uses the set of query
%objects $Q$ only to make a quick check whether the result is
%empty. Here we just need to consider the distances between pairs
%of query objects. Let $Q=\{q_1,\ldots,q_m\}$ be the set of query
%objects in $I\eps\mbox{-}RQ(Q)$.
There is no possible result if there exists a pair $q,q'$ of
queries in $Q$, such that their $\eps$-ranges do not intersect
(i.e., $d(q,q')>2\cdot\eps$). In this case, there can be no object
$r$ having both $q$ and $q'$ within its $\eps$-range (a necessary
condition for $r$ to be in the result).

\subsubsection*{Query Based Pruning}
%Similar to the previous filter,
%this filter again uses just $Q$, i.e. the set of query
%objects. This time, $Q$ is used to prune objects in $\DB$.
Let $S_i^{\eps}\subseteq\RR^d$ be the $\eps$-sphere around query
point $q_i$ for all $q_i\in Q$, as depicted in the example shown
in Figure \ref{subfig:pruning_space_eps_multiple}. Obviously, any
point in the intersection region of all spheres, i.e.
$\cap_{i=1..m}S_i^{\eps}$, has all query objects $q_i\in Q$ in its
$\eps$-range. Consequently, all objects outside of this region can
be pruned. However, the computation of the search region can
become too expensive in an arbitrary high dimensional space; thus,
we compute the intersection between rectangles that minimally
bound the hyper-spheres and  use it as  a filter. This can be done
quite efficiently even in high dimensional spaces; the resulting
filter rectangle is used as a window query and all objects in it
are passed to the refinement step as candidates.

\subsubsection*{Object Based Pruning} As mentioned in Section
\ref{sec:InverseQueryFramework} this filter is not applicable for
inverse $\eps$-range queries, since objects cannot be used to
prune other objects.

\subsubsection*{Refinement}
In the refinement step, for all  candidates we compute their
distances to all query points $q\in Q$ and report only objects
that are within distance $\eps$ from all query objects.

\subsection{Algorithm}

The implementation of our framework above can be easily converted
to an algorithm, which, after applying the filter steps, performs
a window query to retrieve  the candidates, which are finally
verified. Search can be facilitated
%In consideration of our framework compatible query processing
%strategy, the algorithm for the $I\eps\mbox{-}RQ$ consists of
%three steps. After the first step, the \emph{fast query based
%validation} probing the distance of each pair of query object
%against $\eps$, the second filter, i.e. \emph{query based pruning}
%will be applied. Thereby, the intersection of all minimal bounding
%rectangles of all $\eps$-spheres around all query points $q\in Q$
%will be computed. The resulting rectangle is used as a {\em
%filter} rectangle to search $\DB$. This search can be facilitated
by an R-tree that indexes $\DB$. Starting from the root, we search
the tree, using the filter rectangle. To minimize the I/O cost,
for each entry $P$ of the tree that intersects the filter
rectangle, we compute its distance to all points in $Q$ and access
the corresponding subtree only if all these distances are smaller
than $\eps$.

\section{Inverse $k$-NN Query}

For inverse $k$-nearest neighbor queries (I$k$-NNQ),
% efficiently applying the
%filter-refinement strategy of our framework.
%
%For completeness, we first consider the special case based on one
%single query object. The I$k$-NN query problem with a single query
%object is already solved, as it is equal to the \emph{reverse $k$
%nearest neighbor} (R$k$-NN) query problem \cite{KorMut00}. Here,
%we briefly review the state-of-the-art approach for evaluating
%this query, which is used as a module of our general I$k$-NN
%query.
we first consider the case of a single query object (i.e.,
$|Q|=1$). As discussed in Section \ref{sec:related}, this case can
be processed by the bi-section-based R$k$-NN approach (TPL)
proposed in \cite{TaoPapLia04}, enhanced by the rectangle-based
pruning criterion proposed in \cite{EmrKriKroRenZue10}. The core
idea of TPL is to use bi-section-hyperplanes between database
objects $o$ and the query object $q$ in order to check which
objects are closer to $o$ than to $q$. Each bi-section-hyperplane
divides the object space into two half-spaces, one containing $q$
and one containing $o$. Any object located in the half-space
containing $o$ is closer to $o$ than to $q$. The objects spanning
the hyperplanes are collected in an iterative way. Each object $o$
is then checked against the resulting half-spaces that do not
contain $q$. As soon as $o$ is inside more than $k$ such
half-spaces, it can be pruned.
%An example illustrating the pruning space for a given query object
%is depicted in Figure \ref{subfig:pruning_space_knn_single} for
%$k=2$ {\bf (Nikos: you need to describe the example in detail.)}.
%We apply the rectangle based pruning method
%\cite{EmrKriKroRenZue10}, which does not require the
%materialization of the hyperplanes and provides an index supported
%pruning by means of a rectangle based index structure, e.g. the
%R-tree.
Next, we consider queries with multiple objects (i.e., $|Q|>1$)
and discuss how the framework presented in Section
\ref{sec:InverseQueryFramework} is implemented in this case.

\subsection{Framework Implementation}

\subsubsection*{Fast Query Based Validation}
Recall that this filter uses  the set of query objects $Q$ only,
to perform a quick check on whether the result is empty. Here, we
use the obvious rule that the result is empty if the number of
query objects exceeds the query parameter $k$.
%If
%this is the case we can stop the query and report the message "no
%result".

\subsubsection*{Query Based Pruning}
%As mentioned in
%\ref{sec:InverseQueryFramework} this filter is not directly
%applicable for inverse $k$-NN queries, since the query objects
%alone cannot be used to restrict the pruning space. However, w
We can exploit the query objects in order to reduce the I$k$-NN
query to an I$k'$-NN query with $k'<k$. A smaller query parameter
$k'$ allows us to terminate the query process earlier and reduce
the search space. We first show how $k$ can be reduced by means of
the query objects only.
%Later
%we will see how to use database objects in order to
%further reduce $k$.
The proofs for all lemmas are presented in Appendix
\ref{appendix:proofs}.

\begin{lemma}
\label{lemma:IkNNQ_pruning_criterion_I} Let $\DB\subseteq\RR^d$ be
a set of database objects and $Q\subseteq\DB$ be a set of query
objects.
% and $k\geq|Q|$.%
%\footnote{By definition, for $k<|Q|$, I$k$-NNQ cannot have any
%result.}
Let $\DB' = \DB-Q$. For each $o\in\DB'$, the following statement
holds:
\begin{small}
$$
o\in Ik\mbox{-}NNQ(Q)~\textrm{in}~\DB \Rightarrow \forall q\in Q:
o\in Ik'\mbox{-}NNQ(\{q\})~\textrm{in}~\DB'\cup\{q\},
$$
$$
\mbox{where }k'=k-|Q|+1.
$$
\end{small}
\end{lemma}

Simply speaking, if a candidate object $o$ is not in the
$Ik'\mbox{-}NNQ(\{q\})$ result of some $q\in Q$ considering only
the points $\DB'\cup\{q\}$, then $o$ cannot be in the
$Ik\mbox{-}NNQ(Q)$ result considering all points in \DB\ and $o$
can be pruned. As a consequence, $Ik'\mbox{-}NNQ(\{q\})$ in
$\DB'\cup\{q\}$ can be used to prune candidates for any $q\in Q$.
The pruning power of $Ik'\mbox{-}NNQ(\{q\})$ depends on how $q\in
Q$ is selected.

From Lemma \ref{lemma:IkNNQ_pruning_criterion_I}
%and Corollary
%\ref{lemma:IkNNQ_TrueHit_criterion_I}
we can conclude the following:

\begin{lemma}
\label{lemma:IkNNQ_query_point_filter} Let $o\in\DB-Q$ be a
database object and $q_{ref}\in Q$ be a query object such that
$\forall q\in Q: d(o,q_{ref})\geq d(o,q)$. Then
$$
o\in Ik\mbox{-}NNQ(Q)\Leftrightarrow o\in
Ik'\mbox{-}NNQ(\{q_{ref}\})~\textrm{in}~\DB'\cup\{q\},
$$
where $k'=k-|Q|+1$.
\end{lemma}

Lemma \ref{lemma:IkNNQ_query_point_filter} suggests that for any
candidate object $o$ in $\DB$, we should use the furthest query
point to check whether $o$ can be pruned.

\subsubsection*{Object Based Pruning}
Up to now, we only used the query points in order to reduce $k$ in
the inverse $k$-NN query. Now, we will show how to consider
database objects in order to further decrease  $k$.

\begin{lemma}
\label{lemma:convex_hull_criterion_I} Let $Q$ be the set of query
objects and $\mathcal{H}\subseteq\DB-Q$ be the non-query(database)
objects covered by the convex hull of $Q$. Furthermore, let
$o\in\DB$ be a database object and $q_{ref}\in Q$ a query object
such that $\forall q\in Q: d(o,q_{ref})\geq d(o,q)$. Then for each
object $p\in\mathcal{H}$ it holds that $d(o,p)\leq d(o,q_{ref})$.
\end{lemma}

According to the above lemma the following statement holds:

\begin{lemma}
\label{lemma:IkNNQ_pruning_criterion_II} Let $Q$ be the set of
query objects, $\mathcal{H}\subseteq\DB-Q$ be the database
(non-query) objects covered by the convex hull of $Q$ and let
$q_{ref}\in Q$ be a query object such that $\forall q\in Q:
d(o,q_{ref})\geq d(o,q)$. Then
$$
\forall o\in\DB-\mathcal{H}-Q: o\in
Ik\mbox{-}NNQ(Q)\Leftrightarrow
$$
at most $k'=k-|\mathcal{H}|-|Q|$ objects $p\in\DB-\mathcal{H}$ are
closer to $o$ than $q_{ref}$, and
$$
\forall o\in\mathcal{H}: o\in Ik\mbox{-}NNQ(Q)\Leftrightarrow
$$
at most $k'=k-|\mathcal{H}|-|Q|+1$ objects $p\in\DB-\mathcal{H}$
are closer to $o$ than $q_{ref}$.
\end{lemma}

Based on Lemma \ref{lemma:IkNNQ_pruning_criterion_II}, given the
number of objects in the convex hull of $Q$, we can prune objects
outside of the hull from I$k$-NN($Q$). Specifically, for an
I$k$-NN query we have the following pruning criterion: An object
$o\in\DB$ can be pruned, as soon as we find more than $k'$ objects
$p\in\DB-\mathcal{H}$ outside of the convex hull of $Q$, that are
closer to $o$ than $q_{ref}$. Note that the parameter $k'$ is set
according to Lemma \ref{lemma:IkNNQ_pruning_criterion_II} and
depends on whether $o$ is in the convex hull of $Q$ or not.
Depending on the size of $Q$ and the number of objects within the
convex hull of $Q$, $k'=k-|\mathcal{H}|+1$ can become negative. In
this case,
% special case
%when $k'=k-|\mathcal{H}|+1$ becomes negative,
we can terminate query evaluation immediately, as no object can
qualify the inverse query (i.e., the inverse query result is
guaranteed to be empty). The case where $k'=k-|\mathcal{H}|+1$
becomes zero is another special case, as all objects outside of
$\mathcal{H}$ can be pruned. For all objects in the convex hull of
$Q$ (including all query objects) we have to check whether there
are objects outside of $\mathcal{H}$ that prune them.
\begin{figure}
    \centering

    \subfigure[Pruning $o_1$\label{subfig:iknn1}]{\includegraphics[width=0.45\columnwidth]{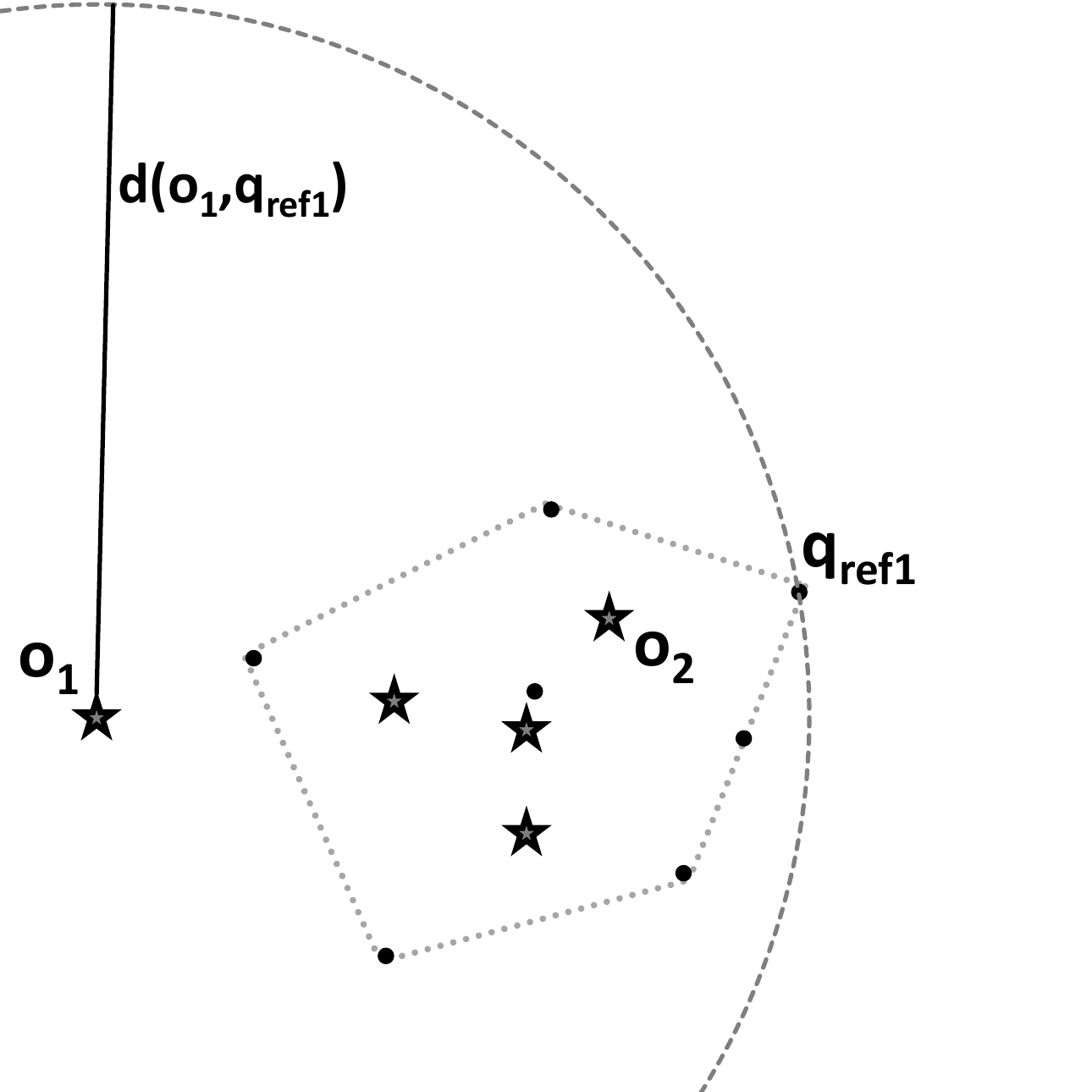}}
    ~~~~~~
    \subfigure[Pruning $o_2$\label{subfig:iknn2}]{\includegraphics[width=0.45\columnwidth]{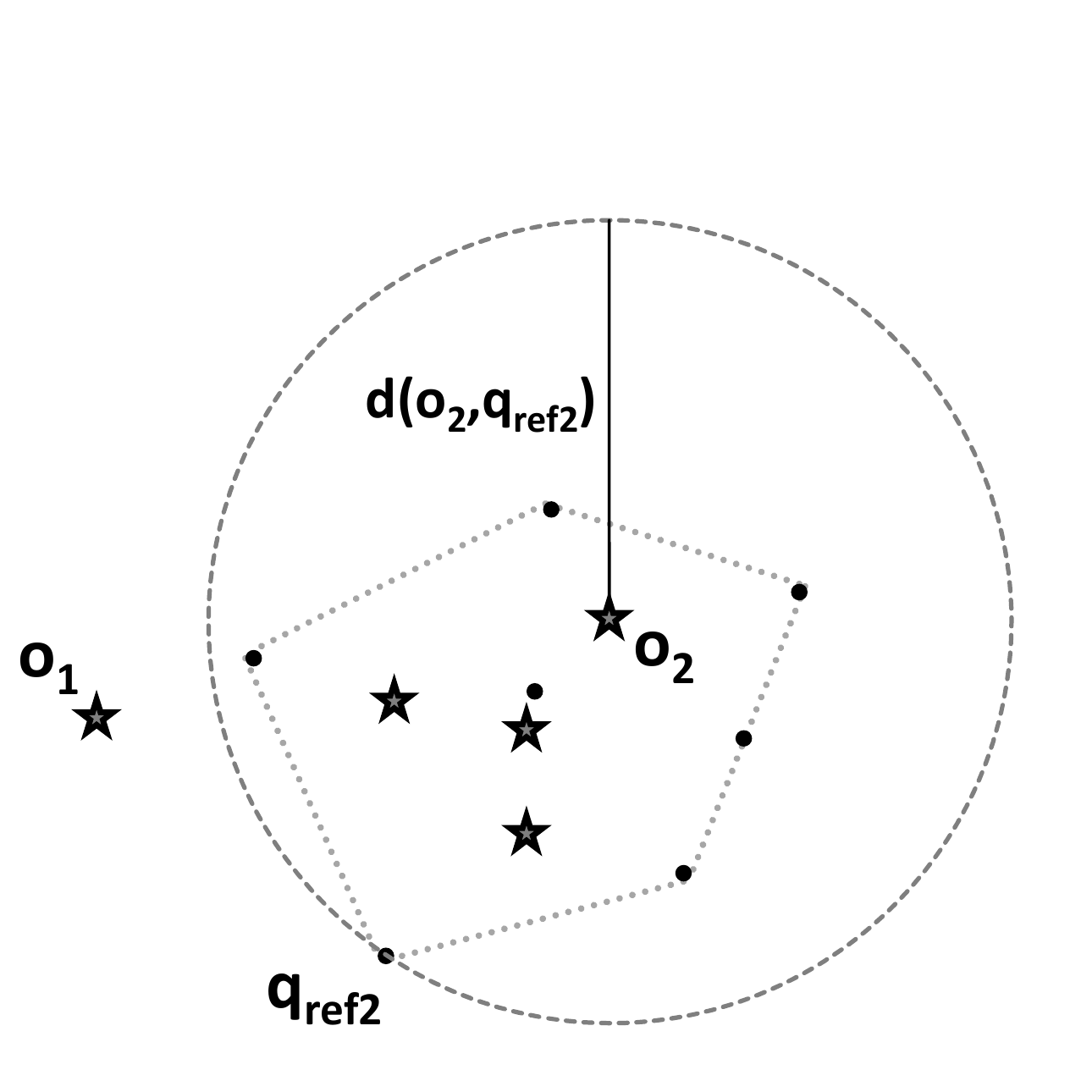}}
    \caption{I$k$-NN pruning based on Lemma \ref{lemma:IkNNQ_pruning_criterion_II}}
    \label{fig:iknn}
\end{figure}

As an example of how Lemma \ref{lemma:IkNNQ_pruning_criterion_II}
can be used, consider the data shown in Figure \ref{fig:iknn} and
assume that
 we wish to perform an inverse 10NN query using a set $Q$ of
seven query objects, shown as points in the figure; non-query
database points are represented by stars. In Figure
\ref{subfig:iknn1}, the goal is to determine whether candidate
object $o_1$ is a result, i.e., whether $o_1$ has all $q\in Q$ in
its 10NN set. The query object having the largest distance to
$o_1$ is $q_{ref1}$. Since $o_1$ is located outside of the convex
hull of $Q$ (i.e, $o\in\DB-\mathcal{H}-Q$), the first equivalence
of Lemma \ref{lemma:IkNNQ_pruning_criterion_II}, states that $o_1$
is a result if at most $k^{\prime}=k-|\mathcal{H}|-|Q|=10-4-7=-1$
objects in $\DB-\mathcal{H}-Q$ are closer to $o_1$ than
$q_{ref1}$. Thus, $o_1$ can be safely pruned without even
considering these objects (since obviously, at least zero objects
are closer to $o_1$ than $q_{ref1}$). Next, we consider object
$o_2$ in Figure \ref{subfig:iknn2}. The query object with the
largest distance to $o_2$ is $q_{ref2}$. Since $o_2$ is inside the
convex hull of $Q$, the second equivalence of Lemma
\ref{lemma:IkNNQ_pruning_criterion_II} yields that $o_2$ is a
result if at most $k^{\prime}=k-|\mathcal{H}|-|Q|+1=10-4-7+1=0$
objects $\DB-\mathcal{H}-Q$ are closer to $o_2$ than $q_{ref2}$.
This, $o_2$ remains a candidate until at least one object in
$\DB-\mathcal{H}-Q$ is found that is closer to $o_2$ than
$q_{ref2}$.

\subsubsection*{Refinement}
%In the refinement step, we  verify the remaining
%candidates.
Each remaining candidate is checked whether it is a result of the
inverse query by performing a $k$-NN search and verifying whether
its result includes $Q$.

%{\bf Nikos: the part below has to be revised. We need a concrete
%algorithm that browses the tree and prunes nodes and points using
%the lemmata. Given an R-tree node, how is $q_{ref}$ determined?
%Then, how can this node be pruned? Would an aggregate tree help?}

\subsection{Algorithm}
\label{sec:iknn-algo} We now present a complete algorithm that
traverses an {\em aggregate} R-tree ($ARTree$), which indexes
$\DB$ and computes  $Ik\mbox{-}NNQ(Q)$ for a given set $Q$ of
query objects, using Lemma \ref{lemma:IkNNQ_pruning_criterion_II}
to prune the search space. The entries in the tree nodes are
augmented with the cardinality of objects in the corresponding
sub-tree. These counts can be used to accelerate search, as we
will see later.

In a nutshell, the algorithm, while traversing the tree, attempts
to prune nodes based on the lemma using the information known so
far about the points of $\DB$ that are included in the convex hull
({\em filtering}). The objects that survive the pruning are
inserted in the {\em candidates} set. During the {\em refinement}
step, for each point $c$ in the candidates set, we run a $k$-NN
query to verify whether $c$ contains $Q$ in its $k$-NN set.
%Finding the points of $\DB$ inside the convex hull of $Q$ (which are used
%by the lemma to prune points outside the hull) is done concurrently
%with the tree traversal and the

%\subsubsection{Reference Query Selection}
%
%As mentioned above, we select for each database object $o$ a
%reference query object $q_{ref}$ by taking the one which is most
%distant to $o$. Similar to the inverse $\eps$-range query, for
%this step it suffices to take the vertices of the convex hull of
%$Q$ into account, in consideration of Lemma
%\ref{lema:ch_reference_query_selection}.

%\subsubsection{Query Refinement}
%
%Given the reduced $k$ parameter and the reference query object for
%each object $o\in\DB$, the I$k$-NN query works ...
%
%The main difference in our approach is that we have multiple query
%objects. such that we have to chose the reference query object
%$q_{ref}$ individually for each object as mentioned above.

Algorithm \ref{alg:iknn} in Appendix \ref{app:alg} is a pseudocode
of our approach. The $ARTree$ is traversed in a best-first search
manner \cite{HjaSam95}, prioritizing the access of the nodes
according to the maximum possible distance (in case of a non-leaf
entry we use {\em MinDist}) of their contents to the query points
$Q$. In specific, for each R-tree entry $e$ we can compute, based
on its MBR, the furthest possible point $q_{ref}$ in $Q$ to a
point indexed under $e$. Processing the entries with the smallest
such distances first helps to find points in the convex hull of
$Q$ earlier, which helps making the pruning bound tighter.

Thus, initially, we set $|\mathcal{H}|=0$, assuming that in the
worst case the number of non-query points in the convex hull of
$Q$ is 0. If the object which is deheaped is inside the convex
hull, we increase $|\mathcal{H}|$ by one. If a non-leaf entry is
deheaped and its MBR is contained in the hull, we increase
$|\mathcal{H}|$ by the number of objects in the corresponding
sub-tree, as indicated by its augmented counter.

During the tree traversal, the accessed tree entries could be in
one of the following sets (i) the set of {\em candidates}, which
contains objects that could possibly be results of the inverse
query, (ii) the set of {\em pruned entries}, which contains
(pruned) entries whose subtrees may not possibly contain inverse
query results, and (iii) the set of entries which are currently in
the priority queue. When an entry $e$ is deheaped, the algorithm
checks whether it can be pruned. For this purpose, it initializes
a $prune\_counter$ which is a lower bound of the number of objects
that are closer to every point $p$ in $e$ than $Q$'s furthest
point to $p$. For every entry $e'$ in all three sets (candidates,
pruned, and priority queue), we increase the $prune\_counter$ of
$e$ by the number of points in $e'$ if the following condition
holds: $\forall p \in e, \forall p' \in e': dist(e,e') <
dist(e,q_{ref})$. This condition can efficiently be checked using
the technique from \cite{EmrKriKroRenZue10}. An example were this
condition is fulfilled is shown in Figure \ref{fig:prune_count}.
Here the $prune\_counter$ of $e$ can be increased by the number of
points in $e'$.

\begin{figure}
    \centering
    \includegraphics[width = 0.5\columnwidth]{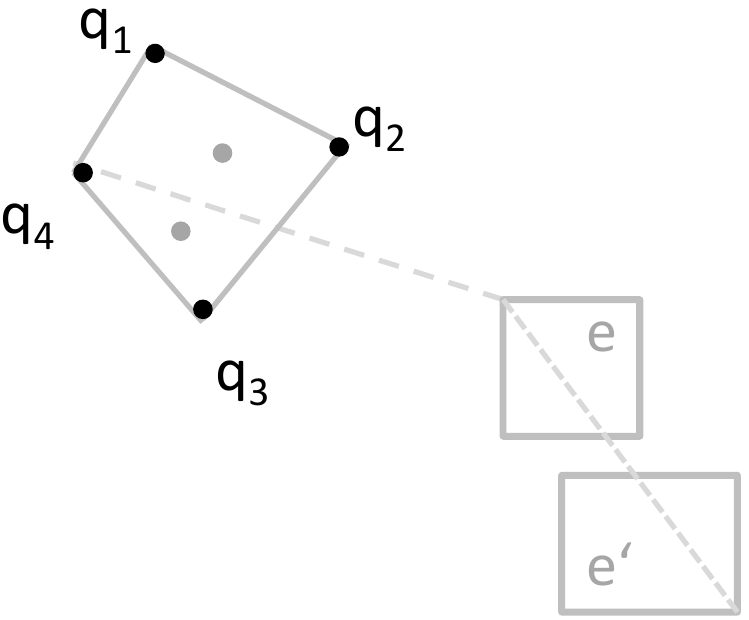}
    \caption{Calculating the $prune\_count$ of e}
    \label{fig:prune_count}
\end{figure}

While updating $prune\_counter$ for $e$, we check whether
\linebreak $prune\_counter>k-|\mathcal{H}|-|Q|$
($prune\_counter>k-|\mathcal{H}|-|Q|+1$) for entries that are
entirely outside of (intersect) the convex hull. As soon as this
condition is true, $e$ can be pruned as it cannot contain objects
that can participate in the inverse query result (according to
Lemma \ref{lemma:IkNNQ_pruning_criterion_II}). Considering again
Figure \ref{fig:prune_count} and assuming the number of points in
$e'$ to be 5, $e$ could be pruned for $k \leq 10$ (since
$prune\_counter(5) > k(10)-|\mathcal{H}|(2)-|Q|(4)$ holds). In
this case $e$ is moved to the set of pruned entries. If $e$
survives pruning, the node pointed to by $e$ is visited and its
entries are enheaped if $e$ is a non-leaf entry; otherwise $e$ is
inserted in the {\em candidates} set.

When the queue becomes empty, the filter step of the algorithm
completes with a set of {\em candidates}. For each object $c$ in
this set, we check whether $c$ is a result of the inverse query by
performing a $k$-NN search and verifying whether its result
includes $Q$. In our implementation, to make this test faster, we
replace the $k$-NN search by an aggregate $\eps$-range query
around $c$, by setting $\eps = d(c,q_{ref})$, where $q_{ref}$ is
the furthest point of $Q$ to $p$. The objective is to count
whether the number of objects in the range is greater than $k$. In
this case, we can prune $c$, otherwise $c$ is a result of the
inverse query. $ARTree$ is used to process the aggregate
$\eps$-range query; for every entry $e$ included in the
$\eps$-range, we just increase the aggregate count by the
augmented counter to $e$ without having to traverse the
corresponding subtree. In addition, we perform batch searching for
candidates that are close to each other, in order to optimize
performance. The details are skipped due to space constraints.
%An active page
%list (APL) is used to organize index entries when exploring the
%R-tree in a best-first search manner starting with the root. All
%entries in the APL are sorted w.r.t. the ...

%\subsubsection{Algorithm}
%
%\textbf{Filter:}
%\begin{itemize}
%    \item $Q'\subseteq Q: Q'=\{q\in Q: q \mbox{ is a vertex in the
%    convex hull of }Q\}$
%    \item $\mathcal{H} = \{o\in\DB: o \in \mbox{ convex hull of
%    }Q'\}$
%    \item $k' = k - |\mathcal{H}| + 1$
%    \item if $k'<0$ then report $\emptyset$
%    \item else:
%    \item initialize APL list sorted in ascending order to $\max_{q_i\in
%    Q'}(d_{min}(q_i,e))$. ($e$ is an entry in APL which is either an index page region (MBR) or a database object)
%    \item APL.insert(root)
%    \item while APL not empty DO:
%    \item $e$ = APL.fetch\_next()
%    \item $q_{ref} = argmax_{q_i\in Q'}(d_{min}(q_i,))$
%\end{itemize}

\vspace{-3mm}
\section{Inverse Dynamic Skyline Query}
\label{sec:idsq}
%In this section, we introduce processing
%strategies for the inverse dynamic skyline query ($I\mbox{-}DSQ$)
%in accordance to our framework (cf. Section
%\ref{sec:InverseQueryFramework}).

We again first discuss the case of a single query object, which
corresponds to the reverse dynamic skyline query \cite{LiaChe08}
and then present a solution for the more interesting case where
$|Q|>1$. Let $q$ be the (single) query object with respect to
which we want to compute the inverse dynamic skyline. Any object
$o\in \DB$ defines a pruning region, such that any object $o'$ in
this region cannot be part of the inverse query result. Formally:
\begin{definition}[Pruning Region]
\label{def:pruning_region} Let $q=(q^1,\ldots,q^d)\in Q$ be a
single $d$-dimensional query object and $o=(o^1,\ldots,o^d)\in\DB$
be any $d$-dimensional database object. Then the pruning region
$PR_q(o)$ of $o$ w.r.t. $q$ is defined as the $d$-dimensional
rectangle where the $i$th dimension of $PR_q(o)$ is given by
$[\frac{q^i+o^i}{2},+\infty]$ if $q^i\leq o^i$ and
$[-\infty,\frac{q^i+o^i}{2}]$ if $q^i\geq o^i$.\\
\end{definition}
\vspace{-3mm}

The pruning region of an object $o$ with respect to a single query
object $q$ is illustrated by the shaded region in Figure
\ref{subfig:skyline_1a}.

\begin{figure}
    \centering
    \subfigure[pruning region\label{subfig:skyline_1a}]{\includegraphics[width=0.48\columnwidth]{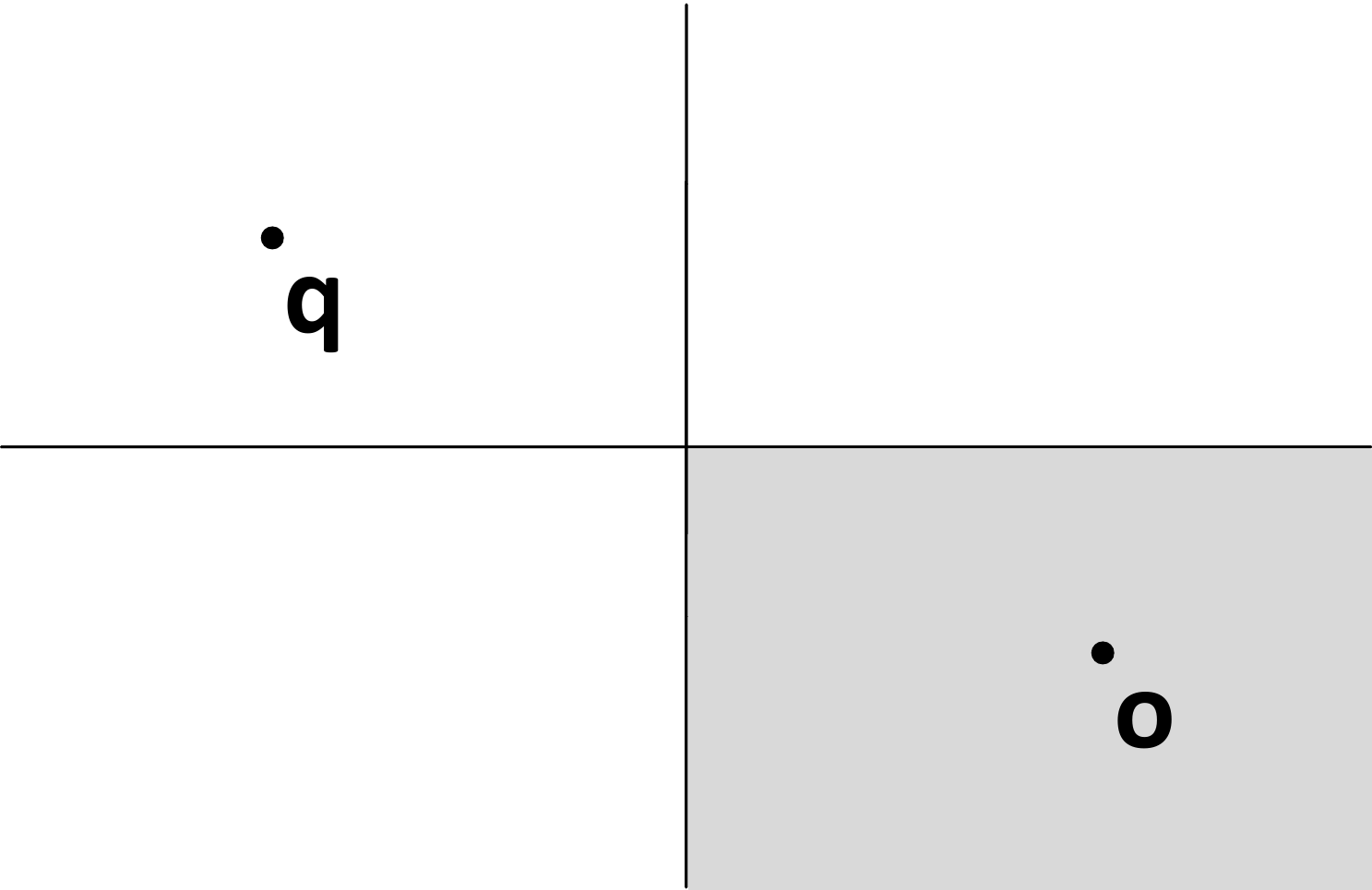}}
    \hfill
    \subfigure[candidates\label{subfig:skyline_1b}]{\includegraphics[width=0.48\columnwidth]{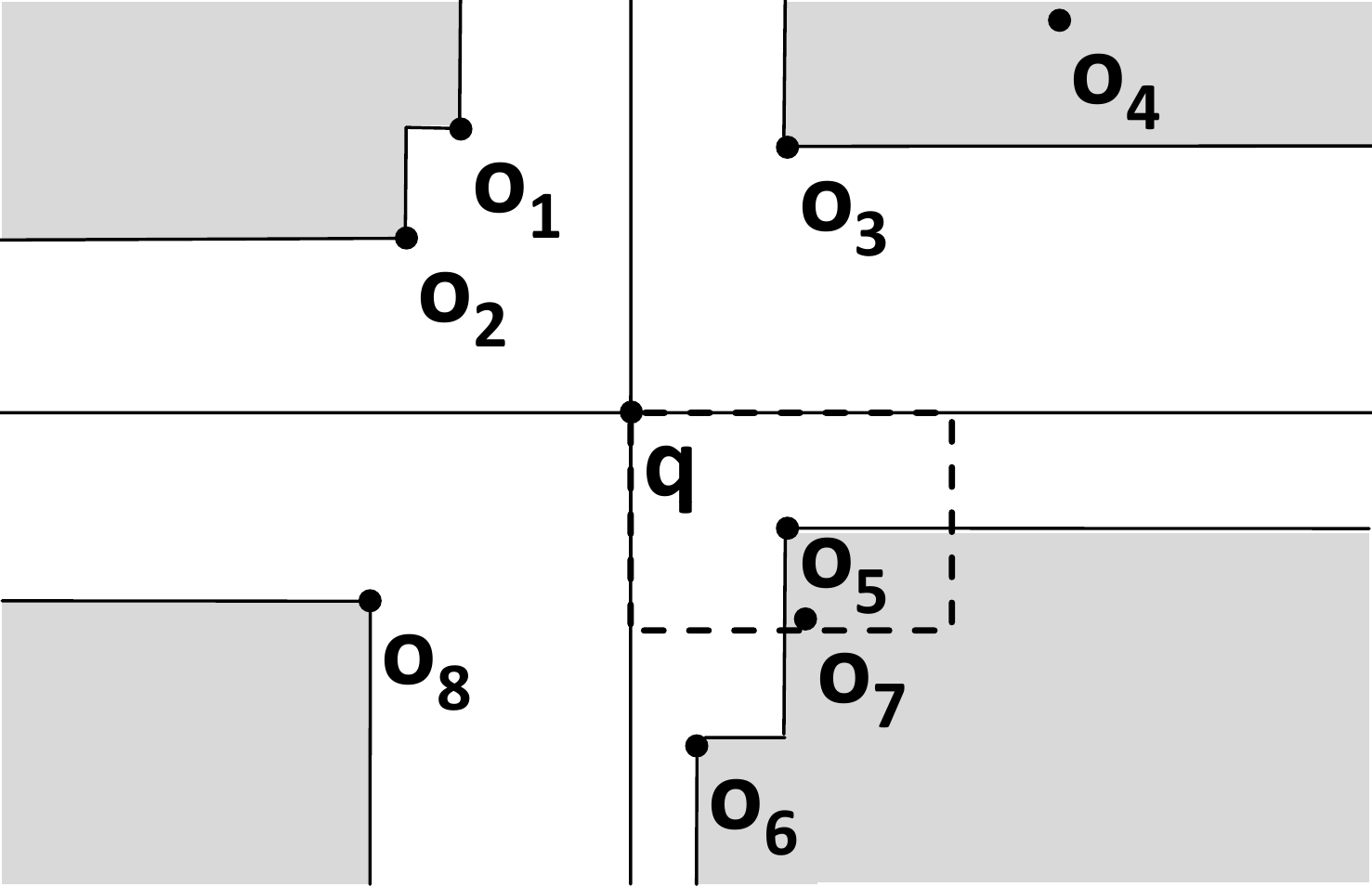}}
\vspace{-3mm} \caption{Single-query case}
    \label{fig:skyline_1}
\end{figure}

{\bf Filter step.} As shown in \cite{LiaChe08}, any object
$p\in\DB$ can be safely pruned if $p$ is contained in the pruning
region of some $o\in\DB$ w.r.t. $q$ (i.e. $p\in PR_q(o)$).
Accordingly, we can use $q$ to divide the space into $2^d$
partitions by splitting along each dimension at $q$. Let $o\in\DB$
be an object in any partition $P$; $o$ is an $I\mbox{-}DSQ$ {\em
candidate}, iff there is no other object $p\in P\subseteq\DB$ that
dominates $o$ w.r.t. $q$.

Thus, we can derive all $I\mbox{-}DSQ$ candidates as follows:
First, we split the data space into the $2^d$ partitions at the
query object $q$ as mentioned above. Then in each partition, we
compute the skyline\footnote{Only objects within the same
partition are considered for the dominance relation.}, as
illustrated in the example depicted in Figure
\ref{subfig:skyline_1b}. The union of the four skylines is the set
of the inverse query candidates (e.g.,
$\{o_1,o_2,o_3,o_5,o_6,o_8\}$ in our example).

{\bf Refinement.} The result of the reverse dynamic skyline query
is finally obtained by verifying for each candidate $c$, whether
there is an object in $\DB$ which dominates $q$ w.r.t. $c$. This
can be done by checking whether the hypercube centered at $c$ with
extent $2\cdot|c^i-q^i|$ at each dimension $i$ is empty. For
example, candidate $o_5$ in Figure \ref{subfig:skyline_1b} is not
a result, because the corresponding box (denoted by dashed lines)
contains $o_7$. This means that in both dimensions $o_7$ is closer
to $o_5$ than $q$ is.

\subsection{IQ Framework Implementation}
\label{subsec:FrameworkImplementation_skyline}

\subsubsection*{Fast Query Based Validation} Following our framework,
first the set $Q$ of query objects is used to decide whether it is
possible to have any result at all. For this, we use the following
lemma:
%In particular for the
%inverse dynamic skyline query, there cannot be any result if the
%following holds:
\begin{lemma}\label{lem:fastSky}
Let $q\in Q$ be any query object and let $\mathcal{S}$ be the set
of $2^d$ partitions derived from dividing the object space at $q$
along the axes into two halves in each dimension. If in each
partition $r\in\mathcal{S}$ there is at least one query object
$q'\in Q$ $(q'\neq q)$, then there cannot be any result.
\end{lemma}
% \begin{proof}
% Let us consider the space partitioning $\mathcal{S}$ derived from
% dividing the object space at $q$. Each $q'$ located within
% partition $r\in\mathcal{S}$ generates a pruning region $PR_q(q')$
% (cf. Definition \ref{def:pruning_region}) that totally covers the
% partition $r'\in\mathcal{S}$ which is opposite to $r$ w.r.t. $q$.
% Since we assume that we have at least one query object $q'\neq q$
% in each partition $r\in\mathcal{S}$, all partitions
% $r'\in\mathcal{S}$ are totally covered by a pruning region and,
% consequently, the complete data space can be pruned. An example in
% the two-dimensional space is illustrated in Figure
% \ref{fig:skyline_proof}.
% \end{proof}

\subsubsection*{Query Based Pruning}
We now propose a  filter, which uses the set $Q$ of query objects
only in order to reduce the space of candidate results. We explore
similar strategies as the fast query based validation. For any
pair of query objects $q,q'\in Q$, we can define two pruning
regions, according to Definition \ref{def:pruning_region}:
$PR_q(q')$ and $PR_{q'}(q)$.
%According to the discussion in Section \ref{sec:sky:single}, a
Any object inside these regions cannot be a candidate of the
inverse query result because it cannot have both $q_1$ and $q_2$
in its dynamic skyline point set. Thus, for every pair of query
objects, we can determine the corresponding pruning regions and
use their union to prune objects or R-tree nodes that are
contained in it. Figure \ref{fig:skyline_3} shows examples of the
pruning space for $|Q|=3$ and $|Q|=4$. Observe that with the
increase of $|Q|$ the remaining space, which may contain
candidates, becomes very limited.

The main challenge is how to encode and use the pruning space
defined by $Q$, as it can be arbitrarily complex in the
multidimensional space. As for the $Ik\mbox{-}NNQ$ case, our
approach is not to explicitly compute and store the pruning space,
but to check on-demand whether each object (or R-tree MBR) can be
pruned by one or more query pairs. This has a complexity of
$O(|Q|^2)$ checks per object. In Appendix \ref{appn:sky_2d}, we
show how to reduce this complexity for the special 2D case. The
techniques shown there can also be used in higher dimensional
spaces, with lower pruning effect.

\begin{figure}
    \centering
    \subfigure[$|Q|=3$\label{subfig:skyline_3a}]{\includegraphics[width=0.48\columnwidth]{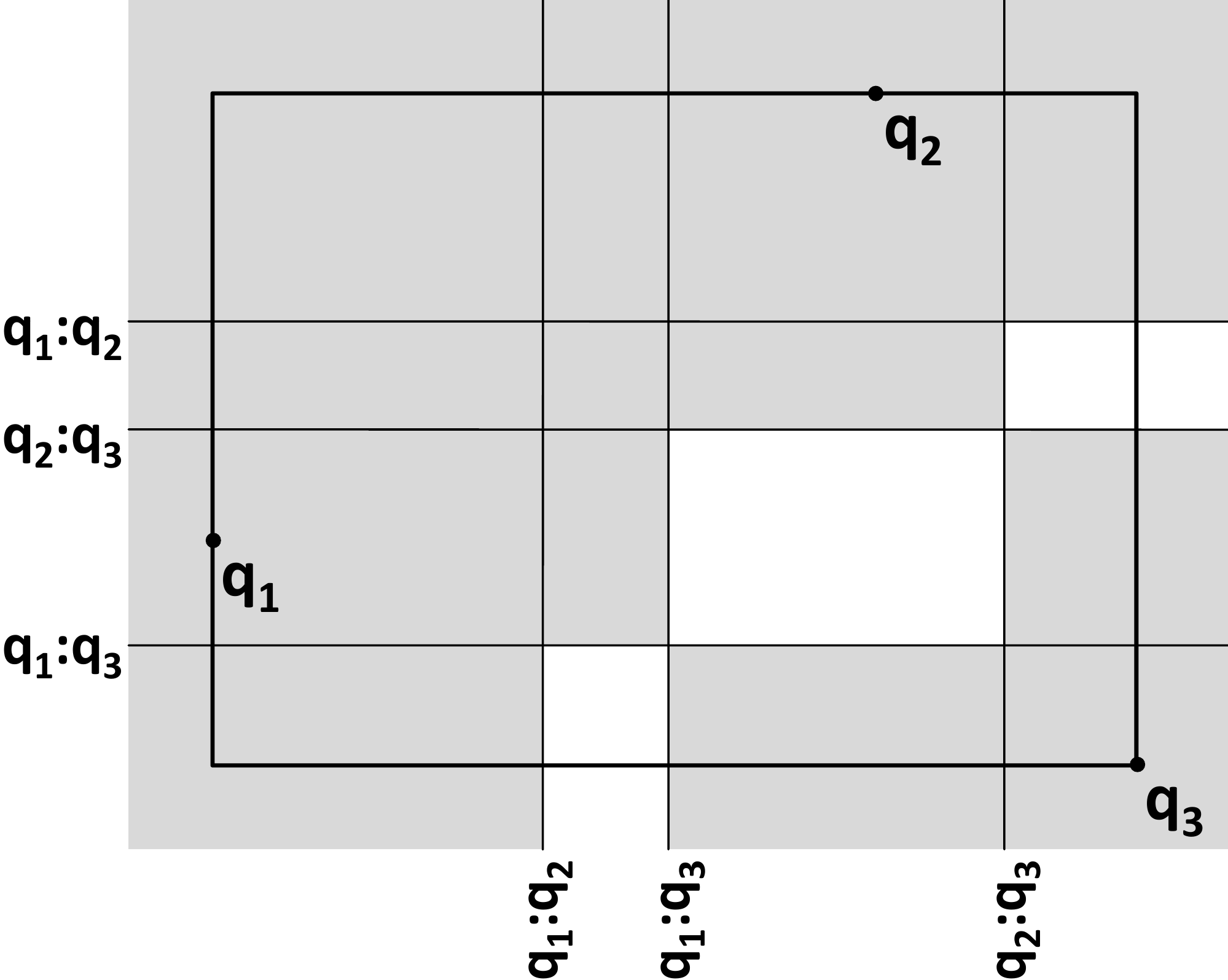}}
    \hfill
    \subfigure[$|Q|=4$\label{subfig:skyline_3b}]{\includegraphics[width=0.48\columnwidth]{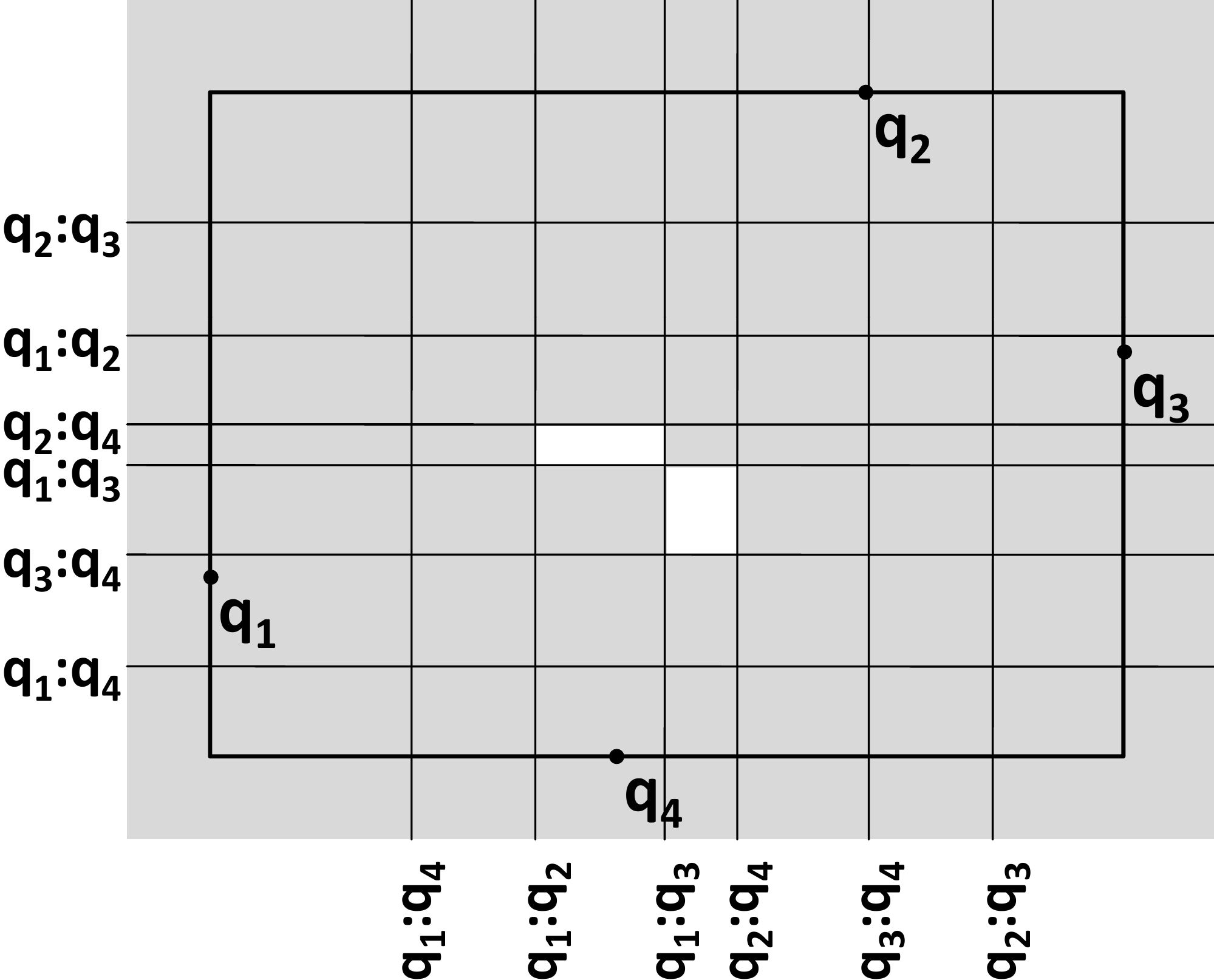}}
\vspace{-3mm}
    \caption{Pruning regions of query objects}
    \label{fig:skyline_3}
\end{figure}

\subsubsection*{Object Based Pruning}
For any candidate object $o$ that is not pruned during the
query-based filter step, we need to check if there exists any
other database object $o^\prime$ which dominates some $q\in Q$
with respect to $o$. If we can find such an $o^\prime$, then $o$
cannot have $q$ in its dynamic skyline and thus $o$ can be pruned
for the candidate list.

\subsubsection*{Refinement}
In the refinement step, each candidate $c$ is verified by
performing a dynamic skyline query using $c$ as query point. The
result should contain all $q_i \in Q$, otherwise $c$ is dropped.
The refinement step can be improved by the following observation
(cf. Figure \ref{fig:refine}): for checking if a candidate $o_1$
has all $q_i \in Q$ in its dynamic skyline, it suffices to check
whether there exists at least one other object $o_j \in \DB$ which
prevents one $q_i$ from being part of the skyline. Such an object
has to lie within the MBR defined by $q_i$ and $q_i'$ (which is
obtained by reflecting $q_i$ through $o_1$). If no point is within
the $|Q|$ MBRs, then $o_1$ is reported as result.

\subsection{Algorithm}
\label{sec:idsq-algo}

The algorithm for $I\mbox{-}DSQ$ is shown as Algorithm
\ref{alg:isnn} in Appendix \ref{app:alg}.
%), similarly to the $Ik\mbox{-}NN$
%algorithm, is based on our filter-refinement framework.
During the filter steps, the tree is traversed in a best first
manner, where entries are accessed by their minimal distance
(MinDist) to the farthest query object. For each entry $e$ we
check if $e$ is completely contained in the union of pruning
regions defined by all pairs of queries $(q_i, q_j)\in Q$; i.e.,
$\bigcup_{(q_i, q_j)\in Q}PR_{q_i}(q_j)$.
% and $PR_{q_j}(q_i)$ defined in Definition
%\ref{def:pruning_region} (line
%\ref{alg:cbp}).
In addition, for each accessed database object $o_i$ and each
query object $q_j$, the pruning region is extended by
$PR_{q_j}(o_i)$. Analogously to the $Ik\mbox{-}NN$ case, lists for
the candidates and pruned entries are maintained. The pruning
conditions of Appendix \ref{appn:sky_2d} are used wherever
applicable to reduce the computational cost. Finally, the
remaining candidates are refined using the refinement strategy
described in Section \ref{subsec:FrameworkImplementation_skyline}.

\begin{figure}[t]
  \centering
  \includegraphics[width=0.2\textwidth]{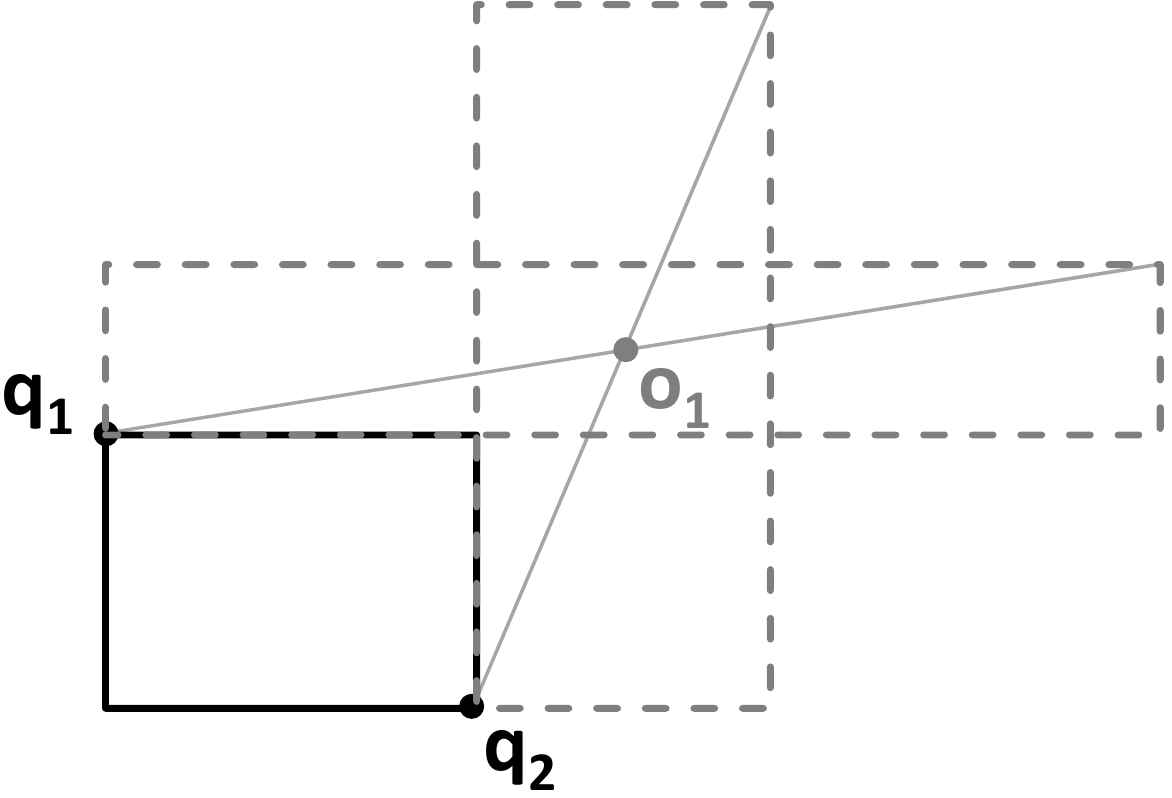}
\vspace{-3mm}
  \caption{Refinement area defined by $q_1, q_2$ and $o_1$}
  \label{fig:refine}
\end{figure}

\vspace{-3mm}
\section{Experiments}
\label{sec:experiments}

For each of the inverse query predicates discussed in the paper,
we compare our proposed solution based on multi-query-filtering
(MQF), with a naive approach (Naive) and another intuitive
approach based on single-query-filtering (SQF). The naive
algorithm (Naive) computes the corresponding reverse query for
every $q\in Q$ and intersects their results iteratively. To be
fair, we terminated Naive as soon as the intersection of results
obtained so far is empty. SQF
 performs a R$k$NN (R$\eps$-range / RDS) query using one randomly
chosen query point as a filter step to obtain candidates. For each
candidate an $\eps$-range ($k$NN / DS) query is issued and the
candidate is confirmed if all query points are contained in the
result of the query (refinement step). Since the pages accessed by
the queries in the refinement step are often redundant, we use a
buffer to further boost the performance of SQF. We employed
$R^*$-trees (\cite{BecKriSchSee90}) of pagesize 1Kb to index the
datasets used in the experiments. For each method, we present the
number of page accesses and runtime. To give insights into the
impact of the different parameters on the cardinality of the
obtained results we also included this number to the charts. In
all settings we performed 1000 queries and averaged the results.
All methods were implemented in Java 1.6 and tests were run on a
dual core (3.0 Ghz) workstation with 2 GB main memory having
windows xp as OS.
%the performance analysis we measured and compared the number of
%required page accesses. In addition, we evaluated the required
%number of distance computations.
The performance evaluation settings are summarized below; the
numbers in \textbf{bold} correspond to the default settings:

\begin{center}
  \begin{tabular}{| l | c |}
    \hline
    \textbf{parameter} & \textbf{values}\\ \hline \hline
    db size & 100000 (synthetic), 175812 (real) \\ \hline
    dimensionality & 2, \textbf{3}, 4, 5 \\ \hline
    $\eps$ & 0.04, 0.05, \textbf{0.06}, 0.07, 0.08, 0.09, 0.1 \\ \hline
    $k$ & 50, \textbf{100}, 150, 200, 250 \\ \hline
    \# inverse queries  & 1, 3, 5, \textbf{10}, 15, 20, 25, 30, 35 \\ \hline
    query extent  &  0.0001, 0.0002, 0.0003, \textbf{0.0004}, 0.0005,
    0.0006\\
    \hline
%    pagesize  & 1024 byte \\ \hline
  \end{tabular}
\end{center}
The experiments were performed using several datasets:
\begin{itemize}
  \item Synthetic Datasets: Clustered and uniformly distributed objects
  in $d$-dimensional space.
  \item Real Dataset: Vertices in the Road Network of North America \footnote{Obtained and modified from
\emph{http://www.cs.fsu.edu/$\sim$lifeifei/SpatialDataset.htm}.
The original source is the \emph{Digital Chart of the World Server
(http://www.maproom.psu.edu/dcw/)}.}. Contains 175,812
two-dimensional points.
\end{itemize}
The datasets were normalized, such that their minimum bounding box
is $[0,1]^d$. For each experiment, the query objects $Q$ for the
inverse query were chosen randomly from the database. Since the
number of results highly depends on the distance between inverse
query points (in particular for the $I\eps\mbox{-}RQ$ and
$Ik\mbox{-}NNQ$) we introduced an additional parameter called
\textit{extent} to control the maximal distance between the query
objects. The value of \textit{extent} corresponds to the volume
(fraction of data space) of a cube that minimally bounds all
queries. For example in the 3D space the default cube would have a
side length of 0.073. A small \textit{extent} assures that the
queries are placed close to each other generally resulting in more
results. In this section, we show the  behavior of all three
algorithms on the uniform datasets only. Experiments on the other
datasets can be found in Appendix~\ref{appn:exp}.

\subsection{Inverse $\eps$-Range Queries}
\label{subsec:experiments_range} We first compared the algorithms
on inverse $\eps$ range queries. Figure \ref{fig:eps-e} shows that
the relative speed of our approach (MQF) compared to Naive grows
significantly with increasing $\eps$; for Naive, the cardinality
of the result set returned by each query depends on the space
covered by the hypersphere which is in $O(\eps^d)$. In contrast,
our strategy applies spatial pruning early, leading to a low
number of page accesses.  SQF  is faster than Naive, but still
needs around twice as much page accesses as MQF. MQF  performs
even better with an increasing number of query points in $Q$ (as
depicted in Figure \ref{fig:eps-q}), as in this case the
intersection of the ranges becomes smaller. The I/O cost of SQF in
this case remains almost constant which is mainly due to the use
of the buffer which lowers the page accesses in the refinement
step. Similar results can be observed when varying the database
size (Figure \ref{fig:eps-scale}) and query extent (Figure
\ref{fig:eps-ext}). For the data dimensionality experiment
(Figure \ref{fig:eps-dim}) we set epsilon such that the sphere
defined by $\eps$ covers always the same percentage of the
dataspace, to make sure that we still obtain results when
increasing the dimensionality (note, however, that the number of
results is still unsteady). Increasing dimensionality has a
negative effect on performance.
%  some
%dimensionality effects, which is shown by the unsteady number of results.
However MQF  copes better with  data dimensionality than the other
approaches. In a last experiment (see Figure \ref{fig:eps-cpu}) we
compared the computational costs of the algorithms. Even though
Inverse Queries are I/O bound, MQF is still preferable for
main-memory problems.

\begin{figure}[h]
    \centering
    \subfigure[I/O cost w.r.t. $\eps$.]{
        \label{fig:eps-e}
        \includegraphics[width = 0.3\columnwidth]{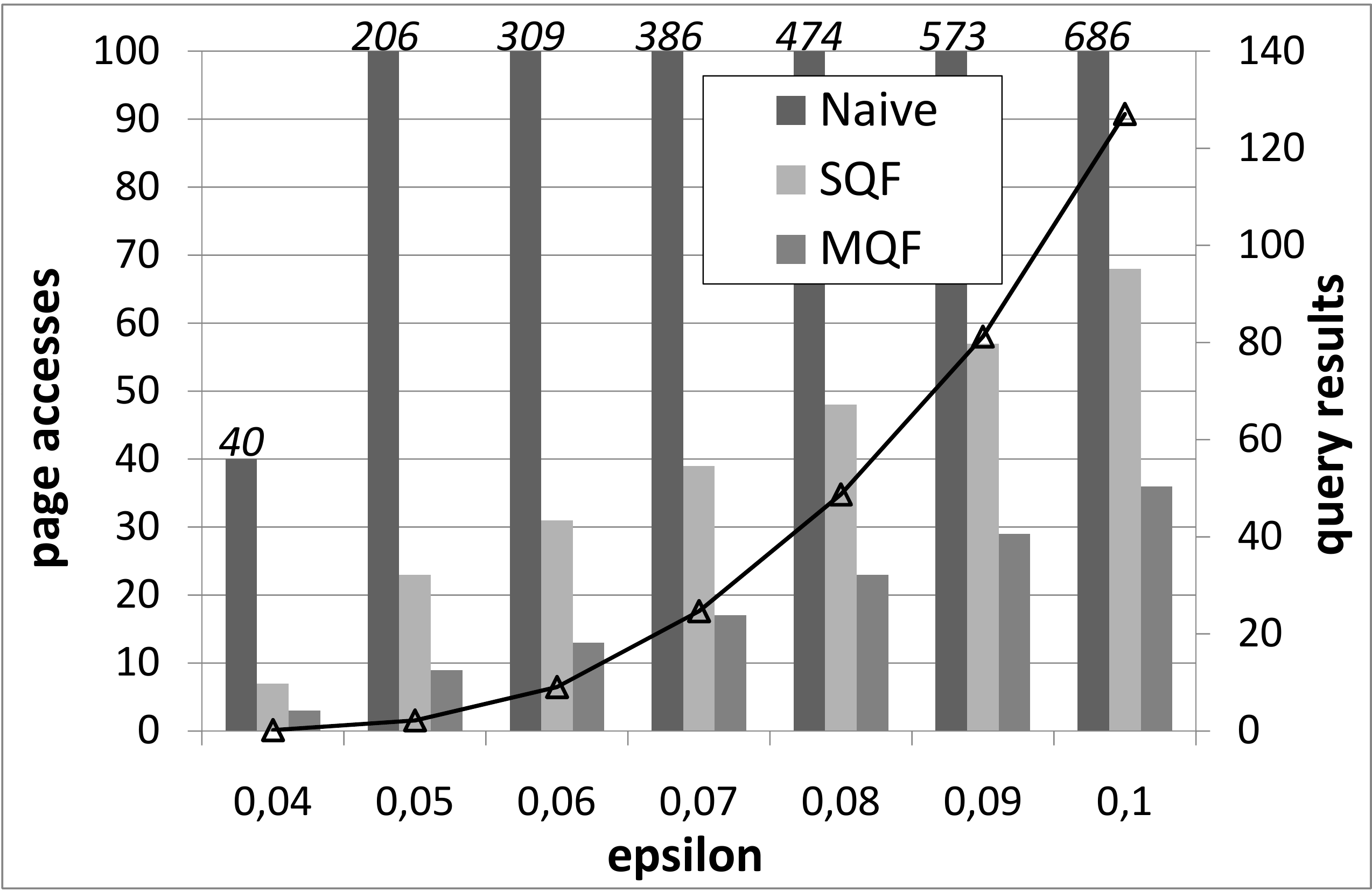}
    }
    \subfigure[I/O cost w.r.t. $|Q|$.]{
        \label{fig:eps-q}
        \includegraphics[width = 0.3\columnwidth]{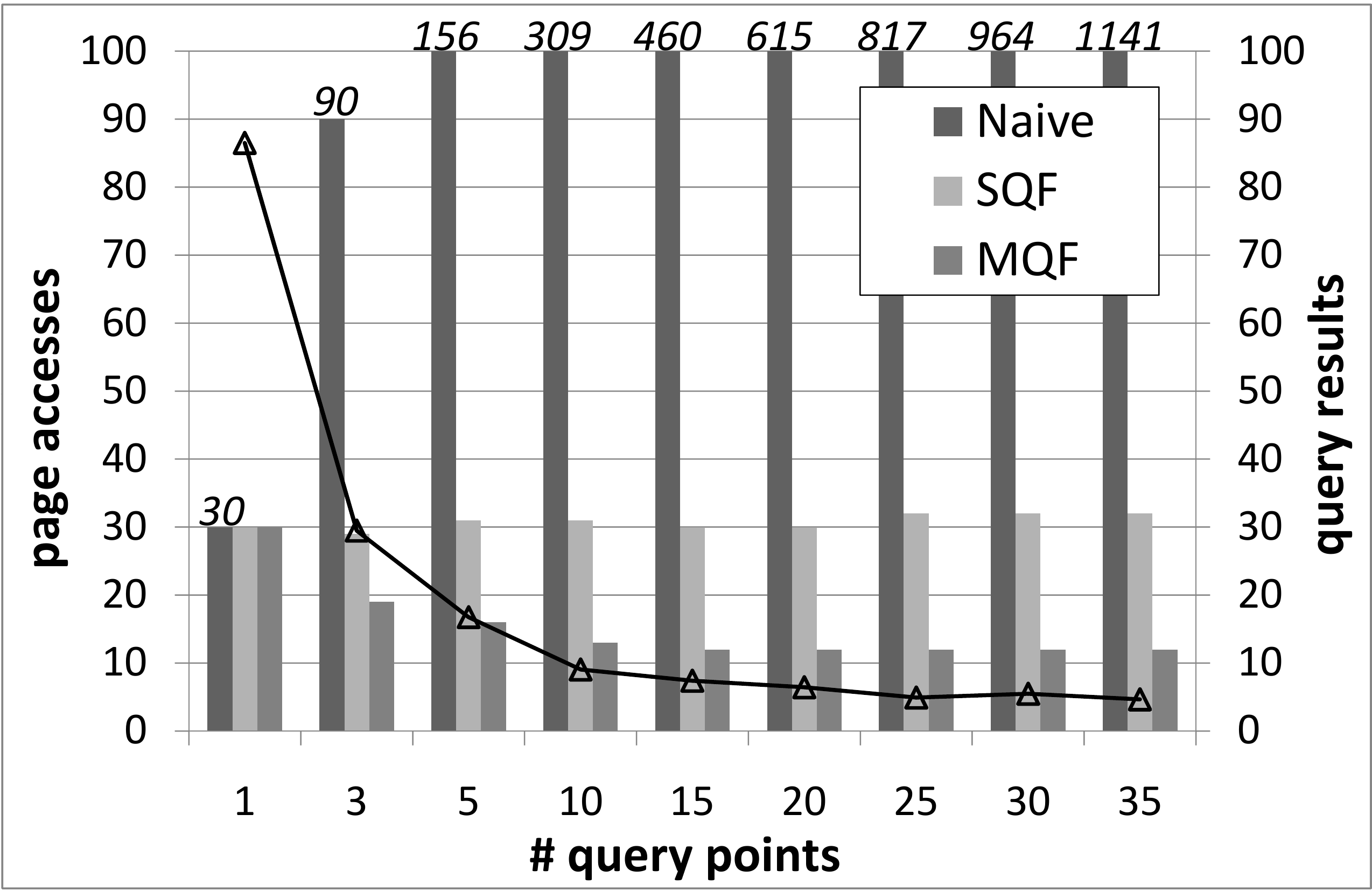}
    }
    \subfigure[I/O cost w.r.t. $d$.]{
        \label{fig:eps-dim}
        \includegraphics[width = 0.3\columnwidth]{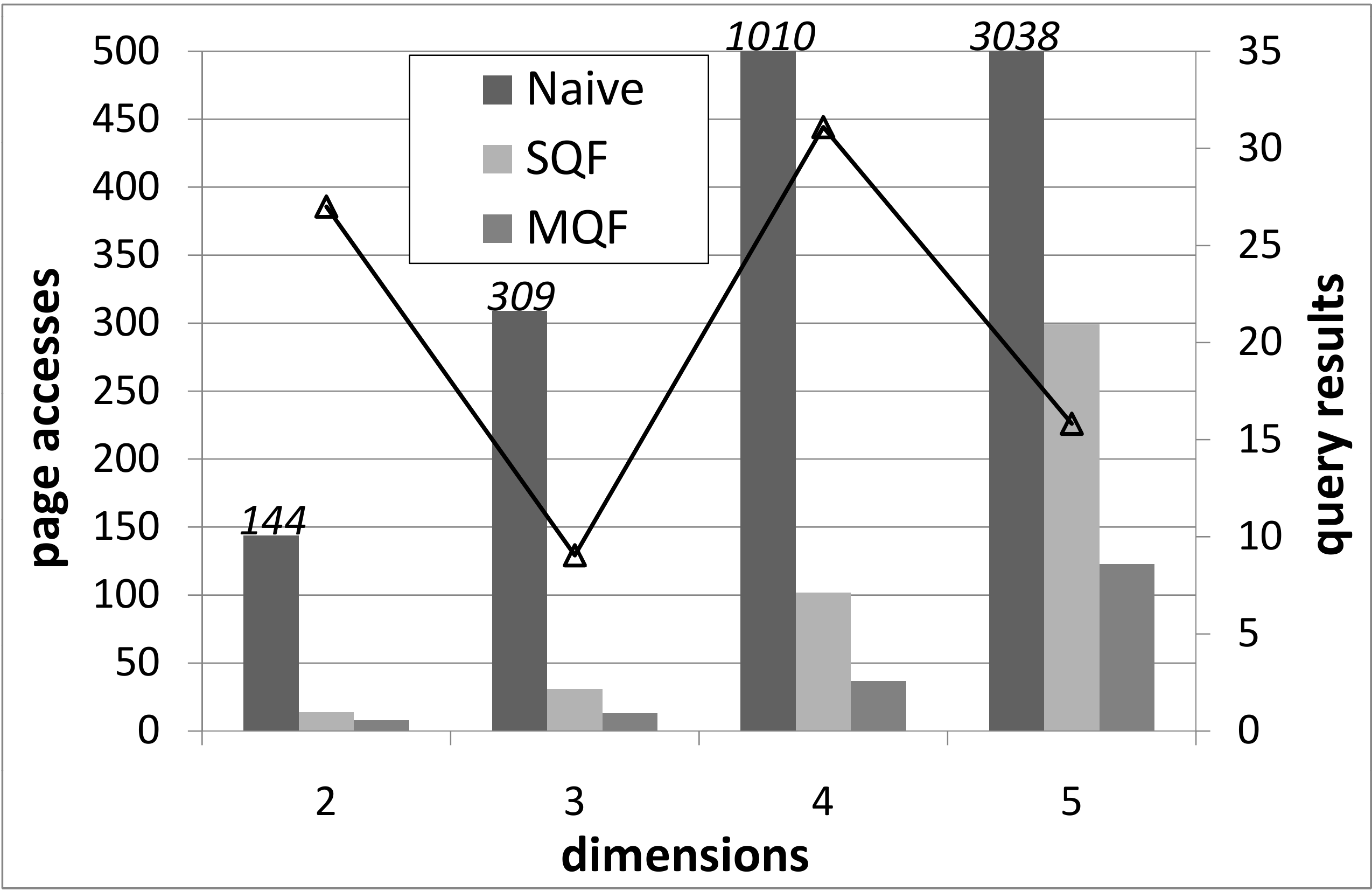}
    }
    \subfigure[I/O cost w.r.t. extent.]{
        \label{fig:eps-ext}
        \includegraphics[width = 0.3\columnwidth]{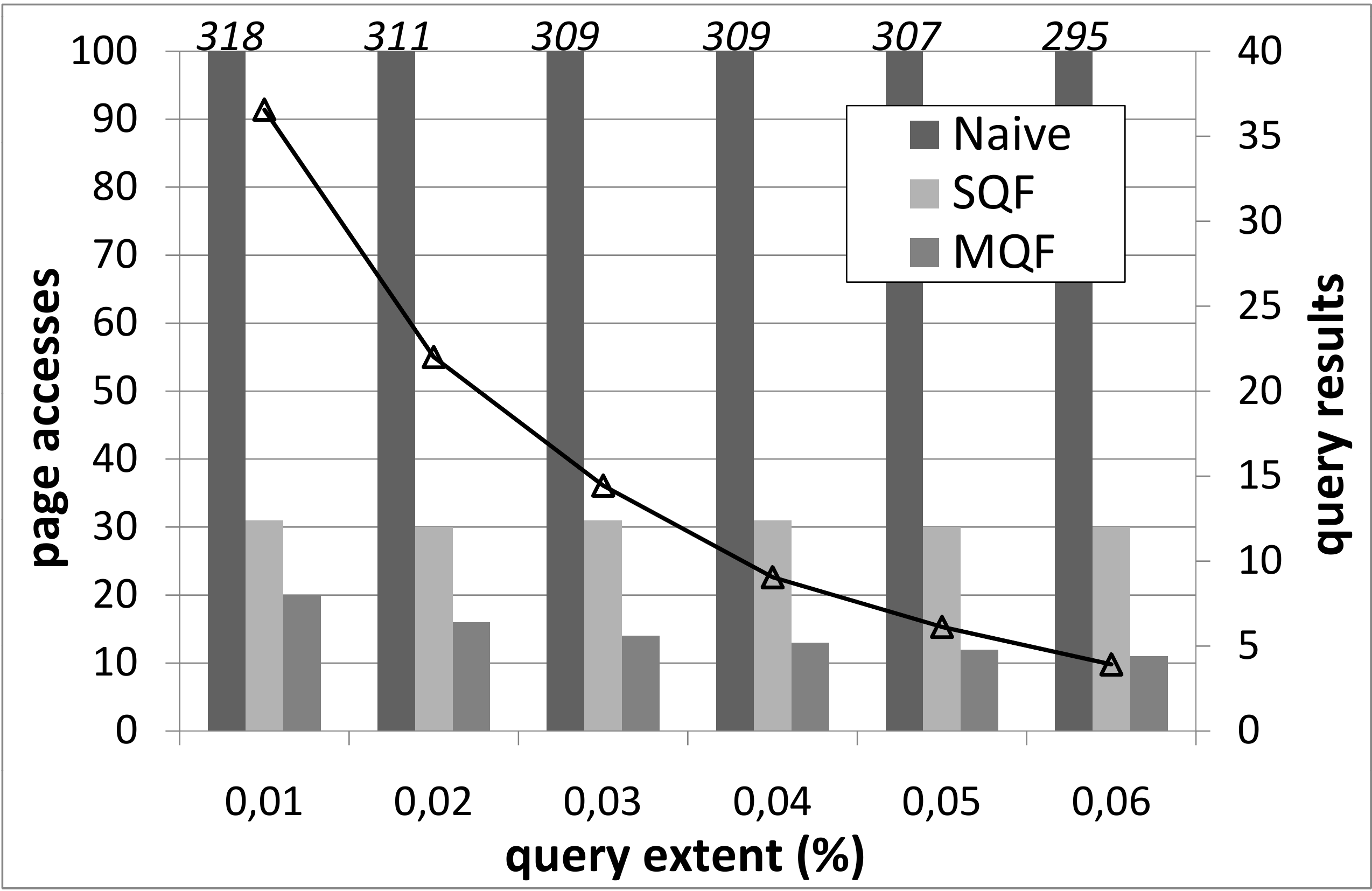}
    }
    \subfigure[I/O cost w.r.t. $|\DB|$.]{
        \label{fig:eps-scale}
        \includegraphics[width = 0.3\columnwidth]{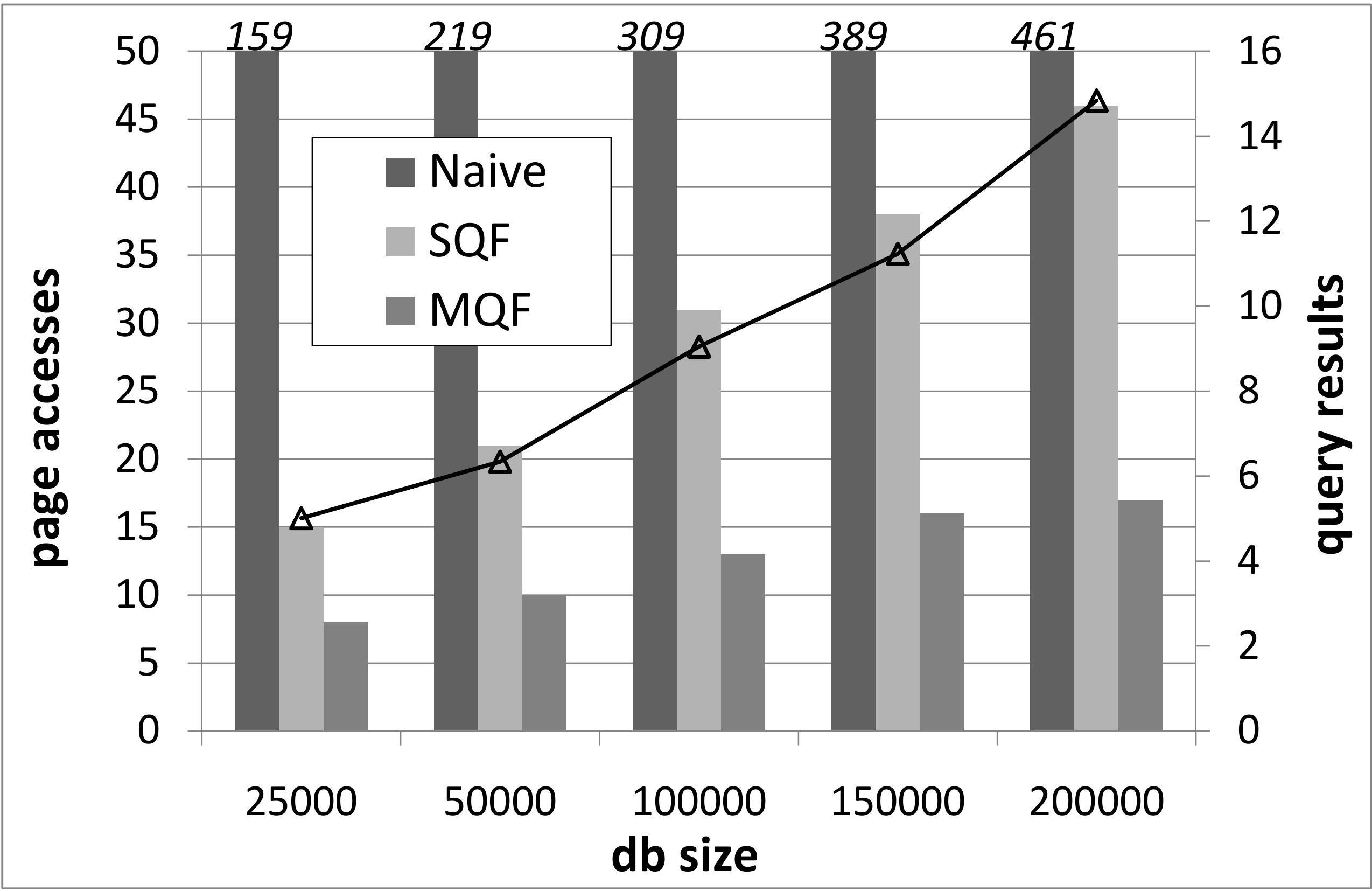}
    }
    \subfigure[CPU cost w.r.t. $|Q|$.]{
        \label{fig:eps-cpu}
        \includegraphics[width = 0.3\columnwidth]{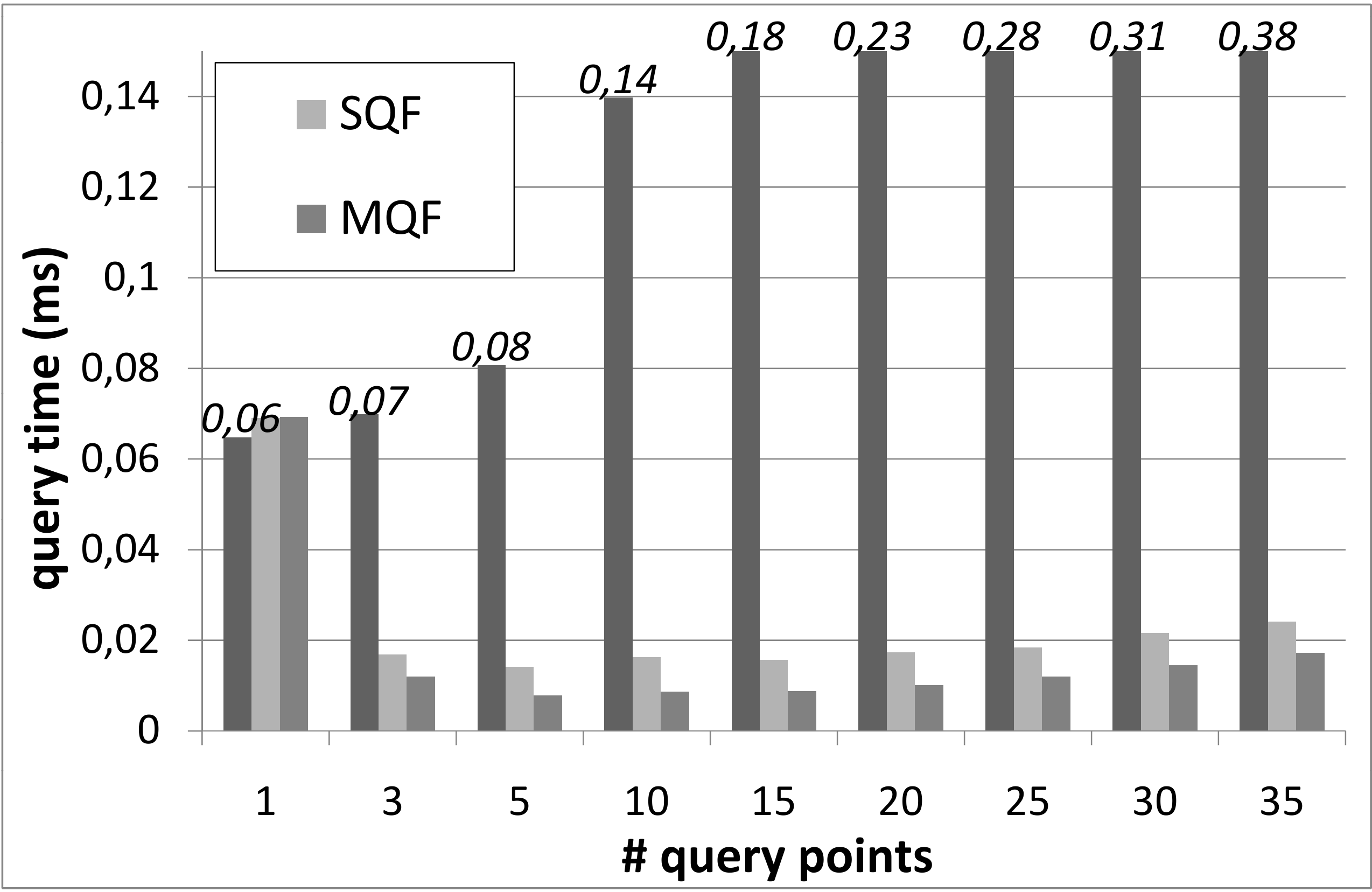}
    }
\vspace{-3mm}
    \caption{$I\eps\mbox{-}Q$ algorithms on uniform dataset}
\end{figure}

\subsection{Inverse $k$NN Queries}
\label{subsec:experiments_knn}

The three approaches for inverse $k$NN search show a similar
behavior as those for the I$\eps$-RQ. Specifically the behavior
for varying $k$ (Figure \ref{fig:knn-k}) is comparable to varying
$\eps$ and increasing the query number (Figure \ref{fig:knn-q})
and the query extent (Figure \ref{fig:knn-ext}) yields the
expected results. When testing on datasets with different
dimensionality, the advantage of MQF becomes even more significant
when $d$ increases (cf. Figure \ref{fig:knn-dim}). In contrast to
the I$\eps$-RQ results for I$k$NN queries the page accesses of MQF
decrease (see Figure \ref{fig:knn-scale}) when the database size
increases (while the performance of SQF still degrades). This can
be explained by the fact, that the number of pages accessed is
strongly correlated with the number of obtained results. Since for
the I$\eps$-RQ the parameter $\eps$ remained constant, the number
of results increased with a larger database. For I$k$NN the number
of results in contrast decreases and so does the number of
accessed pages by MQF. As in the previous set of experiments MQF
has also the lowest runtime (Figure \ref{fig:knn-cpu}).

\begin{figure}[h]
    \centering
    \subfigure[I/O cost w.r.t. $k$.]{
        \label{fig:knn-k}
        \includegraphics[width = 0.3\columnwidth]{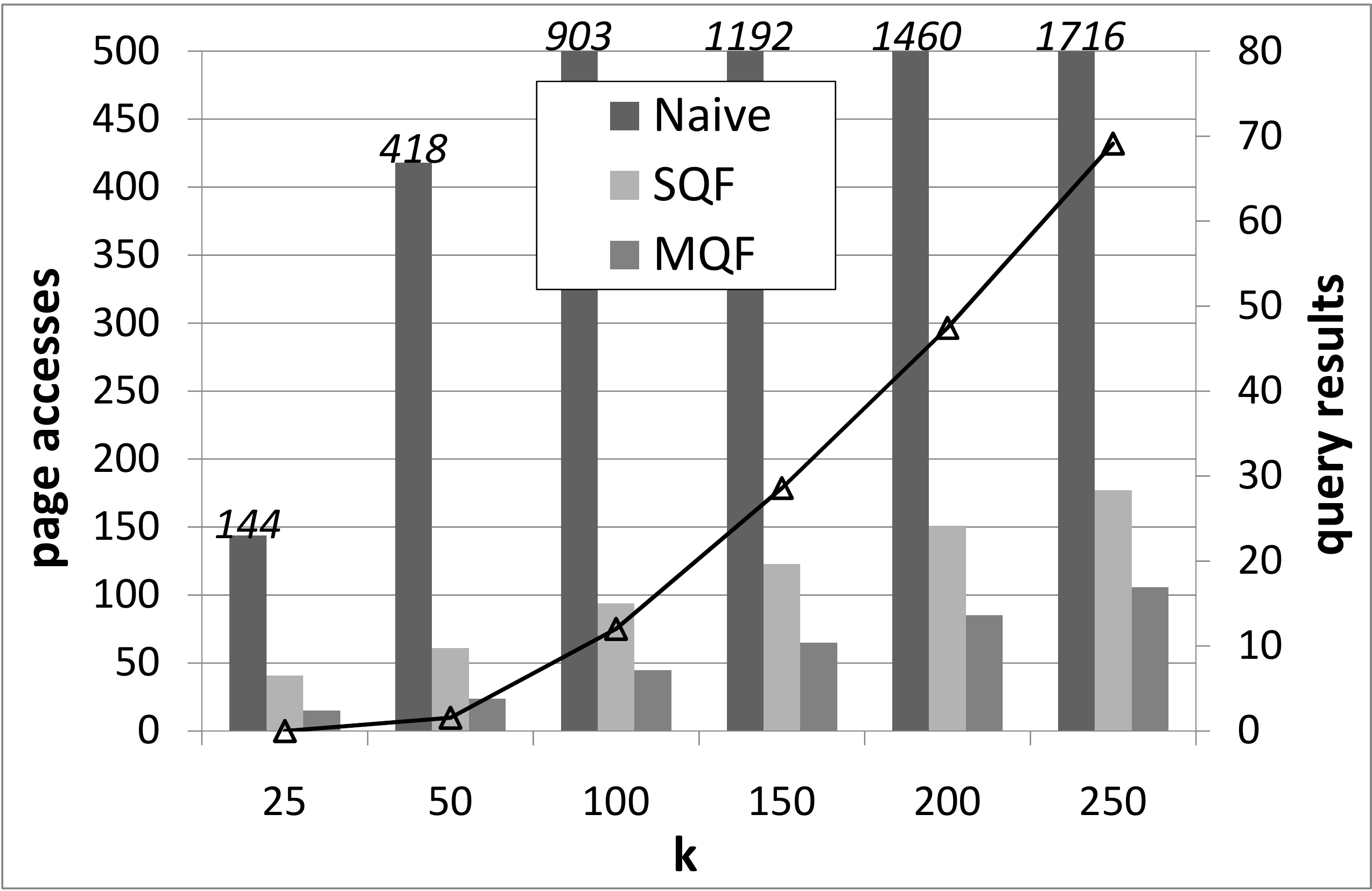}
    }
    \subfigure[I/O cost w.r.t. $|Q|$.]{
        \label{fig:knn-q}
        \includegraphics[width = 0.3\columnwidth]{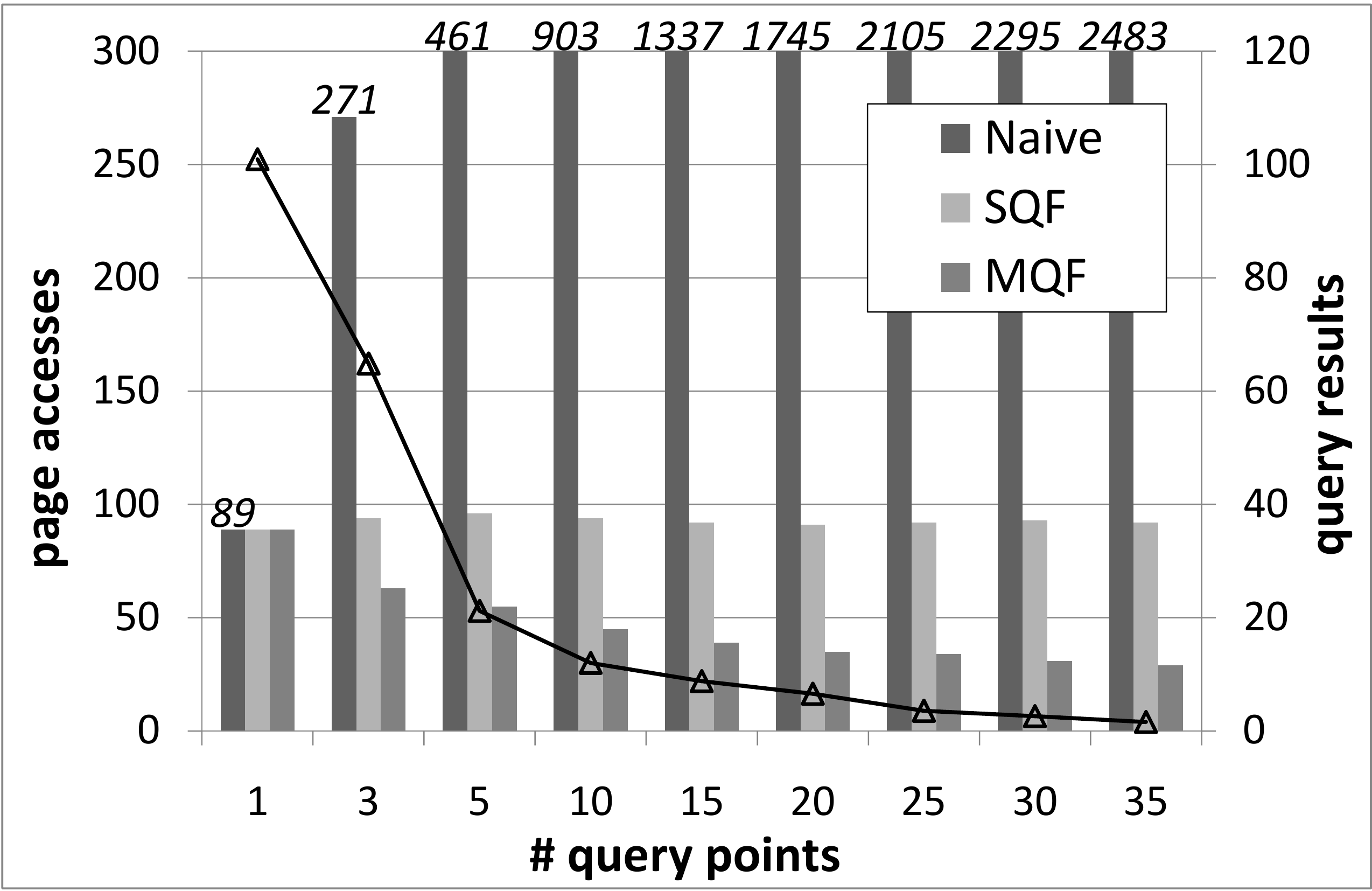}
    }
    \subfigure[I/O cost w.r.t. $d$.]{
        \label{fig:knn-dim}
        \includegraphics[width = 0.3\columnwidth]{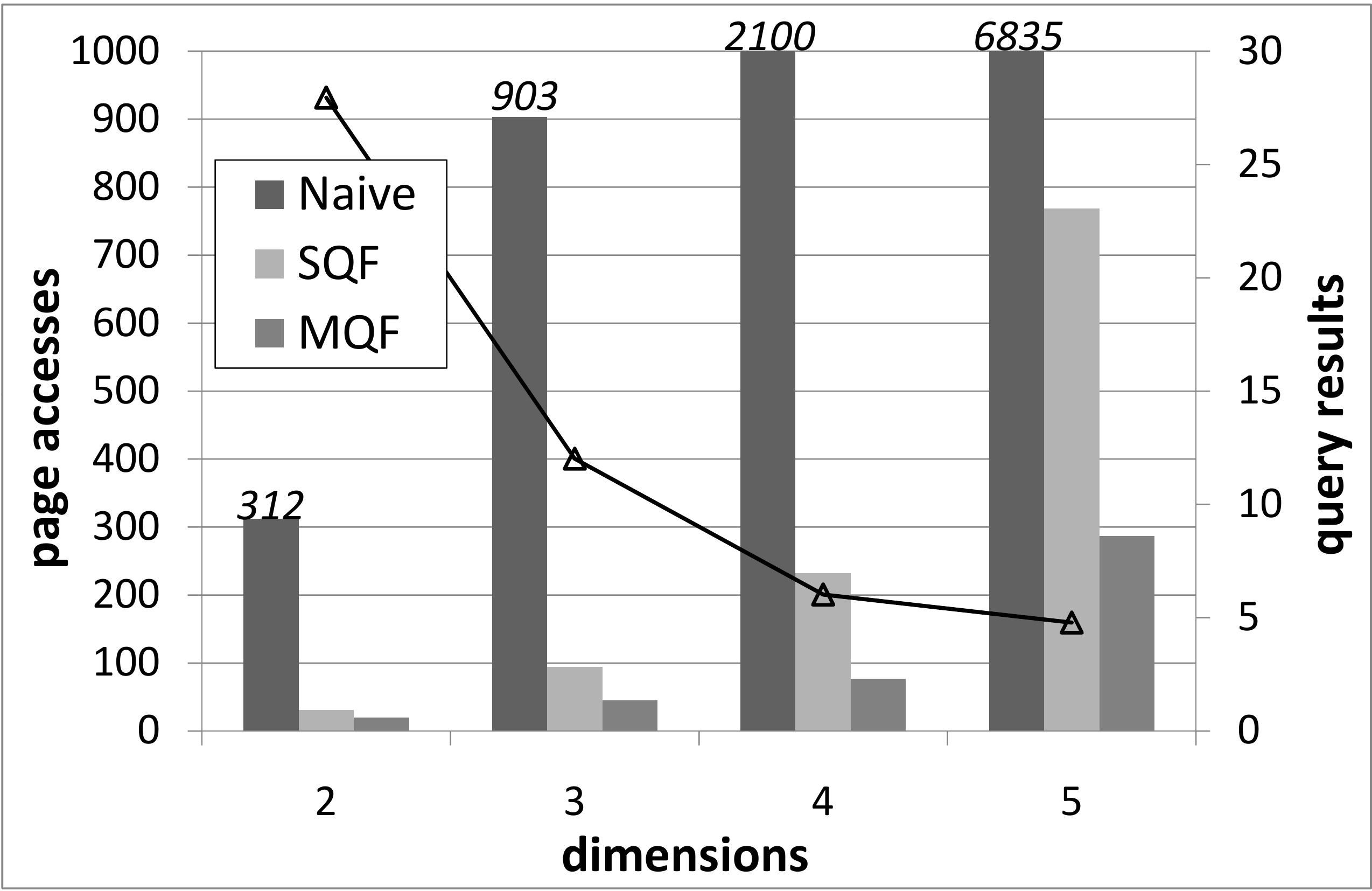}
    }
    \subfigure[I/O cost w.r.t. extent.]{
        \label{fig:knn-ext}
        \includegraphics[width =
        0.3\columnwidth]{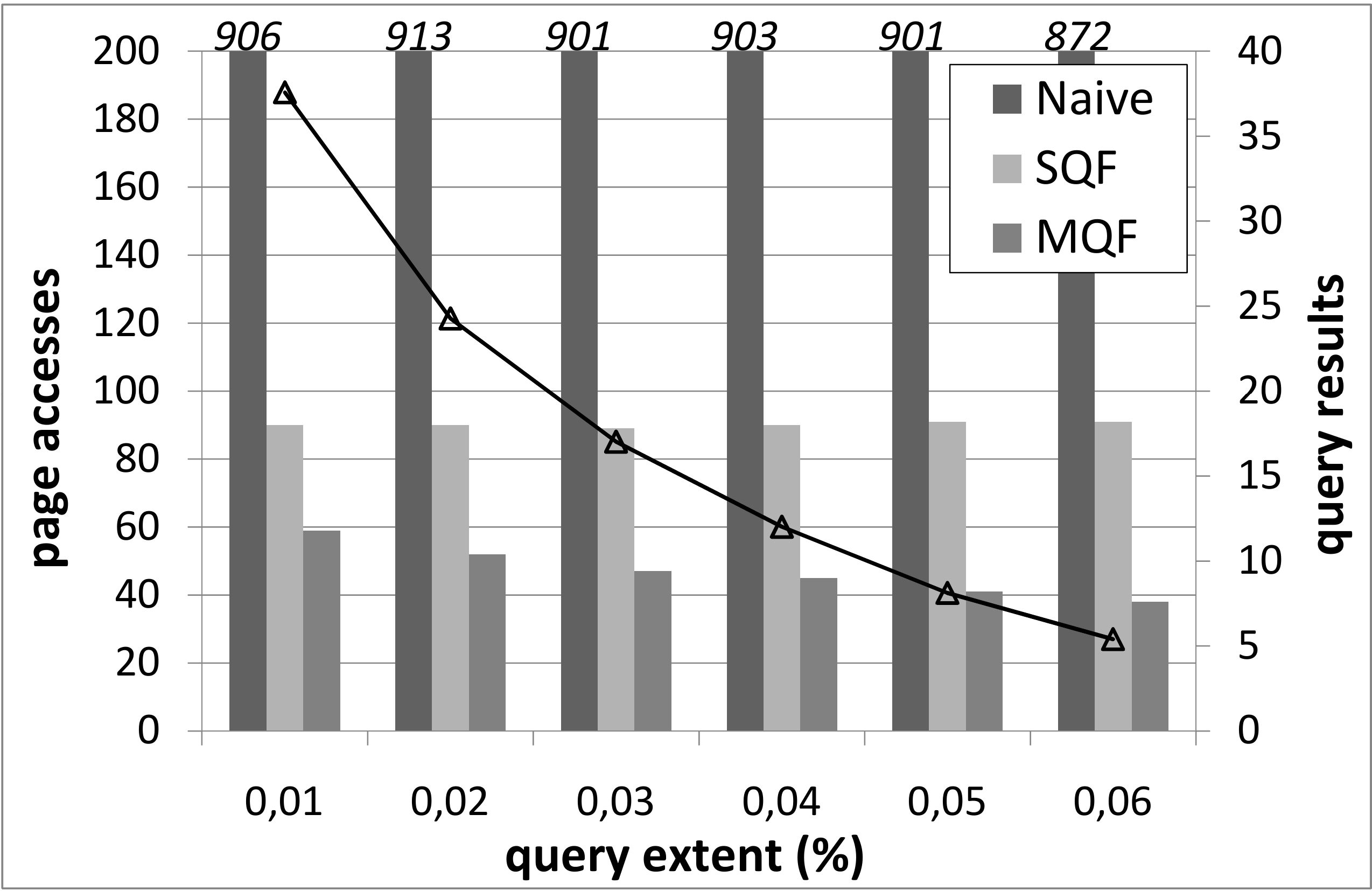} }
    \subfigure[I/O cost w.r.t. $|\DB|$.]{
        \label{fig:knn-scale}
        \includegraphics[width = 0.3\columnwidth]{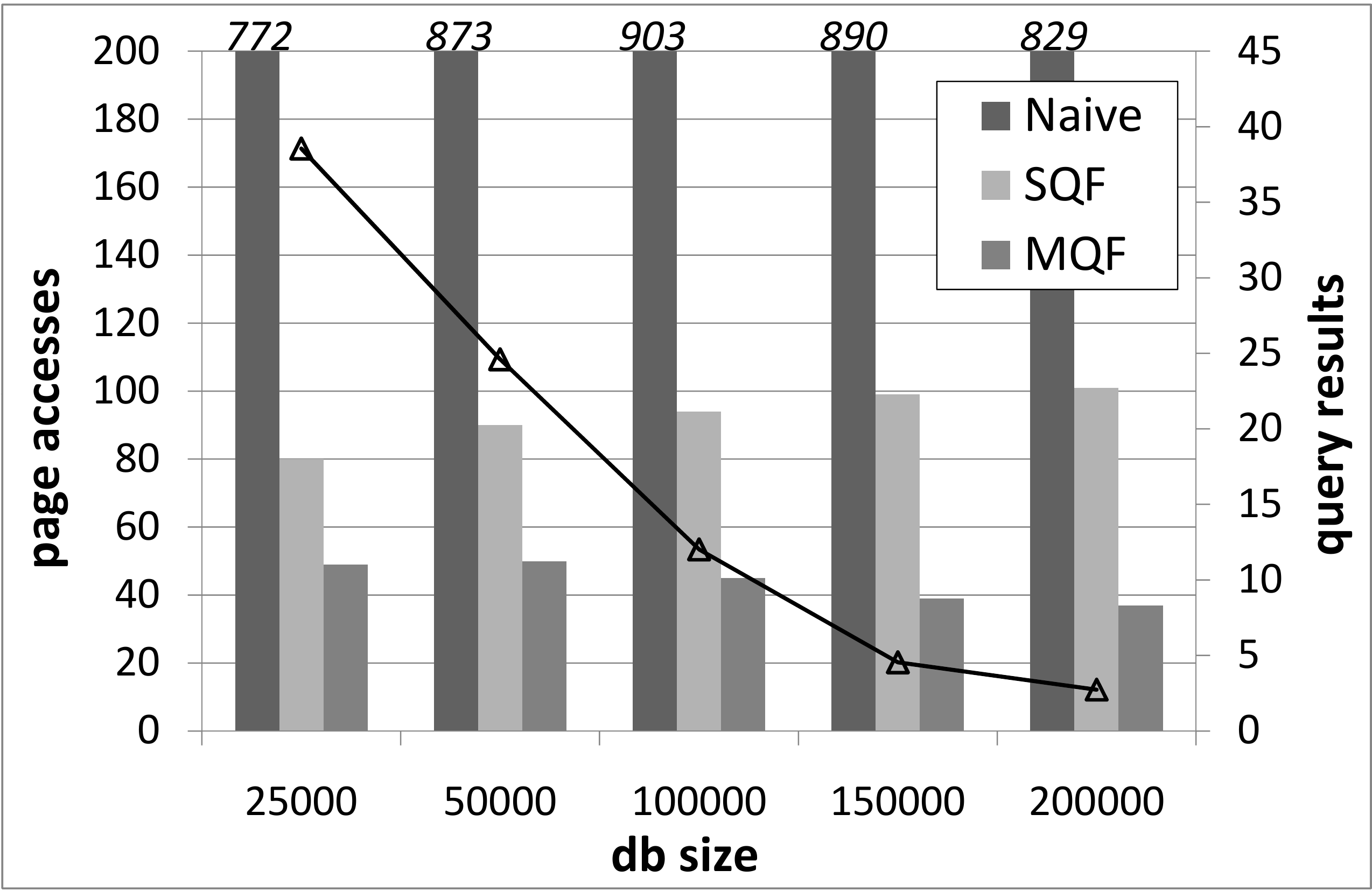}
    }
    \subfigure[CPU cost w.r.t. $|Q|$.]{
        \label{fig:knn-cpu}
        \includegraphics[width = 0.3\columnwidth]{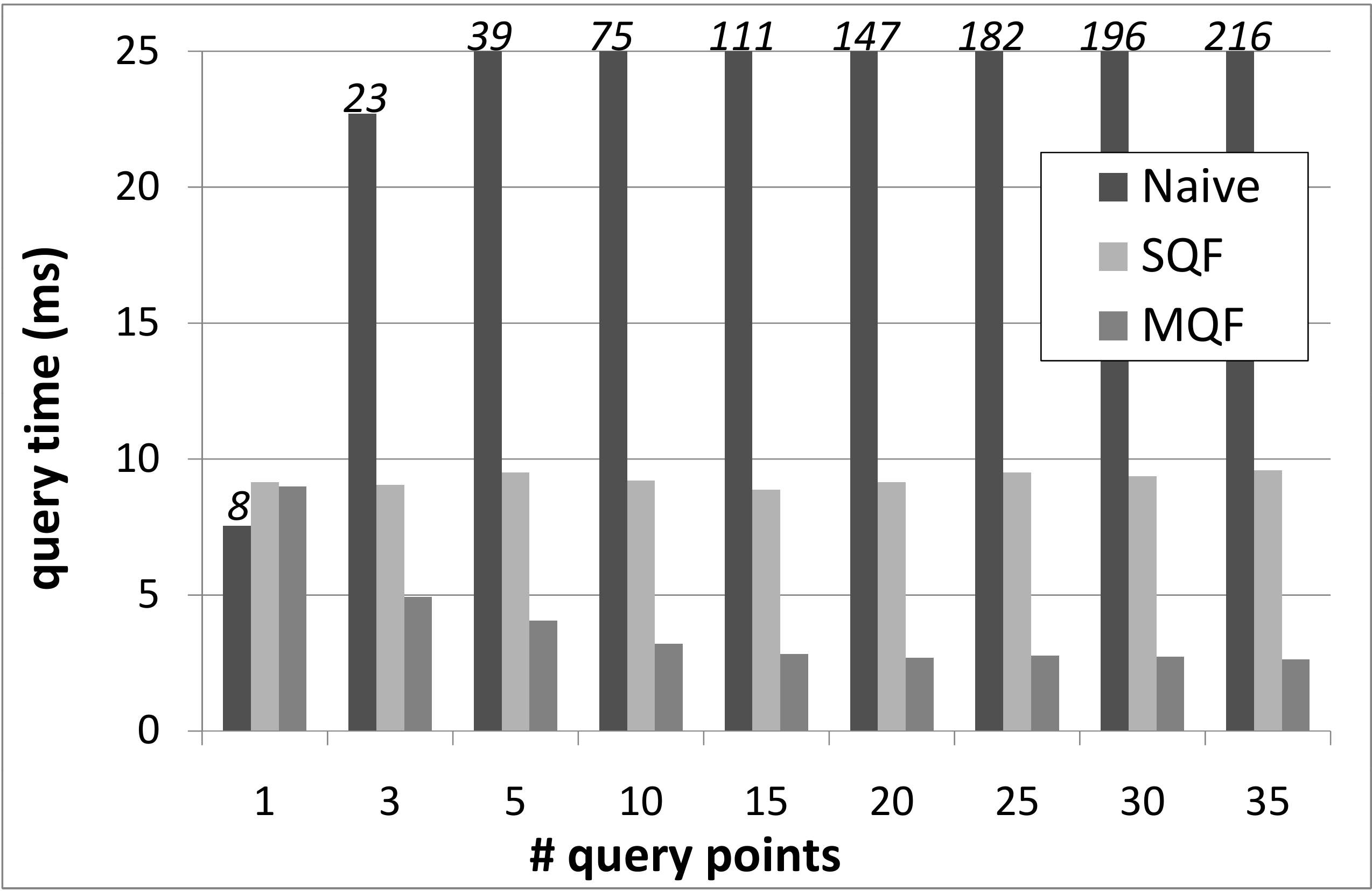}
    }
\vspace{-3mm}
    \caption{$Ik\mbox{-}NNQ$ algorithms on uniform dataset}
\end{figure}

\subsection{Inverse Dynamic Skyline Queries}
\label{subsec:experiments_skyline} Similar results as for the
$Ik\mbox{-}NNQ$ algorithm are obtained for the inverse dynamic
skyline queries ($I\mbox{-}DSQ$). Increasing the number of queries
in $Q$ reduces the cost of the MQF approach while the costs of the
competitors  increase. Since the average number of results
approaches 0 faster than for the other two types of inverse
queries we choose 4 as the default size of the query set. Note
that the number of results for $I\mbox{-}DSQ$ intuitively
increases exponentially with the dimensionality of the dataset
(cf. Figure \ref{fig:skyline-dim}), thus this value can be much
larger for higher dimensional datasets. Increasing the distance
among the queries does not affect the performance as seen in
Figure \ref{fig:skyline-ext}; regarding the number of results in
contrast to inverse range- and $k$NN-queries, inverse dynamic
skyline queries are almost not sensitive to the distance among the
query points. The rationale  is that dynamic skyline queries can
have results which are arbitrary far away from the query point,
thus the same holds for the inverse case. The same effect can be
seen for increasing database size (cf. Figure
\ref{fig:skyline-scale}). The advantage of MQF remains constant
over the other two approaches. Like inverse range- and
$k$NN-queries, $I\mbox{-}DSQ$ are I/O bound (see Figure
\ref{fig:skyline-cpu}), but MQF is still preferable for
main-memory problems.

\begin{figure}[h]
    \centering
    \subfigure[I/O cost w.r.t. $|Q|$.]{
        \label{fig:skyline-q}
        \includegraphics[width = 0.3\columnwidth]{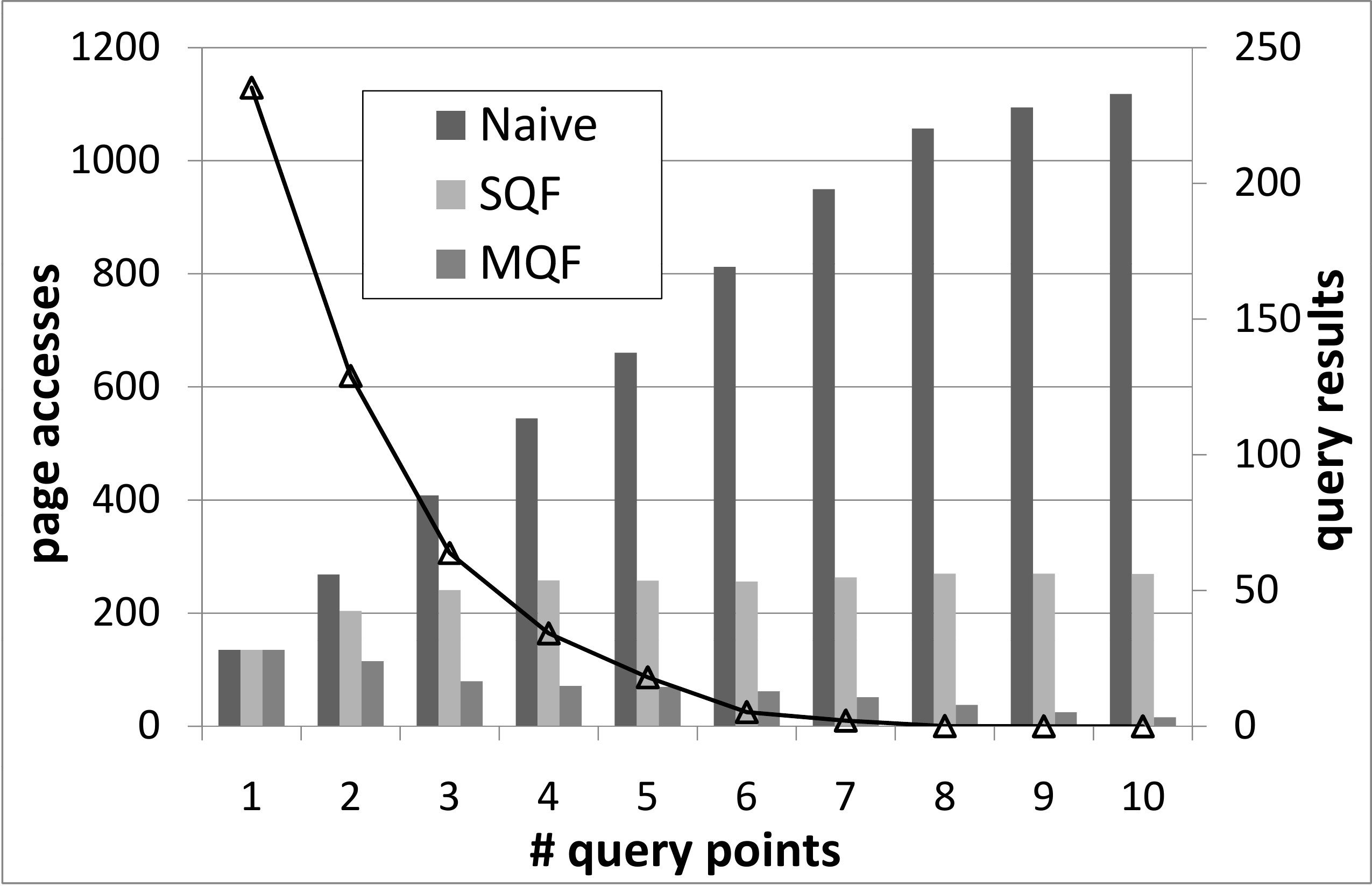}
    }
    \subfigure[I/O cost w.r.t. $d$.]{
        \label{fig:skyline-dim}
        \includegraphics[width = 0.3\columnwidth]{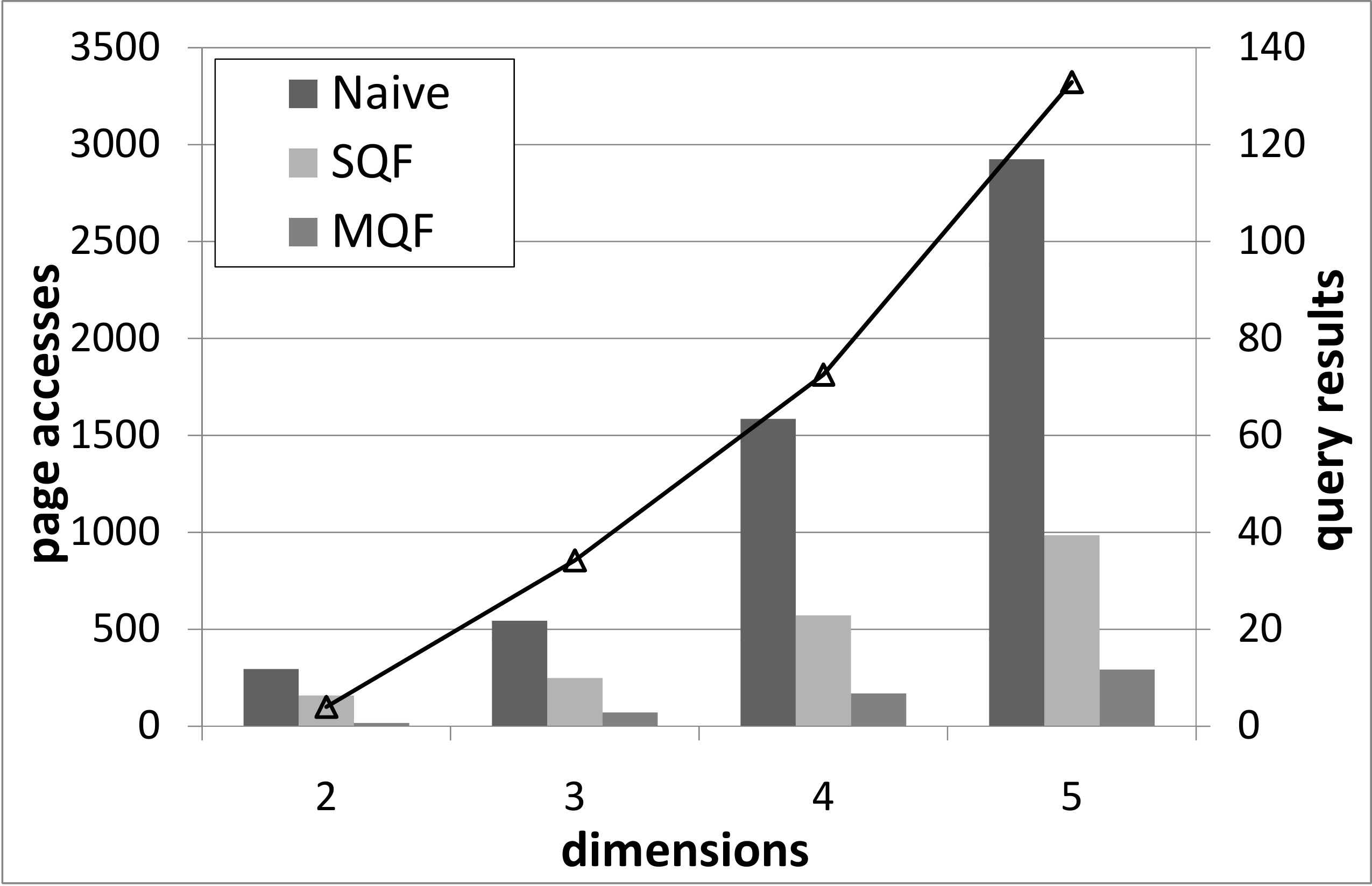}
    }
    \subfigure[I/O cost w.r.t. extent.]{
        \label{fig:skyline-ext}
        \includegraphics[width =
        0.3\columnwidth]{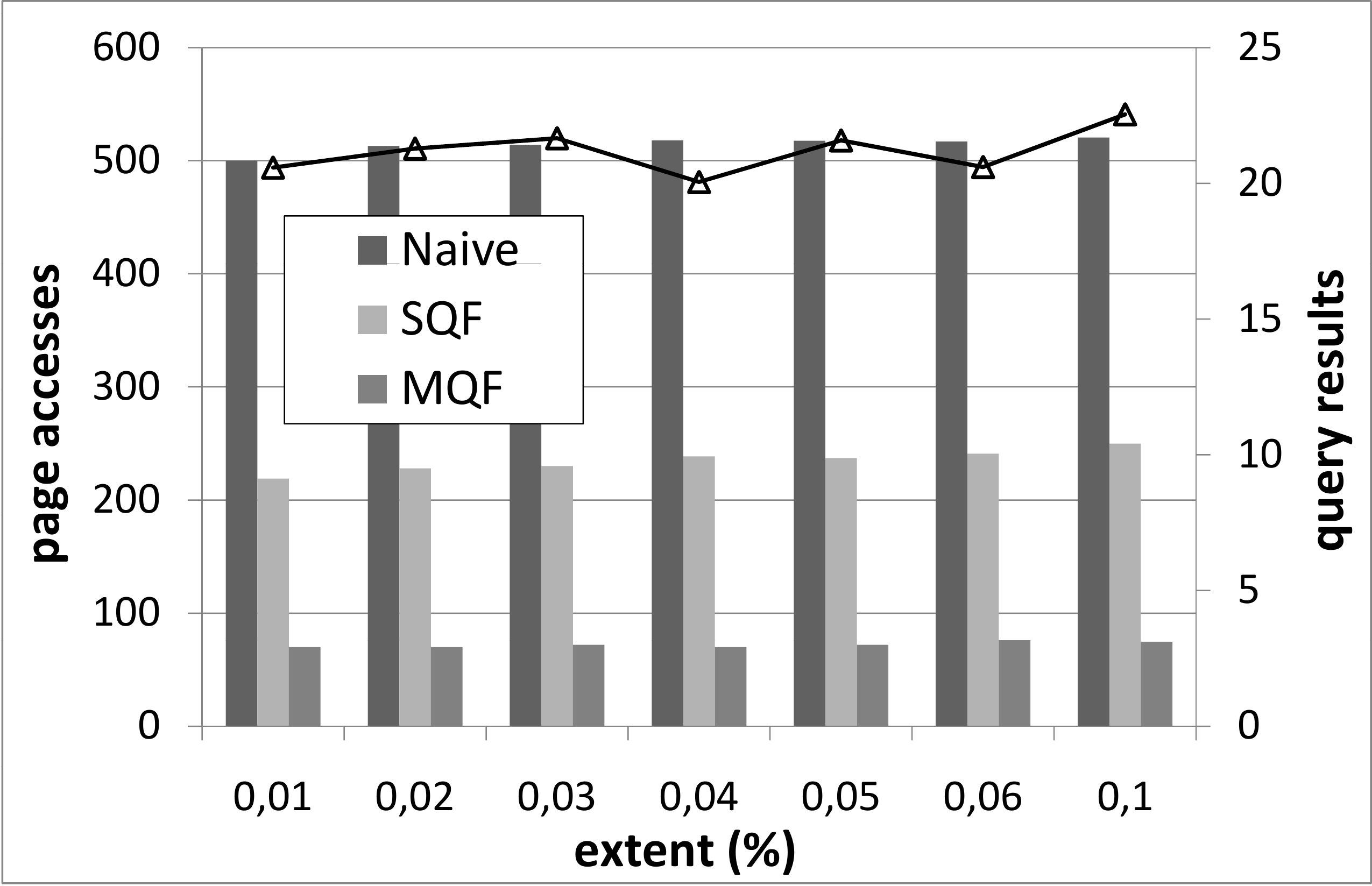} }
    \subfigure[I/O cost w.r.t. $|\DB|$.]{
        \label{fig:skyline-scale}
        \includegraphics[width = 0.3\columnwidth]{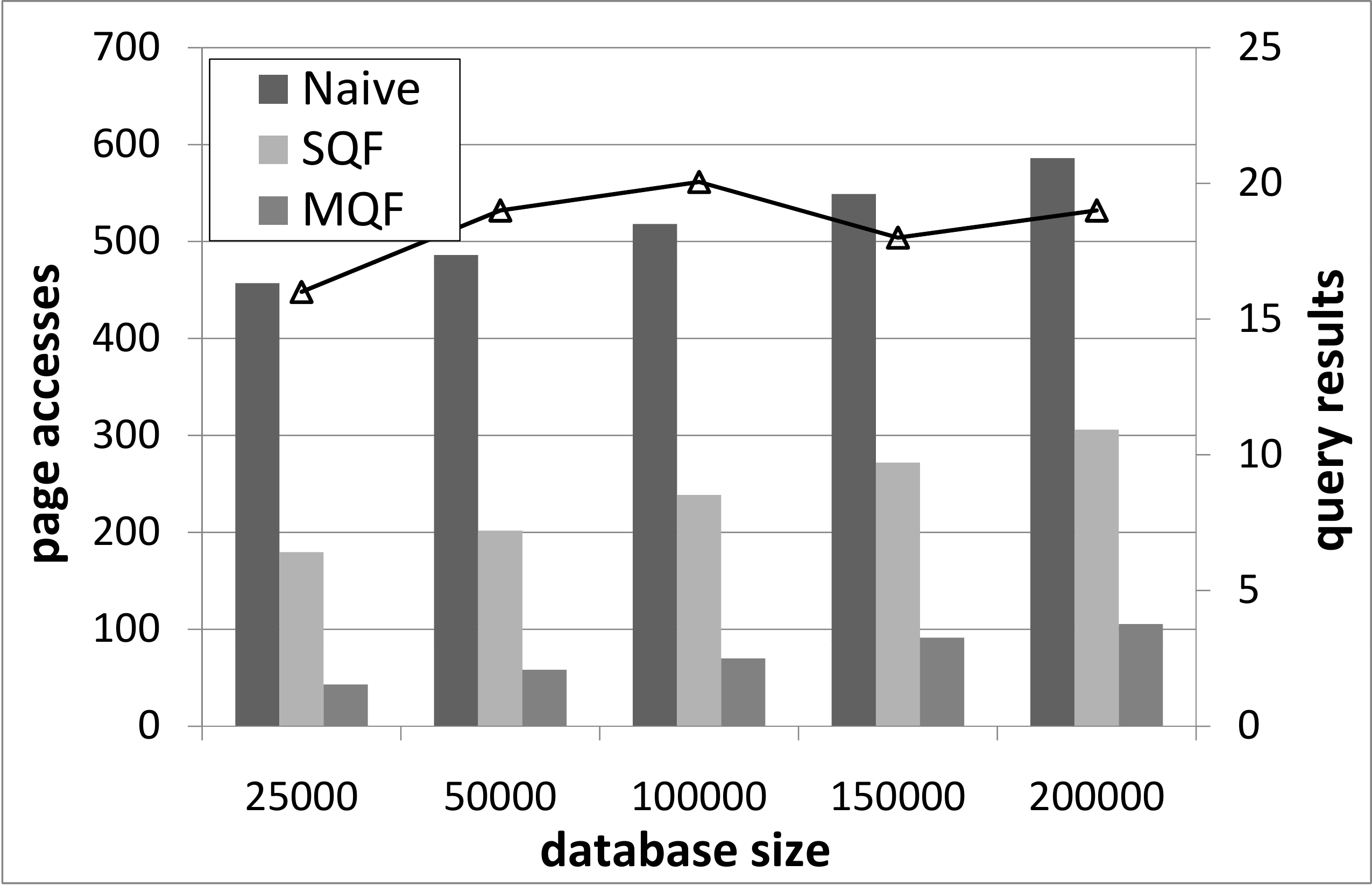}
    }
    \subfigure[CPU cost w.r.t. $|Q|$.]{
        \label{fig:skyline-cpu}
        \includegraphics[width = 0.3\columnwidth]{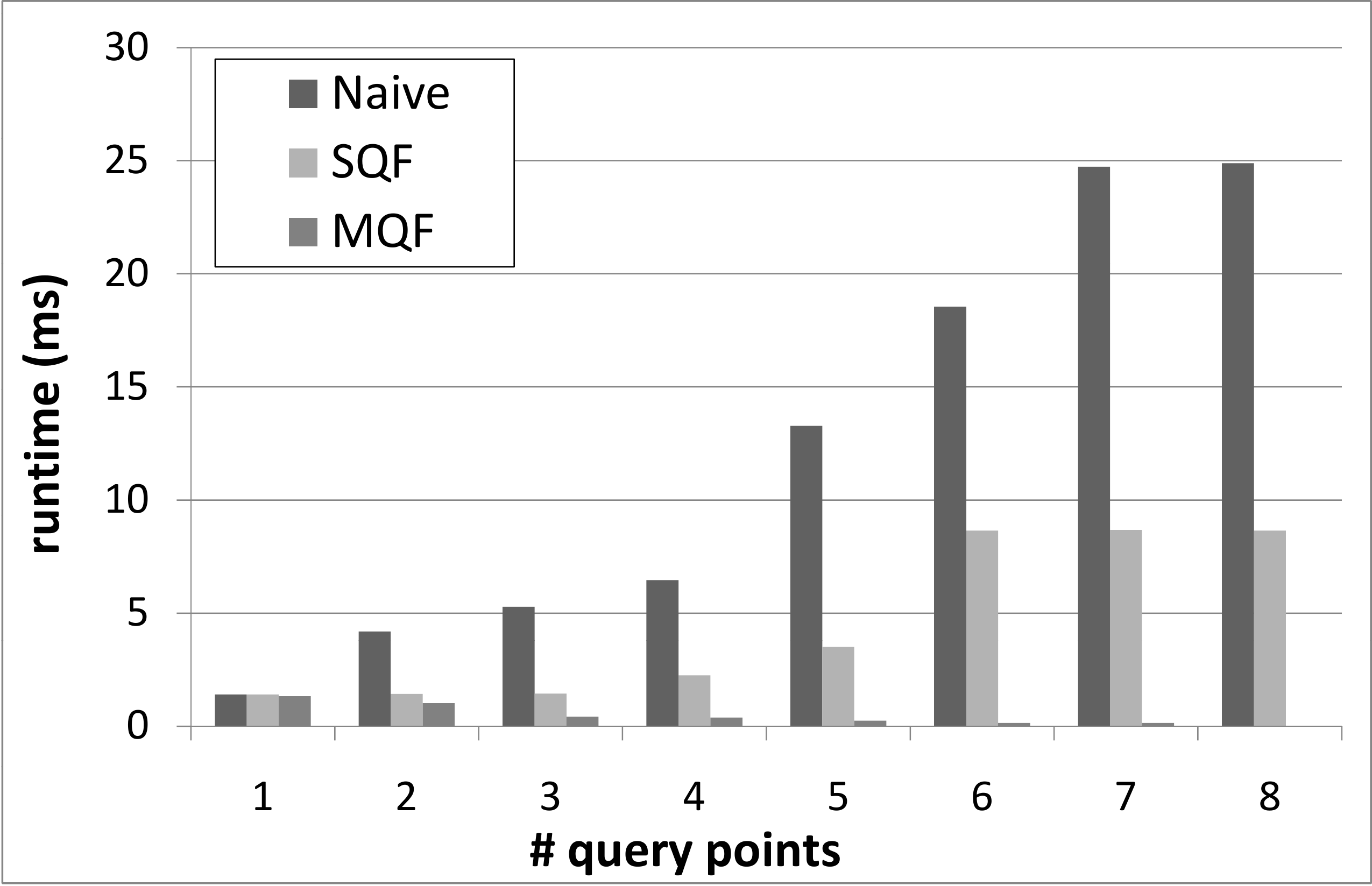}
    }
\vspace{-3mm}
    \caption{$I\mbox{-}DSQ$ algorithms on uniform dataset}
\end{figure}

% \begin{figure}[h]
%    \centering
%    \subfigure[I/O cost w.r.t. $|Q|$.]{
%        \label{fig:skyline-q}
%        \includegraphics[width = 0.45\columnwidth]{../charts/skyline-q}
%    }
%    \subfigure[I/O cost w.r.t. extent.]{
%        \label{fig:skyline-ext}
%        \includegraphics[width = 0.45\columnwidth]{../charts/skyline-extent}
%    }
%
%
%    \subfigure[I/O cost w.r.t. $d$.]{
%        \label{fig:skyline-dim}
%        \includegraphics[width = 0.45\columnwidth]{../charts/skyline-dimensions}
%    }
%    \subfigure[I/O cost w.r.t. $|\DB|$.]{
%        \label{fig:skyline-scale}
%        \includegraphics[width = 0.45\columnwidth]{../charts/skyline-scale}
%    }
%\vspace{-3mm}
%    \caption{Performance of $I\mbox{-}DSQ$ on synthetic dataset}
%    \label{fig:idsq}
%\end{figure}

\vspace{-3mm}
\section{Conclusions}
\label{sec:conclusion} In this paper we introduced and formalized
the problem for inverse query processing. We proposed a general
framework to such queries using a filter-refinement strategy and
applied this framework to the problem of answering inverse
$\eps$-range queries, inverse $k$NN queries and inverse dynamic
skyline queries. Our experiments show that our framework
significantly reduces the cost of inverse queries compared to
straightforward approaches. In the future, we plan to extend our
framework for inverse queries with different query predicates,
such as top-$k$ queries. In addition, we will investigate inverse
query processing in the bi-chromatic case, where queries and
objects are taken from different datasets. Another interesting
extension of inverse queries is to allow the user not only to
specify objects that have to be in the result, but also objects
that must not be in the result.

\newpage

\section*{Acknowledgements}
\begin{footnotesize}
This work was supported by a grant from the Germany/Hong
Kong Joint Research Scheme sponsored by the Research Grants Council of Hong Kong (Reference No. G\_HK030/09) and the Germany Academic Exchange Service of Germany (Proj. ID 50149322)
\end{footnotesize}

%\begin{footnotesize}
\vspace{-3mm}
\bibliographystyle{abbrv}
\bibliography{abbrev,literature}

\begin{thebibliography}{10}

\bibitem{BecKriSchSee90}
N.~Beckmann, H.-P. Kriegel, R.~Schneider, and B.~Seeger.
\newblock {The R*-Tree}: An efficient and robust access method for points and
  rectangles.
\newblock In {\em Proc.\ SIGMOD}, 1990.

\bibitem{DelSee07}
E.~Dellis and B.~Seeger.
\newblock Efficient computation of reverse skyline queries.
\newblock In {\em VLDB}, pages 291--302, 2007.

\bibitem{EmrKriKroRenZue10}
T.~Emrich, H.-P. Kriegel, P.~Kr{\"o}ger, M.~Renz, and A.~Z{\"u}fle.
\newblock Boosting spatial pruning: On optimal pruning of mbrs.
\newblock In {\em SIGMOD, June 6-11, 2010}.

\bibitem{HjaSam95}
G.~R. Hjaltason and H.~Samet.
\newblock Ranking in spatial databases.
\newblock In {\em Proc.\ SSD}, 1995.

\bibitem{KorMut00}
F.~Korn and S.~Muthukrishnan.
\newblock Influence sets based on reverse nearest neighbor queries.
\newblock In {\em Proc.\ SIGMOD}, 2000.

\bibitem{LiaChe08}
X.~Lian and L.~Chen.
\newblock Monochromatic and bichromatic reverse skyline search over uncertain
  databases.
\newblock In {\em SIGMOD Conference}, pages 213--226, 2008.

\bibitem{SinFerTos03}
A.~Singh, H.~Ferhatosmanoglu, and A.~S. Tosun.
\newblock High dimensional reverse nearest neighbor queries.
\newblock In {\em Proc.\ CIKM}, 2003.

\bibitem{StaAgrAbb00}
I.~Stanoi, D.~Agrawal, and A.~E. Abbadi.
\newblock Reverse nearest neighbor queries for dynamic databases.
\newblock In {\em Proc.\ DMKD}, 2000.

\bibitem{TaoPapLia04}
Y.~Tao, D.~Papadias, and X.~Lian.
\newblock Reverse {kNN} search in arbitrary dimensionality.
\newblock In {\em Proc.\ VLDB}, 2004.

\bibitem{VlaDouKotNor10}
A.~Vlachou, C.~Doulkeridis, Y.~Kotidis, and K.~N{\o}rv{\aa}g.
\newblock Reverse top-k queries.
\newblock In {\em ICDE}, pages 365--376, 2010.

\bibitem{YanLin01}
C.~Yang and K.-I. Lin.
\newblock An index structure for efficient reverse nearest neighbor queries.
\newblock In {\em Proc.\ ICDE}, 2001.

\end{thebibliography}
%\end{footnotesize}

\appendix

\section{Proofs of Lemmas}\label{appendix:proofs}
\subsection{Proof of Lemma \ref{lemma:IkNNQ_pruning_criterion_I}}
\begin{proof}
 By contradiction: Let $q\in Q$ such that:
$$
o\not\in Ik'\mbox{-}NNQ(\{q\})~\textrm{in}~\DB'\cup\{q\}.
$$
That is, $o$ does not have $q$ as one of its $k'$-nearest
neighbors in the database $\DB'\cup\{q\}$ containing all (and
only) non-query objects and $q$. This implies that there exist at
least $k'$ objects $o^\prime\in \DB'\cup\{q\}$ such that
$dist(o,o^\prime)<dist(o,q)$. Let $q^{ref}\in Q$ be the query
farthest from $o$. Thus, for each of the $Q-1$ objects
$q^{\prime}\in Q, q\neq q^{ref}$ it holds that
$dist(o,q)<dist(o,q^{ref})$ and for each of the $k'$ object
$o^{\prime}$ it holds that $dist(o,o^{\prime})<dist(o,q)\leq
dist(o,q^{ref})$. Thus there exist at least $Q-1+k'=k$ objects
that are closer to $o$ than $q^{ref}$. This implies that
$$
q^{ref}\not\in kNN(o)~\textrm{in}~\DB
$$
and thus
$$
o\not\in Ik\mbox{-}NNQ(Q)~\textrm{in}~\DB.
$$
Therefore, we have shown that
$$
\neg (o\in Ik'\mbox{-}NNQ(\{q\})~\textrm{in}~\DB'\cup\{q\})
\Rightarrow \neg (o\in Ik\mbox{-}NNQ(Q)~\textrm{in}~\DB)
$$
which is equivalent to
$$
o\in Ik\mbox{-}NNQ(Q)~\textrm{in}~\DB \Rightarrow \forall q\in Q:
o\in Ik'\mbox{-}NNQ(\{q\})~\textrm{in}~\DB'\cup\{q\}
$$
\end{proof}
As a side note: The counter direction $\mathbf{\Leftarrow:}$ also
holds, but is not required for pruning.
\begin{proof}
Assume that
$$
\forall q\in Q: o\in
Ik'\mbox{-}NNQ(\{q\})~\textrm{in}~\DB'\cup\{q\}
$$
$$
\mbox{where }k'=k-|Q|+1.
$$
Then, for each $q\in Q$ there exists at most $k'-1$ objects
$o^\prime \in \DB \setminus Q$ such that
$dist(o,o^\prime)<dist(o,q)$. In addition, there exist at most
$|Q|-1$ query objects $q^\prime\in Q\setminus \{q\}$ such that
$dist(o,q^\prime)<dist(o,q)$. Thus, there exist at most
$k'-1+|Q|-1=k-|Q|+1-1+|Q|-1=k-1$ objects which are closer to $o$
than to $q$, thus $q$ must be a $k$NN of $o$. Since this holds for
each $q\in Q$ we get
$$
o\in Ik\mbox{-}NNQ(Q)~\textrm{in}~\DB
$$
\end{proof}

\subsection{Proof of Lemma \ref{lemma:IkNNQ_query_point_filter}}
\begin{proof}
Due to Lemma \ref{lemma:IkNNQ_pruning_criterion_I} we have
$$
o\in Ik\mbox{-}NNQ(Q)~\textrm{in}~\DB \Rightarrow \forall q\in Q:
o\in Ik'\mbox{-}NNQ(\{q\})~\textrm{in}~\DB'\cup\{q\},
$$
Again, let $q^{ref}$ be the query point with the largest distance
to $o$. Since $q^{ref}\in Q$ we get:
$$
\forall q\in Q: o\in
Ik'\mbox{-}NNQ(\{q\})~\textrm{in}~\DB'\cup\{q\} \Rightarrow $$
$$o\in Ik'\mbox{-}NNQ(\{q^{ref}\})~\textrm{in}~\DB'\cup\{q^{ref}\}
$$
\end{proof}

%The following equivalence states that if and only if $o$ has
%$q^{ref}$ in its kNN set, then $o$ has all $q\in Q$ in its kNN
%set:
%$$
%\forall q\in Q: o\in Ik\mbox{-}NNQ(\{q\})\Leftrightarrow o\in
%Ik\mbox{-}NNQ(\{q^{ref}\})
%$$
%Now, since there exist exactly $|Q|-1$ query objects that are
%closer to $o$ than $q^{ref}$ it holds that:
%$$
%o\in Ik\mbox{-}NNQ(\{q^{ref}\})\Leftrightarrow o\in
%Ik'\mbox{-}NNQ(\{q^{ref}\})~\textrm{in}~\DB'\cup\{q\}
%$$
%
%using Corollary \ref{lemma:IkNNQ_TrueHit_criterion_I} we get
%$$
%o\in
%Ik'\mbox{-}NNQ(\{q^{ref}\})~\textrm{in}~\DB'\cup\{q\}\Rightarrow
%\forall q\in Q:o\in
%Ik'\mbox{-}NNQ(\{q\})~\textrm{in}~\DB'\cup\{q\}
%$$
%we get
%$$
%\forall q\in Q: o\in Ik'\mbox{-}NNQ(\{q\}) \Leftrightarrow o\in
%Ik'\mbox{-}NNQ(\{q^{ref}\})
%$$

\begin{figure}
    \centering
    \includegraphics[width = 0.5\columnwidth]{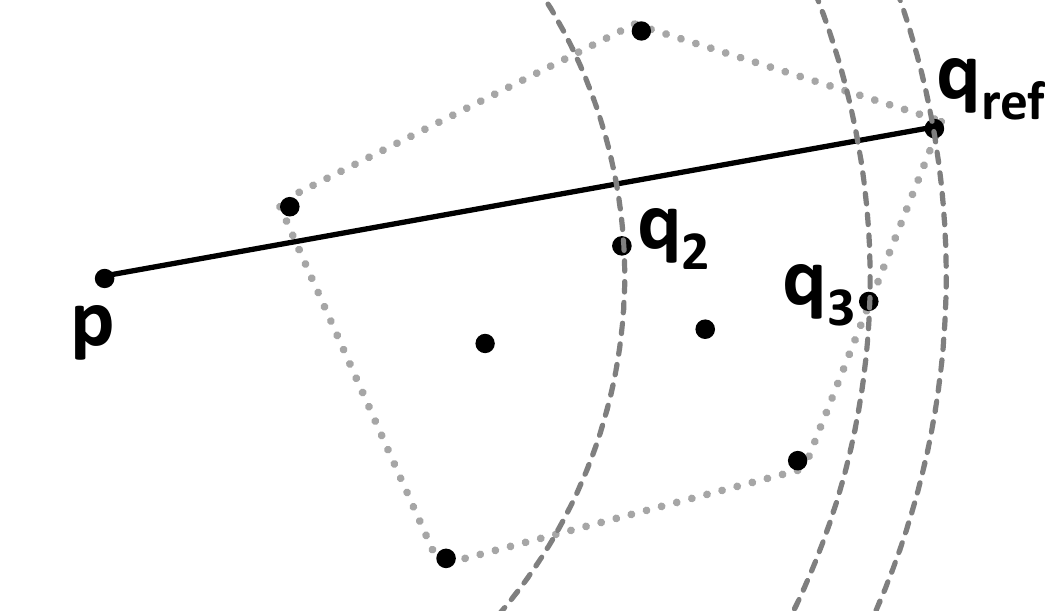}
    \caption{Illustration of Corollary \ref{lemma:convex_hull_property}}
    \label{fig:convexhull}
\end{figure}

\subsection{Proof of Lemma \ref{lemma:convex_hull_criterion_I}}
We first require the following corollary:
\begin{corollary}
\label{lemma:convex_hull_property} Let $Q\in\RR^d$ be a set of
points and $Q'\subseteq Q$ be the vertices of the
%smallest subset of $Q$ spanning
%the
convex hull of $Q$
%(vertices of the convex hull)
in $\RR^d$. Then, for each point $o\in\RR^d$, the farthest point
in $Q$ to $o$ must be in $Q'$ as well.
\end{corollary}
\begin{proof}
Consider any point $p\in\RR^d$ and its farthest point $q\in Q$.
Then all points in $Q$ must be located in the hyper-sphere
centered at $o$ with radius $d(o,q)$. Now, we can proof the above
lemma by contradiction assuming that $q$ is not a convex-hull
vertex. If $q$ is assumed to be within the convex-hull (not lying
on the margin of the convex hull), then the hyper-sphere splits
the convex-hull into points that are inside the sphere and points
that are out-side of the sphere as shown for $q_2$ in Figure
\ref{fig:convexhull}. Consequently, the convex hull contains
points that are farther from $o$ than $q$ which contradicts the
assumption. Now, we assume that $q$ lies on the margin (but not on
a vertex) of the convex hull which corresponds to a region of a
hyper-plane like $q_3$ in our example. If we move along this
hyper-plane starting from $q$, we are still within the convex-hull
but leave the hyper-sphere of $o$. Consequently, again, the convex
hull contains points that are farther from $o$ than $q$ which
again contradicts the assumption.
\end{proof}
Now we can use Corollary \ref{lemma:convex_hull_property} to prove
Lemma \ref{lemma:convex_hull_criterion_I}:
%\begin{lemma}
%Let $Q$ be the set of query objects and
%$\mathcal{H}\subseteq\DB-Q$ be the database (non-query) objects in
%the space $C\subseteq \RR^d$ covered by the convex hull $Q'$ of
%$Q$. Furthermore, let $o\in\DB$ be a database object and
%$q_{ref}\in Q$ a query object such that $\forall q\in Q:
%d(o,q_{ref})\geq d(o,q)$. Then for each object $p\in\mathcal{H}$
%it holds that $d(o,p)\leq d(o,q_{ref})$.
%\end{lemma}
\begin{proof}
By definition of $q_{ref}$ it holds that
$$
q_{ref}=argmax_{q\in Q}(dist(o,q))
$$
Since the vertices of the convex hull of $Q$ consists only of
points in $Q$, Corollary \ref{lemma:convex_hull_property} leads to
$$
q_{ref}=argmax_{c\in C}(dist(o,c))
$$
%Since $H\subseteq C$ it holds that
Thus,
$$
\forall c\in C:dist(o,q_{ref})\geq dist(o,c)
$$
and since $\mathcal{H}\subseteq C$:
$$
\forall p\in \mathcal{H}:dist(o,q_{ref})\geq dist(o,p)
$$
\end{proof}

\subsection{Proof of Lemma \ref{lemma:IkNNQ_pruning_criterion_II}}
\begin{proof}
$\mathbf{\Rightarrow:}$ If $o\in Ik\mbox{-}NNQ(Q)$, then all query
points (including $q_{ref}$) are in the $k$NN set of $o$. Since
for all points $p$ in $\mathcal{H}$, $d(o,p)\le d(o,q_{ref})$ (see
Lemma \ref{lemma:convex_hull_criterion_I}), all points in
$\mathcal{H}$ should also be in the $k$NN set of $o$. Therefore,
in the (worst) case, where $q_{ref}$ is the $k$-th NN of $o$,
there can be $k-|\mathcal{H}|-|Q|$ points outside the convex hull
closer to $o$ than $q_{ref}$, if $o\notin \mathcal{H}$, or
$k-|\mathcal{H}|+1-|Q|$ points if $o\in \mathcal{H}$.

$\mathbf{\Leftarrow:}$ If $o$ is outside the hull, from the points
in $\mathcal{H}\cup Q$, $q_{ref}$ is the furthest one to $o$ (see
Lemma \ref{lemma:convex_hull_criterion_I}). If there are at most
$k-|\mathcal{H}|-|Q|$ points outside the hull closer to $o$ than
$q_{ref}$ is, then the distance ranking of $q_{ref}$ is at most
$k$. Since all other points in $Q$ are closer to $o$ than
$q_{ref}$ is, it should be $o\in Ik\mbox{-}NNQ(Q)$. If $o\in
\mathcal{H}$, the bound should be $k-|\mathcal{H}|-|Q|+1$, as $o$
should be excluded from $\mathcal{H}$ in the proof.
\end{proof}

\subsection{Proof of Lemma \ref{lem:fastSky}}
\begin{proof}
Let us consider the space partitioning $\mathcal{S}$ derived from
dividing the object space at $q$. Each $q'$ located within
partition $r\in\mathcal{S}$ generates a pruning region $PR_q(q')$
(cf. Definition \ref{def:pruning_region}) that totally covers the
partition $r'\in\mathcal{S}$ which is opposite to $r$ w.r.t. $q$.
Since we assume that we have at least one query object $q'\neq q$
in each partition $r\in\mathcal{S}$, all partitions
$r'\in\mathcal{S}$ are totally covered by a pruning region and,
consequently, the complete data space can be pruned. An example in
the two-dimensional space is illustrated in Figure
\ref{fig:skyline_proof}.
\end{proof}

\begin{figure}
    \centering
    \includegraphics[width=0.48\columnwidth]{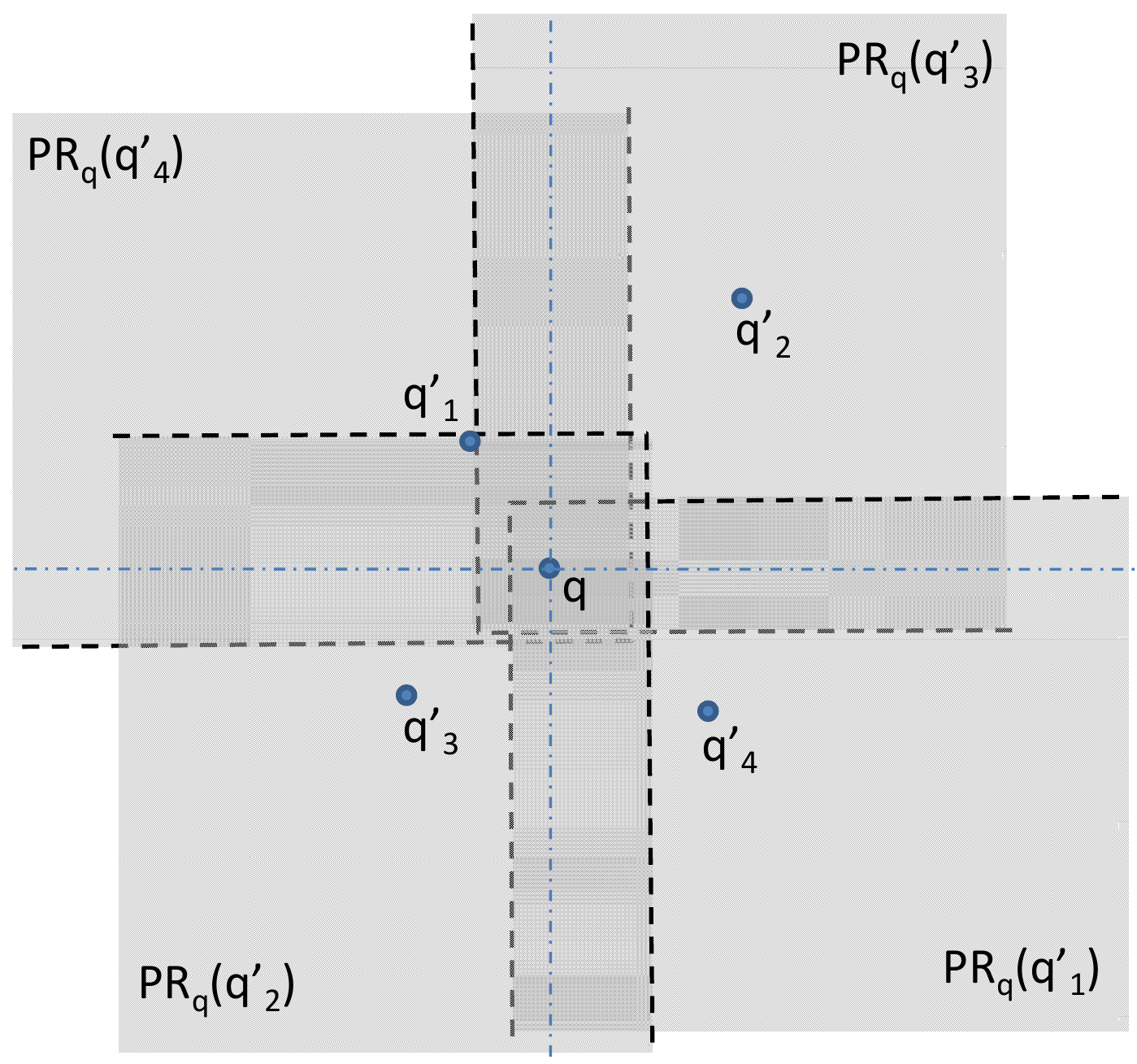}
\vspace{-3mm} \caption{Fast query based validation filter in
2D-space}
    \label{fig:skyline_proof}
\end{figure}

\section{Algorithms}
\label{app:alg}

In this section we illustrate the pseudo code of the I$k$NNQ (cf.
Algorithm \ref{alg:iknn}) and the IDSQ (cf. Algorithm
\ref{alg:isnn}) algorithm. A more detailed explanation is given in
Section \ref{sec:iknn-algo} and \ref{sec:idsq-algo} respectively.

\begin{algorithm}
\footnotesize
  \caption{Inverse kNNQuery}
  \label{alg:iknn}
    \begin{algorithmic}[1]
        \REQUIRE $Q$, $k$, $ARTree$
        \STATE \textit{//Fast Query Based Validation}
        \IF {$|Q|>k$}
            \STATE return "no result" and terminate algorithm
        \ENDIF
        \STATE $pq$ PriorityQueue ordered by $max_{q_i \in Q}$MinDist
        \STATE $pq.add(ARTree.root~\textrm{entries})$
        \STATE $|\mathcal{H}|= 0$
         \STATE LIST $candidates, prunedEntries$
        \STATE \textit{//Query/Object Based Pruning}
        \WHILE{$\neg pq.isEmpty()$}
            \STATE $e = pq.poll()$
            \IF {$getPruneCount(e, Q, candidates, prunedEntries, pq) > k -
            |\mathcal{H}|-|Q|$}
                \STATE $prunedEntries.add(e)$
            \ELSIF{$e.isLeafEntry()$}
                \STATE $candidates.add(e)$
            \ELSE
                \STATE $pq.add(e.getChildren())$
            \ENDIF
            \IF{$e \in convexHull(Q)$}
                \STATE $|\mathcal{H}| += e.agg\_count$
            \ENDIF
        \ENDWHILE
        \STATE \textit{//Refinement Step}
        \STATE LIST $result$
        \FOR{ $c \in candidates$}
            \IF {$q_{ref} \in knnQuery(c, k)$}
                \STATE $result.add(c)$
            \ENDIF
        \ENDFOR
        \STATE return ($result$)
    \end{algorithmic}
\end{algorithm}

\begin{algorithm}
\footnotesize
  \caption{Inverse Dynamic Skyline Query}
  \label{alg:isnn}
    \begin{algorithmic}[1]
        \REQUIRE $Q$, $ARTree$
        \STATE $pq$ PriorityQueue ordered by $min_{q_i \in Q}MaxDist$
        \STATE $pq.add(ARTree.root~\textrm{entries})$
         \STATE LIST $candidates, prunedEntries$
        \STATE \textit{//Filter step}
        \WHILE{$\neg pq.isEmpty()$}
            \STATE $e = pq.poll()$
            \IF {$canBePruned(e, Q, candidates, prunedEntries, pq)$}
            \label{alg:cbp}
                \STATE $prunedEntries.add(e)$
            \ELSIF{$e.isLeafEntry()$}
                \STATE $candidates.add(e)$
            \ELSE
                \STATE $pq.add(e.getChildren())$
            \ENDIF
        \ENDWHILE
        \STATE \textit{//Refinement Step}
        \STATE LIST $result$
        \FOR{ $c \in candidates$}
            \IF {$Q \in dynamicSkyline(c)$}
                \STATE $result.add(c)$
            \ENDIF
        \ENDFOR
        \STATE return ($result$)
    \end{algorithmic}
\end{algorithm}

\section{Pruning  Candidates in Inverse Dynamic Skyline Queries}
\label{appn:sky_2d}

Here, we show how pruning a data object from the candidates set of
an inverse skyline query can be accelerated. The pruning
conditions discussed here hold for the special 2D case.
Nonetheless some of them can be extended to spaces of higher
dimensionality, with lower pruning effectiveness.

%
%We now show some techniques that can be used to determine the pruning
%space of multiple objects more efficiently than enumerating all pairs
%of points in $Q$. These t
%First, we focus on special properties that are only valid for the
%two-dimensional case but allow us to apply more efficient pruning
%methods than for the general multi-dimensional case.

\begin{figure}
    \centering

    \subfigure[corner point\label{subfig:skyline_2a}]{\includegraphics[width=0.3\columnwidth]{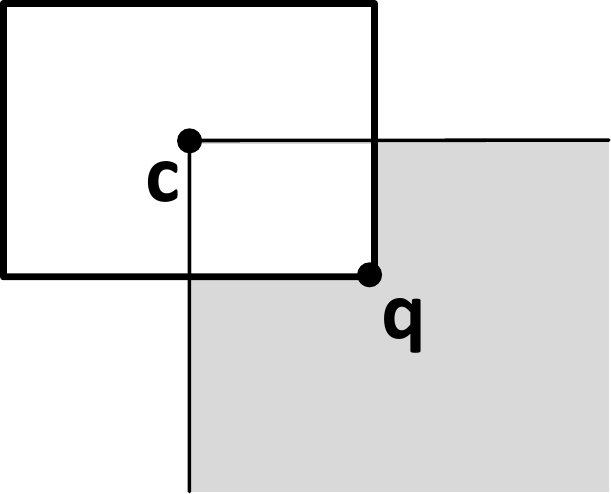}}
    ~~~~~~
    \subfigure[edge point\label{subfig:skyline_2b}]{\includegraphics[width=0.3\columnwidth]{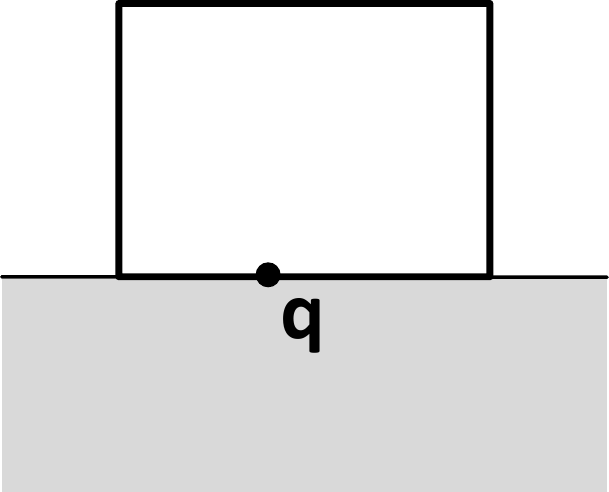}}
    \caption{Using query points at the border of \qmbr to prune space}
    \label{fig:skyline_2}
\end{figure}

%\begin{figure}
%    \centering
%    \subfigure[.\label{subfig:skyline_3a}]{\includegraphics[width=0.9\columnwidth]{../figures/skyline_3a.pdf}}
%    \subfigure[.\label{subfig:skyline_3b}]{\includegraphics[width=0.9\columnwidth]{../figures/skyline_3b.pdf}}
%    \caption{.}
%    \label{fig:skyline_3}
%\end{figure}

\subsection{Query-Based Pruning}
%For a given candidate database object $o$, we can perform an
%initial filter step by testing for each pair $q,q^\prime \in
%Q,q\neq q^\prime$ if $o$ is located in the pruning region defined
%by $q$ and $q^\prime$ using Definition
%\ref{def:pruning_region_query}. Exemplarily, we show the
%selectivity of the pruning methods when combining the pruning
%spaces according to multiple query objects using the above
%introduced pruning methods. Figure \ref{subfig:skyline_3a} shows
%an example for an entire pruning region built by three query
%objects, while the example depicted in Figure
%\ref{subfig:skyline_3b} shows the pruning region according to four
%query objects. Thus any object $o$ (or any object approximation
%such as an index page) that is (fully) located in the pruning
%region, can be safely removed from the list of candidates. We can
%observe that the remaining space that may contain candidates, i.e.
%the remaining search space, is very limited.

Let \qmbr be the rectangle that minimally bounds all query
objects.
%To avoid having to test each of the $O(|Q|^2)$ pairs of query
%objects, we now present pruning conditions based on the rectangle
%that minimally bounds all query objects (\qmbr) %and partitions the
%%object space into nine partitions (A,B,C,D,E,F,G,H,I,J),
%as depicted in Figure \ref{subfig:skyline_2a}.
The main idea is to identify pruning regions by just considering
\qmbr and one query object $q$ located on the boundary of \qmbr.
We first concentrate on pruning objects outside \qmbr, the cases
for pruning objects inside of \qmbr will be discussed later.

\textbf{Pruning Condition I:} Assume that $q$ is located at one
corner of \qmbr, then the axis-aligned region outside of \qmbr
where the $i$th dimension is defined by the interval
$[c^i,+\infty]$ if $q^i>c^i$ and $[-\infty,c^i]$ otherwise (where
$c$ is the center of \qmbr) can be pruned. The rationale is that
since \qmbr is an MBR, at least one other query object must be at
each edge of \qmbr located on the opposite side of $q$ which can
be used to create a pruning region at the side of $q$. In the
example shown in Figure \ref{subfig:skyline_2a}, $q$ is located at
the lower right corner of \qmbr and, therefore, additional query
objects must be located on both the left and upper edge of \qmbr.
As a consequence, any object in the lower-right shaded region can
be pruned.

\textbf{Pruning Condition II:} In the case, where $q$ is located
at a boundary of \qmbr, but not at a corner of \qmbr, the
half-space constructed by splitting the data space along the edge
containing $q$ and does not contain \qmbr defines the pruning
region. Here, the rationale is that since $q$ is on an edge $e$
(but not at a corner) of \qmbr, there must be at least two
additional query objects, located at the edges adjacent to $e$,
respectively. By pairing $q$ with each of these two objects, and
merging the corresponding pruning regions, we obtain the pruning
region. For example, in Figure \ref{subfig:skyline_2b}, for $q$,
we get the shaded pruning region below $q$. As another example,
consider the union of $PQ_{q_1}(q_4)$ and $PQ_{q_3}(q_4)$ in
Figure \ref{subfig:skyline_3b} which prunes the whole hyperplane
below $q_4$. Thus, if all four edges of \qmbr contain four
different query objects, then only objects in \qmbr are candidate
$I\mbox{-}DSQ$ results.

\subsection{Object-Based Pruning}
For any candidate object $o$ that is not pruned during the
query-based filter step, we need to check if there exists any
other database object $o^\prime$ which dominates some $q\in Q$
with respect to $o$. If we can find such an $o^\prime$, then $o$
cannot have $q$ in its dynamic skyline and thus $o$ can be pruned
for the candidate list. Naively, we can determine, for each
database object $o^\prime$ that we have found so far and each
query object $q$, the pruning region $PR_q(o^\prime)$ according to
Definition \ref{def:pruning_region} and check, if $o$ is located
in this region. In the following, we show how to perform this
pruning without considering all possible combinations of database
and query objects.

\begin{figure}
    \centering
    \subfigure[Case 1\label{subfig:skyline_4a}]{\includegraphics[width=0.3\columnwidth]{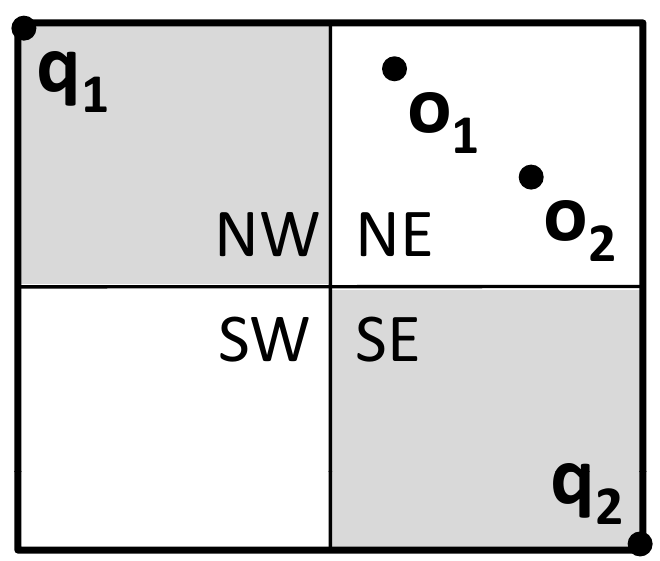}}
    ~~~~~~~
    \subfigure[Case 2\label{subfig:skyline_4b}]{\includegraphics[width=0.3\columnwidth]{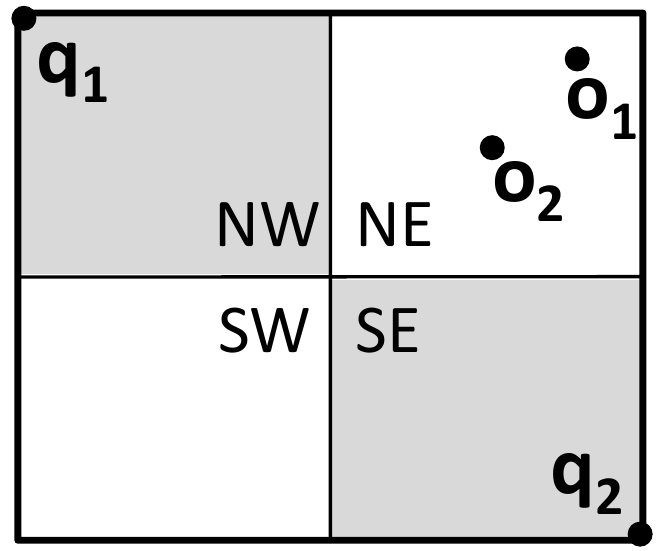}}
    \caption{Object based pruning inside  \qmbr}
    \label{fig:skyline_4}
\end{figure}

%Let $R \subset \qmbr$ be the region where candidate object $o$ lies.

%Without loss of generality, assume that candidate object $o$ is
%located in a region $R$ that cannot be pruned by the query-based
%pruning approach\footnote{This is no loss of generality since due
%to the assumption that $o$ is a candidate, it must be located in
%one of these regions.} and assume that $R$ is located inside the
%query rectangle \qmbr.
%The following pruning conditions are based
%on two query objects $q_1,q_2\in Q$. We

Consider two query points $q_1$ and $q_2$ and the rectangle
$\qmbr_{q_1q_2}$ that minimally bounds $q_1$ and $q_2$. Note, that
$\qmbr_{q_1q_2}$ is fully contained in \qmbr and the union of all
$\qmbr_{q_iq_j}$  for each pair $q_i,q_j\in Q, i\neq j$ is equal
to \qmbr. Consider the axis-aligned space partitioning according
to the center point $c$ of $\qmbr_{q_1q_2}$ resulting into four
partitions, denoted as NE, SE, SW and NW as illustrated in Figure
\ref{fig:skyline_4}. According to the query-based pruning, the two
partitions containing $q_1$ and $q_2$ respectively, can be pruned.
Now, we can show the following:
\begin{corollary}
\label{cor:atmostonecand} In each of the two remaining regions (NE
and SW in the example) within $\qmbr_{q_1q_2}$, there can only be
at most one candidate.
\end{corollary}
To prove this, let us consider the following two cases illustrated
in Figures \ref{subfig:skyline_4a} and \ref{subfig:skyline_4b}:

\textbf{Case 1:} Let us assume that there are two objects $o_1$
and $o_2$ in one of the regions, that have the same topology as
the two query objects $q_1$ and $q_2$, as shown in Figure
\ref{subfig:skyline_4a}. In this case both objects prune each
other. The reason is that $o_1$ is in the pruning region of $o_2$
w.r.t. $q_2$ and $o_2$ is in the pruning region of $o_1$ w.r.t.
$q_1$. Consequently, there cannot exist two candidates in this
partition where the objects have similar spatial relationship as
$q_1$ and $q_2$.

\textbf{Case 2:} Now, we assume that there exist two objects $o_1$
and $o_2$ in one partition within $\qmbr_{q_1q_2}$ such that
$o_1o_2$ is orienter perpendicularly to $q_1q_2$, as illustrated
in the example in Figure \ref{subfig:skyline_3b}. In this case,
$o_1$ is pruned by $o_2$ for both query objects. The reason is
that in each dimension, the distance between $o_1$ and $o_2$ must
be less than the distance between $o_1$ and $q_1$ ($q_2$).
However, $o_2$ can in general not be pruned by $o_1$, thus $o_2$
remains a candidate.

We can now use Corollary \ref{cor:atmostonecand} to obtain the
following pruning condition:

\textbf{Pruning Condition III:} Let $R \subseteq \qmbr$ be a
region inside the query rectangle that cannot be pruned using
query-based pruning. Let $q_1,q_2\in Q$ be two query points for
which it holds that $R$ is fully contained in the rectangle
minimally bounding $q_1$ and $q_2$. Since $R$ cannot be pruned
based on query-pruning only, $R$ must be located in non-pruning
regions (e.g. NE and SW in Figure \ref{subfig:skyline_3a}) of
$q_1$ and $q_2$. Without loss of generality, let us assume that
$R$ is located in the NE region of $\qmbr_{q_1q_2}$. Now Let $O$
be the set of database objects inside $R$. Let $a\in O$ be the
object with the largest $x$ coordinate and let $b\in O$ be the
object with the largest $y$ coordinate. If $a\neq b$ we can prune
$R$. If $a=b$, then $a$ is a candidate and all other objects $c\in
R, c\neq a$ can be pruned.

The above pruning condition allows us to prune objects inside
$\qmbr$. The next pruning condition  allows us to prune objects
outside \qmbr using database objects. In the following, let
$\qmbr_{i}.max$ and $\qmbr_i.min$ denote the maximal and minimal
coordinate of \qmbr, at dimension $i$, respectively.

\begin{figure}[t]
  \centering
  \includegraphics[width=0.33\textwidth]{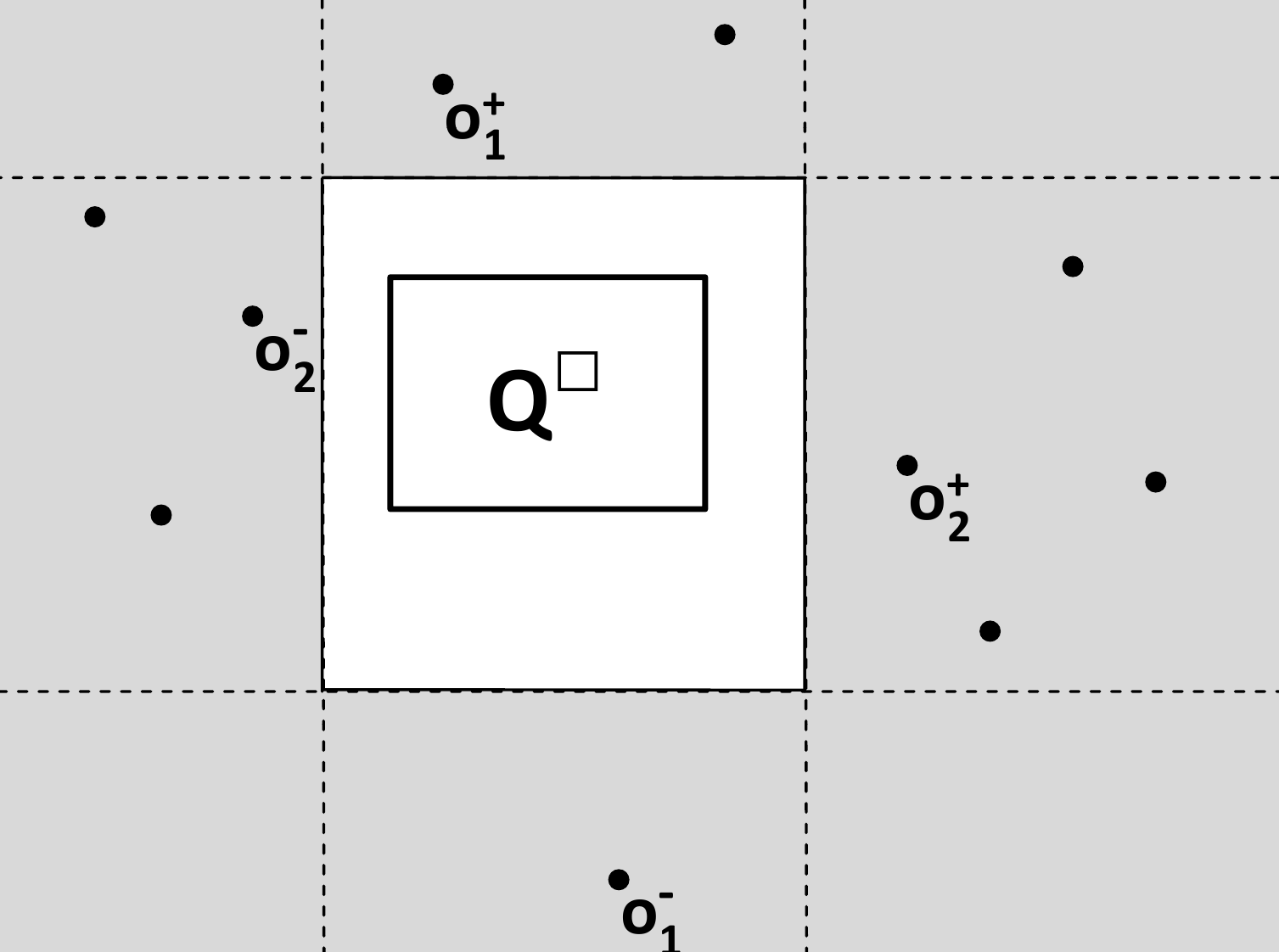}
  \caption{Pruning regions outside of \qmbr.}
  \label{fig:outsidepruning}
\end{figure}

\textbf{Pruning Condition IV}: For the next pruning condition we
use the sets of database objects $O_i\subseteq \DB$ outside \qmbr
for which it holds that each $o\in O_i$ intersects \qmbr in one
dimension but not in the other. Such objects are shown in Figure
\ref{fig:outsidepruning}. Let $O^+_i$ ($O^-_i$) denote the subset
of $O_i$ so that each object $o\in O^+_i$ has a larger (smaller)
coordinate than $\qmbr$ in the other dimension. Now let $o^+_i\in
O^+_i$ ($o^-_i\in O^-_i$) be the object in $O^+_i$ ($O^-_i$) with
the smallest (largest) coordinate in the other dimension $j$. Any
database object which has a $j$ coordinate greater than
$\frac{o^+_i+\qmbr_i.max}{2}$ or less than
$\frac{o^-_i+\qmbr_i.min}{2}$ can be pruned. The rationale of this
pruning condition is that due the MBR property of \qmbr, we know
that at least one query object must be located on each edge of
\qmbr. The pruning regions defined are based on these query
objects. For instance, for object $o_1^+$ in Figure
\ref{fig:outsidepruning} we can exploit that there must be query
objects $q_1$ and $q_2$ located on the left and the right border
edge of \qmbr, respectively. This allows us to create the two
pruning regions $PR_{q_1}(o_1^+)$ and $PR_{q_2}(o_1^+)$ according
to Definition \ref{def:pruning_region}. These pruning regions are
smallest, if $q_1$ and $q_2$ are located at the upper corners of
\qmbr. Thus, we can prune any object above the line biscecting the
upper side of $\qmbr$ and $o_1^+$.
%This pruning regions allows us to
%prune further objects that could not be pruned using query-based
%pruning only.

%\section{Further Applications}
%In this appendix we show further applications for inverse queries:
%\subsection{Item Packages}
%Consider an online shop selling a variety of different products
%stored in a database \DB. The online shop may be able to offer a
%set products $Q\subseteq \DB$ together for a special price. The
%problem at hand is to identify customers which are interested in
%all items of the package, in order to direct an advertizement to
%them. We assume that the preferences of registered customers are
%known. First, we need to define a predicate indicating whether a
%user is interested in a product. A customer may be interested in a
%product if
%\begin{itemize}
%\item the distance between the products features and the customers
%preference is less than a threshold $\eps$.
%
%\item the product is contained in the set of his $k$ favorite
%items, i.e. the $k$-set of product features closest to the user's
%preferences.
%
%\item the product is contained in the customers dynamic skyline,
%i.e. there is no other product that better fits the customers
%preferences in every possible way.
%\end{itemize}
%Thus we require to find all customer preferences to which all
%products $q\in Q$ are similar to. Thus, we want to identify
%customers $r$ for which it holds that a query on \DB with query
%object $r$, using one of the query predicates above, $Q$ is
%contained the result set.

\section{Pruning Techniques For The Bi-Chromatic Case}
\label{app:bi} Here, we explain how the proposed pruning
techniques which are designed for the mono-chromatic case can
easily be adapted to the bi-chromatic case. Here we assume two
data sets $\DB$ and $\DB'$. A bi-chromatic inverse query returns
the set of objects $r\in \DB'$ for which each inverse query object
$q\in Q \subseteq \DB'-$ is contained in the result of a
$\mathcal{P}$ query applied on data set $\DB$ using $r$ as query
object. For each case of the predicate $\mathcal{P}$, we briefly
explain the changes to our technique.

\subsubsection*{Bi-Chromatic I$\eps$-Range Queries}
Here, the filter rectangle can  directly be applied to $\DB'$, so
there is no practical change.

\subsubsection*{Bi-Chromatic IkNN Queries}
In this case we have to avoid allowing objects in $\DB'$ to prune
each other. In contrast to the mono-chromatic case, we only have
to consider objects in $\DB$ to build $H$, i.e. the number of
objects in the convex hull regions of $Q$.

\subsubsection*{Bi-Chromatic Inverse Dynamic Skyline Queries}
In this case, the pruning region is only defined by objects in
$\DB$. This pruning region is used to prune objects in $\DB'$.

\section{Additional Experiments}\label{appn:exp}
In this section we show the behavior of inverse queries on other
datasets than the ones used in section \ref{sec:experiments}. We
excluded from the evaluation the Naive approach due to its poor
performance; this way,
 the difference between MQF and SQF
becomes more clear. Due to space limitations we focused on I$k$NN
queries on the clustered dataset (cf. Figure
\ref{fig:knn-cluster}) and the real dataset (cf. Figure
\ref{fig:knn-real}). Let us note that similar trends could be
observed for the other inverese query types.
% We focused on inverse $k$-NN queries and inverse dynamic skyline
% queries. Here we performed them on the clustered dataset (cf.
% Figures \ref{fig:knn-cluster} and \ref{fig:skyline-cluster}) and
% the real dataset (cf. Figures \ref{fig:knn-real} and
% \ref{fig:skyline-real}).

\begin{figure}[h]
    \centering
    \subfigure[I/O cost w.r.t. $k$.]{
        \label{fig:knn-cluster-k}
        \includegraphics[width = 0.45\columnwidth]{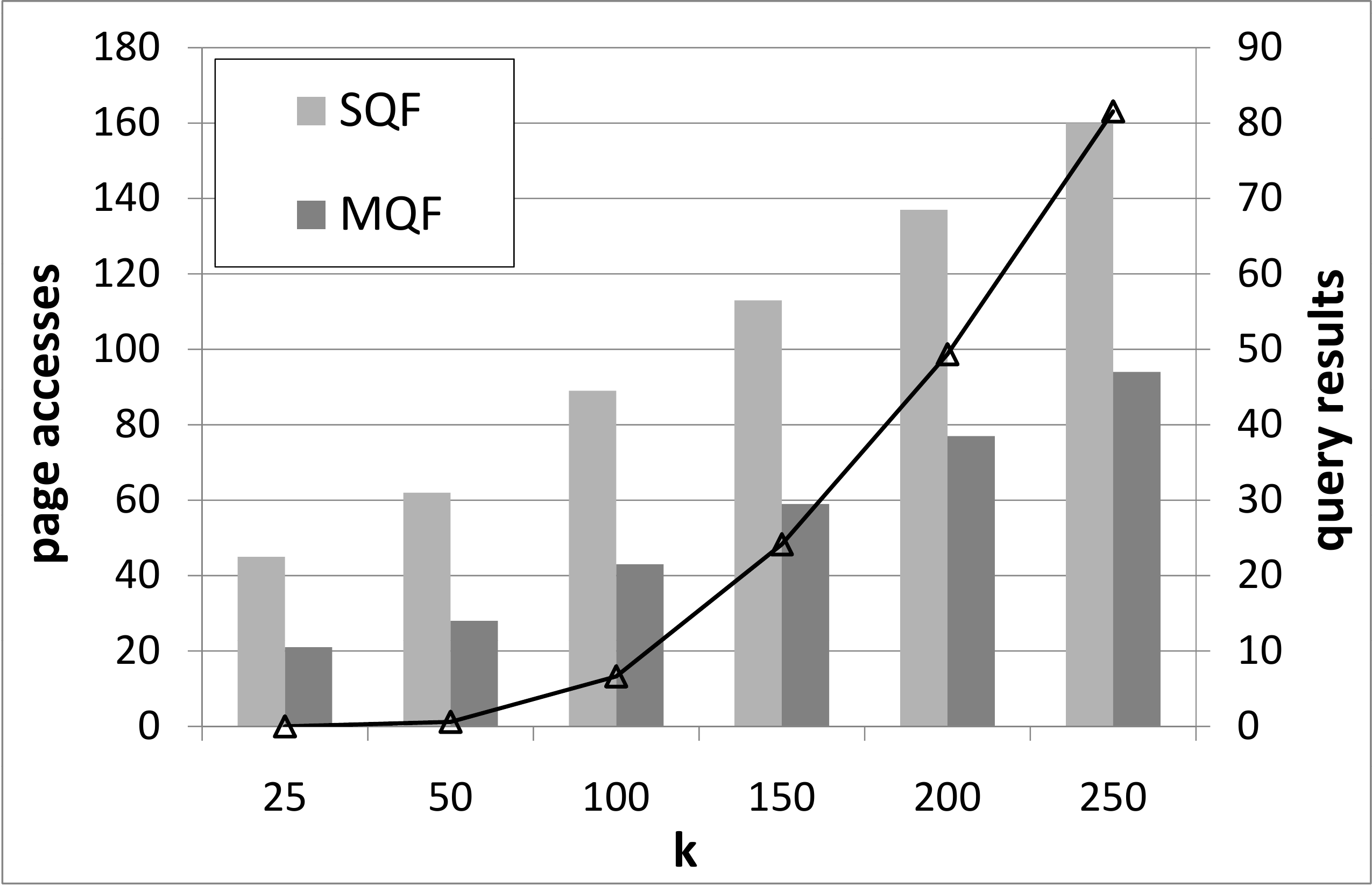}
    }
    \subfigure[I/O w.r.t. $|Q|$.]{
        \label{fig:knn-cluster-q}
        \includegraphics[width = 0.45\columnwidth]{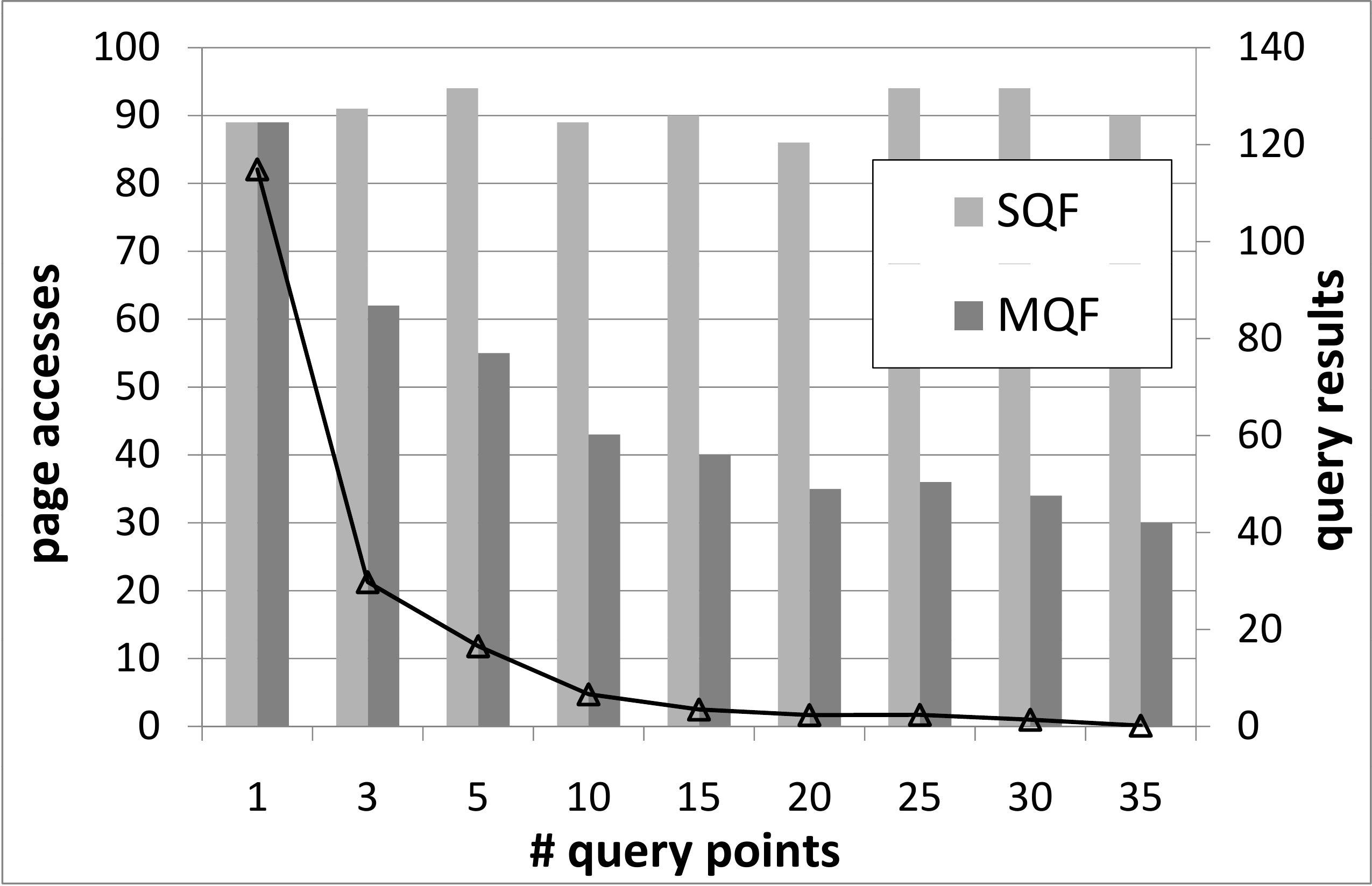}
    }

    \subfigure[I/O w.r.t. $d$.]{
        \label{fig:knn-cluster-dim}
        \includegraphics[width = 0.45\columnwidth]{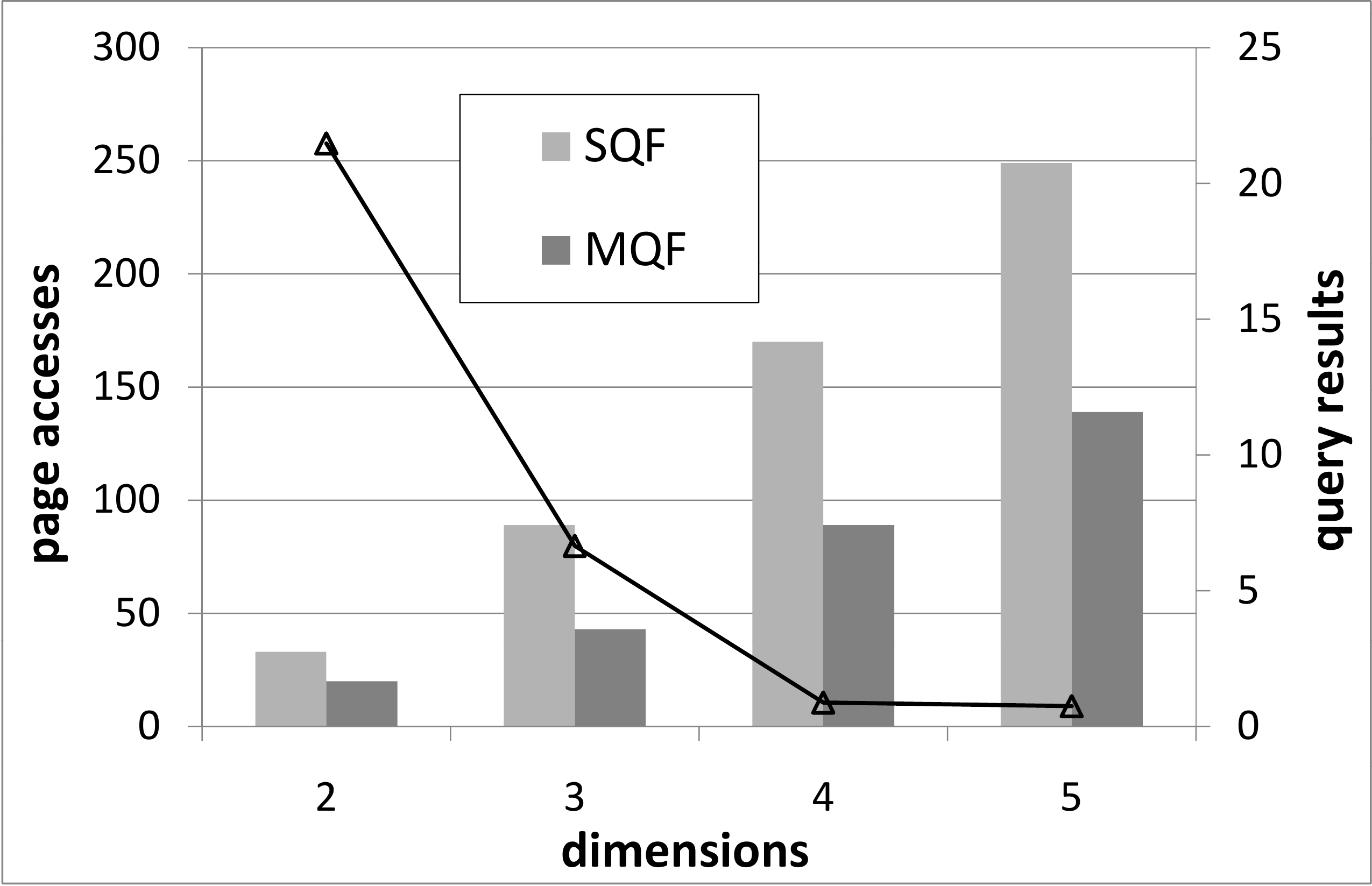}
    }
    \subfigure[I/O w.r.t. extent.]{
        \label{fig:knn-cluster-ext}
        \includegraphics[width =
        0.45\columnwidth]{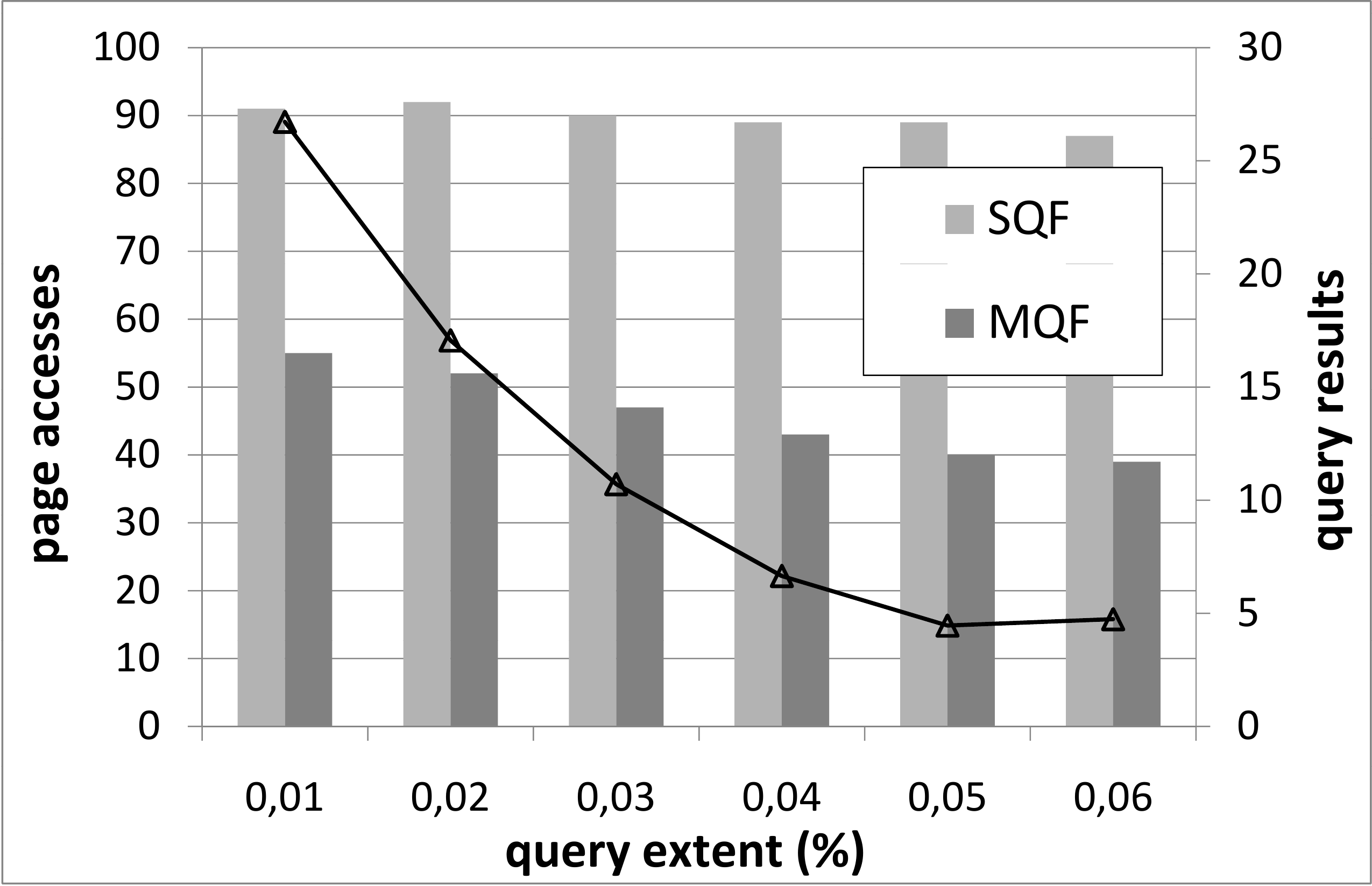} }

    \subfigure[I/O w.r.t. $|\DB|$.]{
        \label{fig:knn-cluster-scale}
        \includegraphics[width = 0.45\columnwidth]{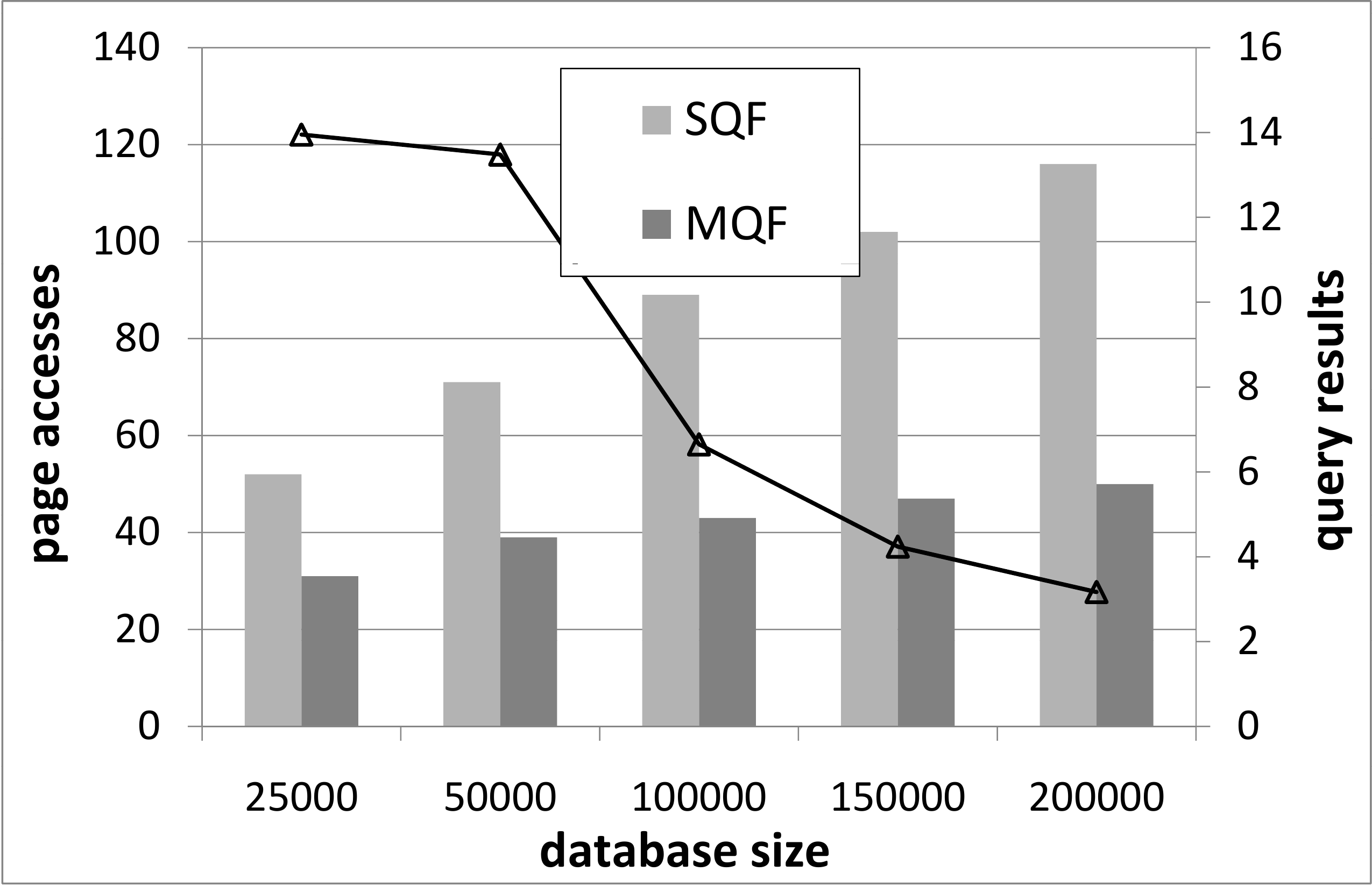}
    }
    \subfigure[CPU w.r.t. $|Q|$.]{
        \label{fig:knn-cluster-cpu}
        \includegraphics[width = 0.45\columnwidth]{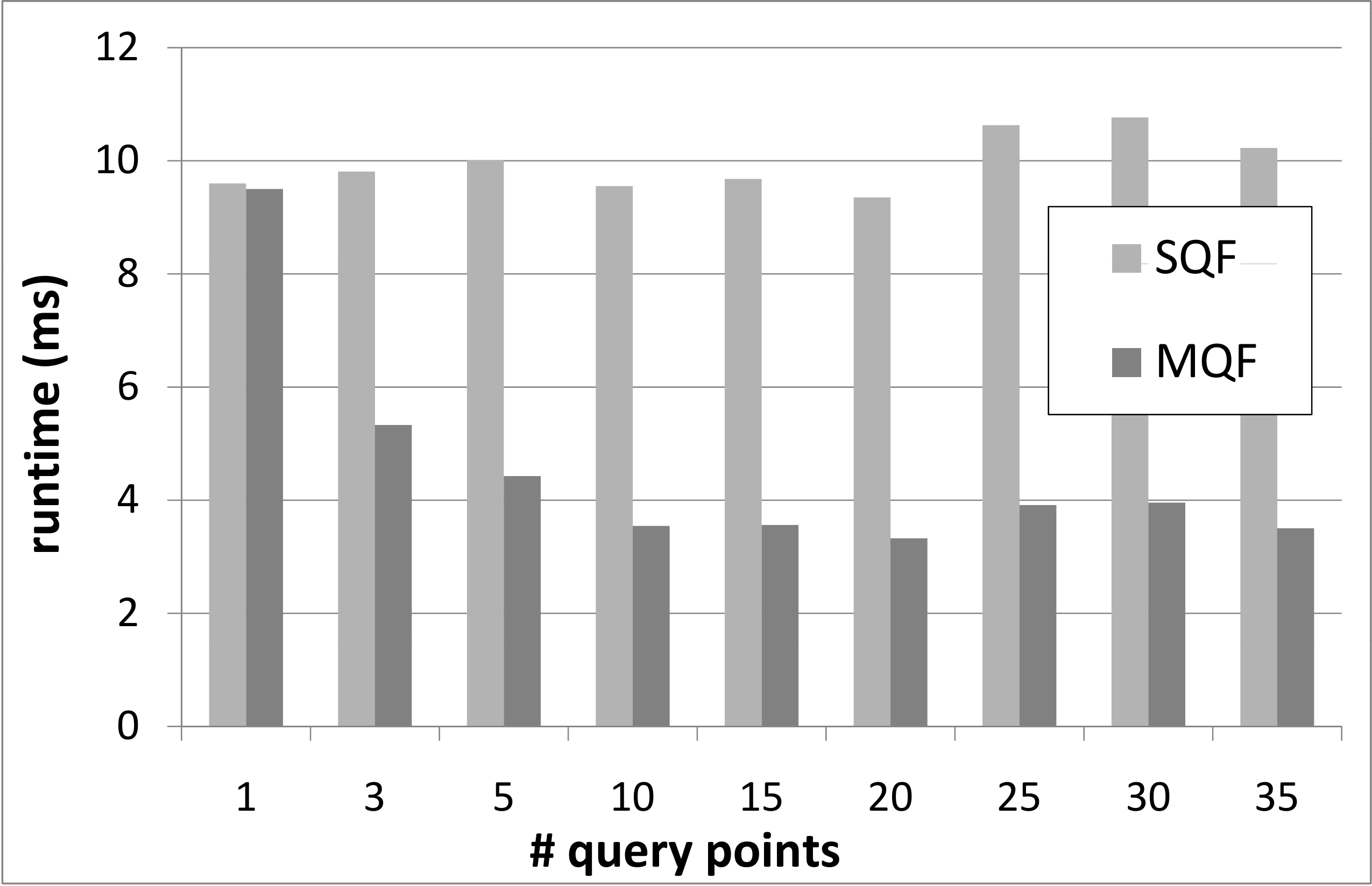}
    }
\vspace{-3mm}
    \caption{$Ik\mbox{-}NNQ$ algorithms on clustered dataset}
    \label{fig:knn-cluster}
\end{figure}

\begin{figure}[h]
    \centering
    \subfigure[I/O cost w.r.t. $k$.]{
        \label{fig:knn-real-k}
        \includegraphics[width = 0.45\columnwidth]{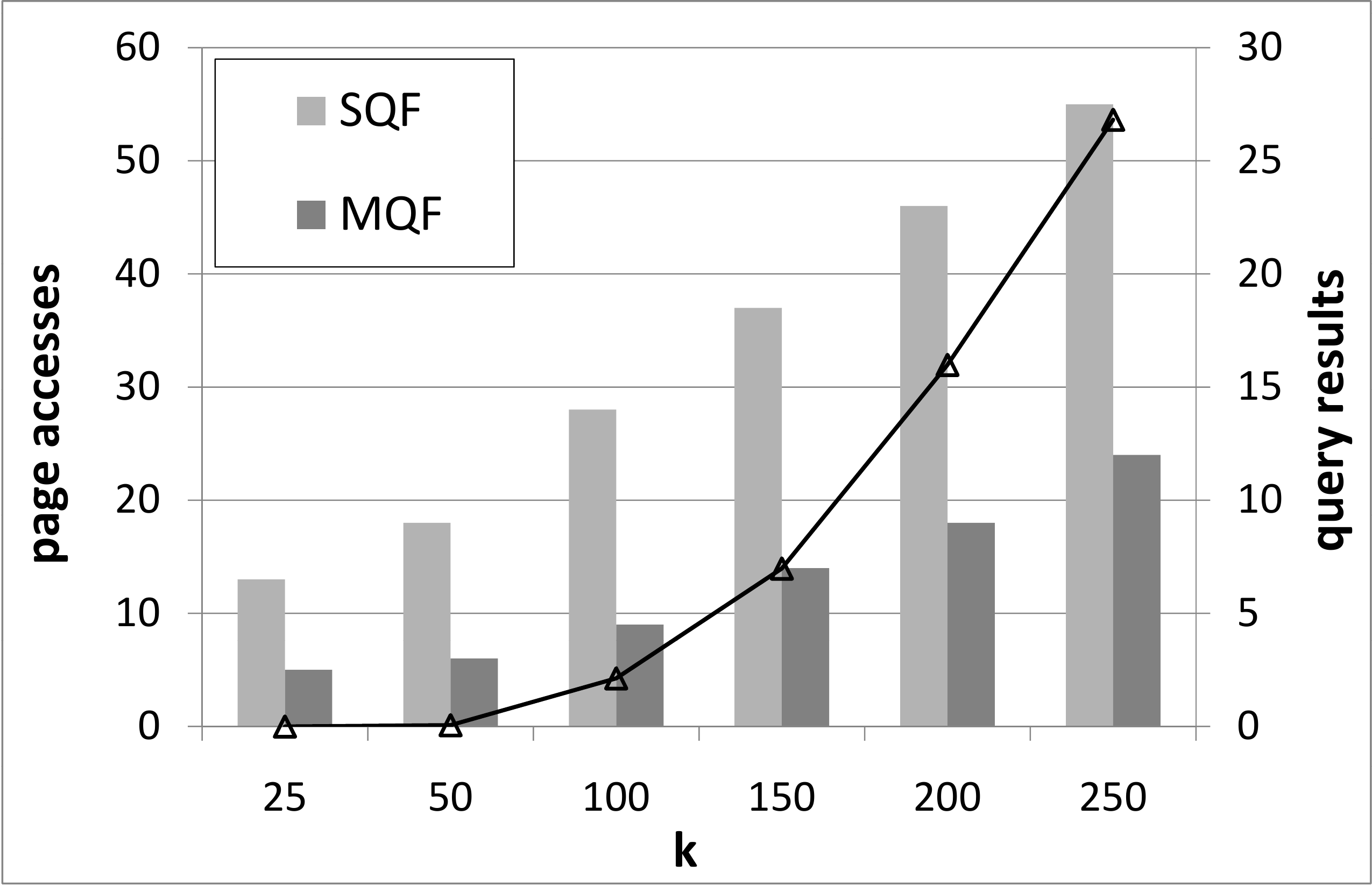}
    }
    \subfigure[I/O w.r.t. $|Q|$.]{
        \label{fig:knn-real-q}
        \includegraphics[width = 0.45\columnwidth]{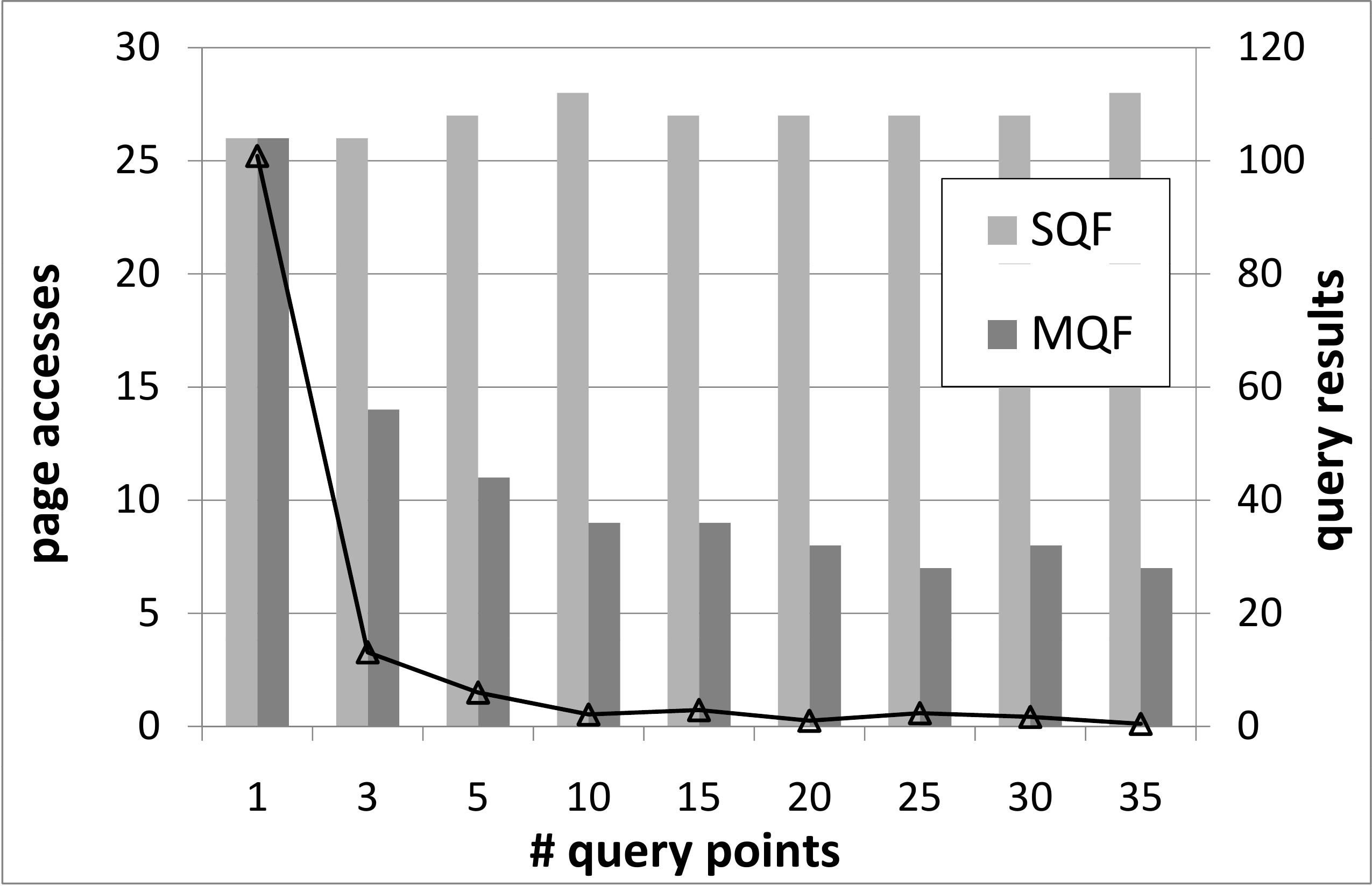}
    }
    \subfigure[I/O w.r.t. extent.]{
        \label{fig:knn-real-ext}
        \includegraphics[width =0.45\columnwidth]{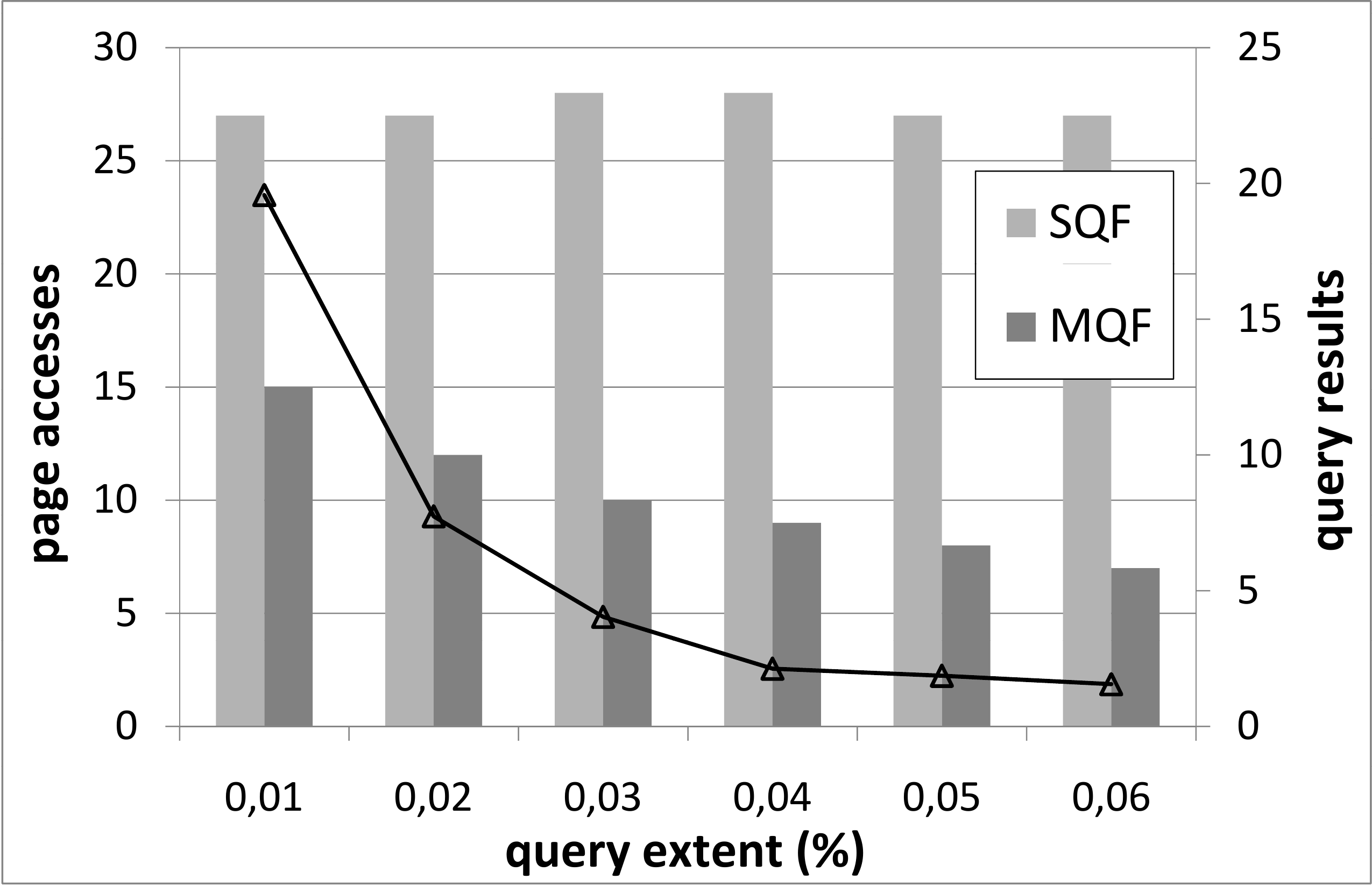} }
    \subfigure[I/O w.r.t. $|\DB|$.]{
        \label{fig:knn-real-scale}
        \includegraphics[width = 0.45\columnwidth]{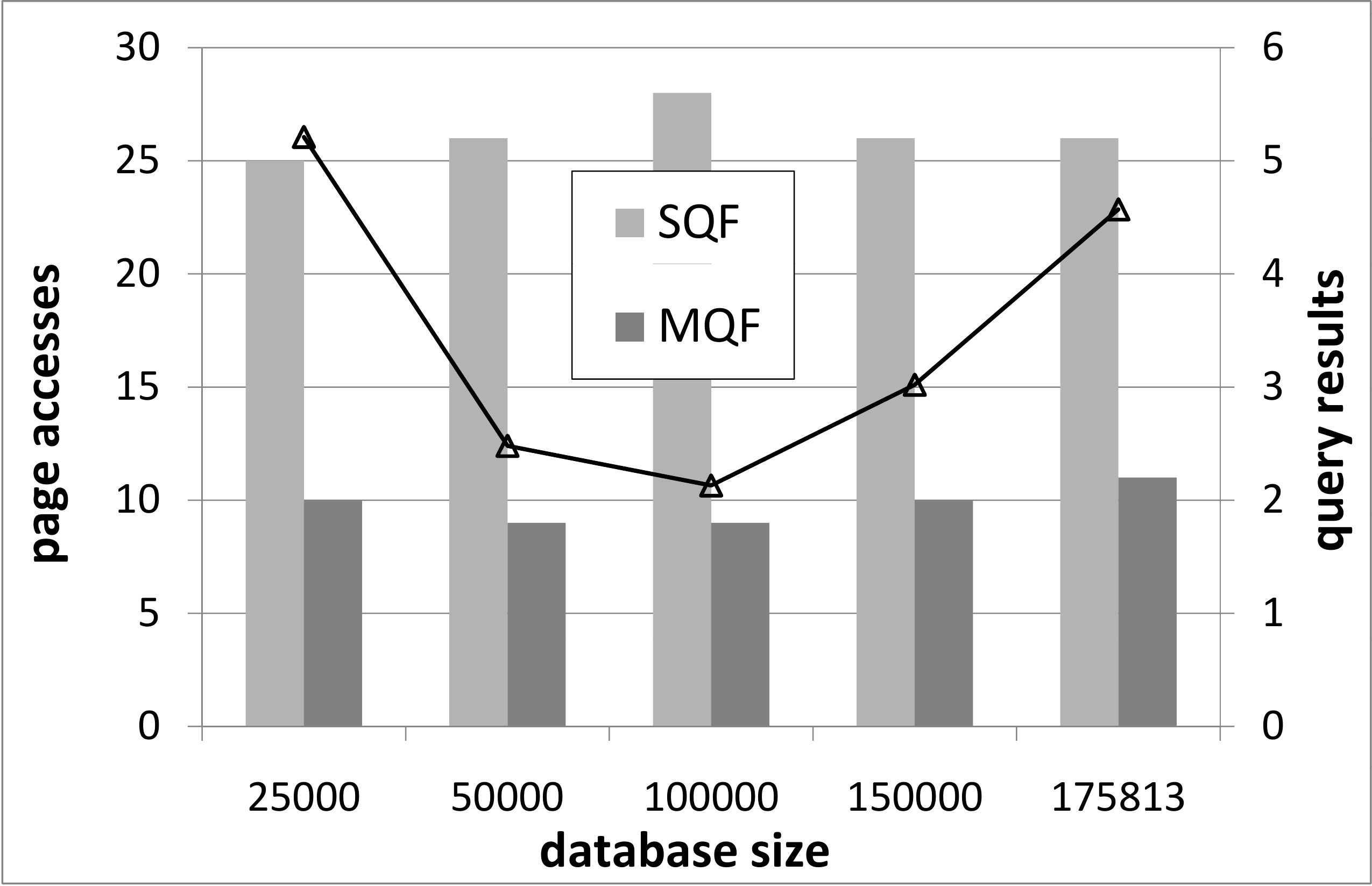}
    }
    \subfigure[CPU cost w.r.t. $|Q|$.]{
        \label{fig:knn-real-cpu}
        \includegraphics[width = 0.45\columnwidth]{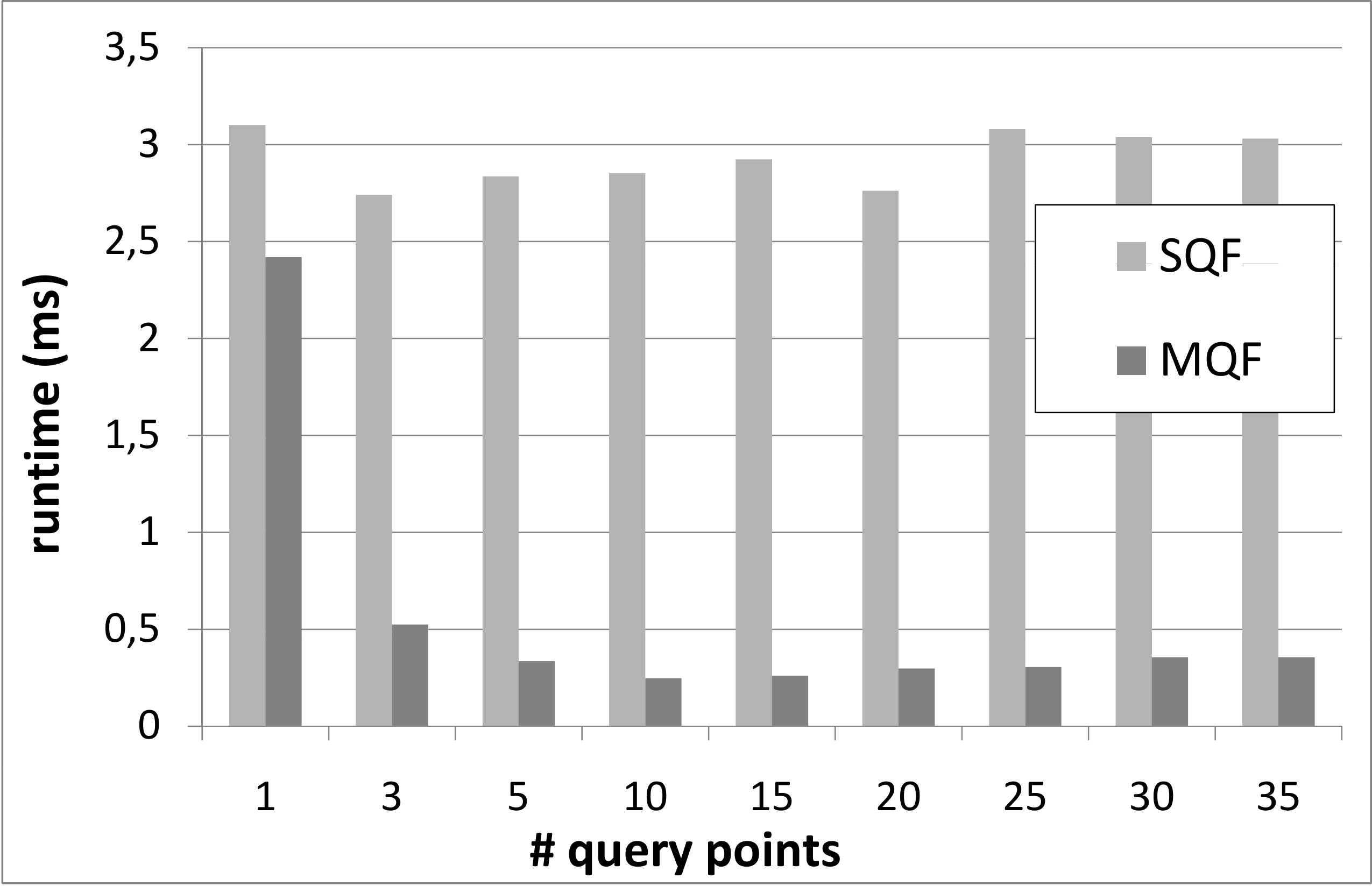}
    }
\vspace{-3mm}
    \caption{$Ik\mbox{-}NNQ$ algorithms on real dataset}
    \label{fig:knn-real}
\end{figure}

\end{document}